\documentclass[11pt,a4paper]{article}
\usepackage{times}

\usepackage{paralist} 

\usepackage{nicefrac}
\usepackage{comment}
\usepackage{amsfonts,amsthm,amsmath}
\usepackage{amssymb}
\usepackage[margin=2.5cm,nohead]{geometry}
\usepackage{thmtools,thm-restate}
\usepackage{xcolor}
\usepackage{mdframed}
\usepackage{xcolor,colortbl}
\usepackage{tcolorbox}

\definecolor{rvwvcq}{rgb}{0.08235294117647059,0.396078431372549,0.7529411764705882}
\definecolor{dtsfsf}{rgb}{0.8274509803921568,0.1843137254901961,0.1843137254901961}
\definecolor{mygreen}{rgb}{0.1803921568627451,0.49019607843137253,0.19607843137254902}

\usepackage{epsfig}

\usepackage[normalem]{ulem}

\usepackage{todonotes}
\usepackage{soul,xcolor}

\usepackage{algpseudocode}
\usepackage{xspace}

\usepackage{comment}
\excludecomment{supress}

\usepackage{mathrsfs}  
\usepackage{tabularx}
\usepackage{multirow}
\usepackage{booktabs}
\usepackage{array}

\usepackage[inline]{enumitem}

\usepackage{nicefrac}

\algrenewcommand\algorithmicrequire{\textbf{Input:}}

\newcolumntype{C}[1]{>{\centering\let\newline\\\arraybackslash}m{#1}}
\usepackage{csquotes}

\usepackage{amsopn}

\newcommand{\OO}{{\mathcal O}}

\newcommand{\npc}{\textsf{NP}-complete\xspace}
\newcommand{\nph}{\textsf{NP}-hard\xspace}
\newcommand{\no}{{\sf No}\xspace}
\newcommand{\yes}{{\sf Yes}\xspace}

\newtheorem{branching rule}{Branching Rule}
\newtheorem{reduction rule}{Reduction Rule}

\newcommand{\FPT}{\textsf{\textup{FPT}}\xspace}

\newcommand{\Co}[1]{\ensuremath{\mathcal{#1}}\xspace}

\newcommand{\sm}{\setminus}
\newcommand{\sse}{\subseteq}

\newcommand{\hide}[1]{}
\newcommand{\Red}[1]{{#1}}
\usepackage{cleveref}

\newcommand{\gm}{{\sc W-GM}\xspace}

\newcommand{\gmpvr}{{\sc GM}\xspace}

\newcommand{\targm}{{\sc Target W-GM}\xspace}

\newcommand{\rainm}{{\sc Rainbow Matching}\xspace}

\newcommand{\Ma}[1]{\textcolor{magenta}{#1}}

\newcommand{\il}[1]{\todo[inline, color={yellow!60}]{\small{#1}}}
\newcommand{\ma}[1]{\todo[color={green!70!yellow!30}]{\small{#1}}}

\newcommand{\setfamily}[4]{\ensuremath{\mathscr{#1}[#2,#3,#4]}\xspace}

\newcommand{\eLabel}[1]{\ensuremath{\langle #1 \rangle}\xspace}

\newcommand{\win}[1]{\ensuremath{{\rm win}(#1)}\xspace}

\newcommand{\Prop}[2]{\ensuremath{\mathscr{Prop}(#1,#2)}\xspace}

\newcommand{\Lset}{\ensuremath{\mathcal{L}}\xspace}

\usepackage{booktabs}
\newcommand{\polynf}{\ensuremath{\psi}}
\newcommand{\polynfpath}{\ensuremath{\psi'}}

\usepackage{cleveref}

\newtheorem{theorem}{Theorem}[section]
\newtheorem{lemma}{Lemma}[section]
\newtheorem{claim}{Claim}[lemma]
\newtheorem{corollary}{Corollary}[section]
\newtheorem{definition}{Definition}[section]
\newtheorem{observation}{Observation}[section]
\newtheorem{proposition}{Proposition}[section]

\newcommand{\shortversion}[1]{}

\usepackage{tcolorbox}

\makeatletter
\DeclareOldFontCommand{\rm}{\normalfont\rmfamily}{\mathrm}
\DeclareOldFontCommand{\sf}{\normalfont\sffamily}{\mathsf}
\DeclareOldFontCommand{\tt}{\normalfont\ttfamily}{\mathtt}
\DeclareOldFontCommand{\bf}{\normalfont\bfseries}{\mathbf}
\DeclareOldFontCommand{\it}{\normalfont\itshape}{\mathit}
\DeclareOldFontCommand{\sl}{\normalfont\slshape}{\@nomath\sl}
\DeclareOldFontCommand{\sc}{\normalfont\scshape}{\@nomath\sc}
\makeatother

\usepackage{todonotes}
\usepackage{tcolorbox}
\usepackage{framed}

\usepackage{flushend}

\title{Gerrymandering on graphs: Computational complexity and parameterized algorithms}

\author{Sushmita Gupta\thanks{Institute of Mathematical Science, HBNI, India.  \texttt{sushmitagupta@imsc.res.in}}
\and
Pallavi Jain\thanks{Indian Institute of Technology Jodhpur, India.   \texttt{pallavi@iitj.ac.in}}
\and
Fahad Panolan\thanks{Indian Institute of Technology Hyderabad, India.   \texttt{fahad@cse.iith.ac.in}}
\and
Sanjukta Roy\thanks{Algorithms and Complexity Group, TU Wien, Austria.\texttt{sanjukta.roy@tuwien.ac.at}}
\and 
Saket Saurabh\thanks{Institute of Mathematical Science, HBNI, India, and University of Bergen, Norway. \texttt{saket@imsc.res.in}}
}

\date{}

\begin{document}
\maketitle
\thispagestyle{empty}

\begin{abstract}
% !TEX root = main.tex

The practice of partitioning a region into areas to favor a particular candidate or a party in an election has been known to exist for the last two centuries. This practice is commonly known as {\em gerrymandering}. Recently, the problem has also attracted a lot of attention from complexity theory perspective. In particular, 
 Cohen-Zemach et al. [AAMAS 2018] proposed a graph theoretic version of gerrymandering problem and initiated an algorithmic study around this, which was continued by Ito et al. [AAMAS 2019]. In this paper we continue this line of investigation and resolve an open problem in the literature, as well as move the algorithmic frontier forward by studying this problem in the realm of parameterized complexity. 
 
%
%In this paper, we continue investigating the gerrymandering problem in the graph theoretic setting as proposed by Cohen-Zemach et al. [AAMAS 2018] and continued by Ito et al. [AAMAS 2019].  

Our contributions in this article are two-fold, conceptual and computational. We first resolve the open question posed by Ito et al. [AAMAS 2019] about the computational complexity of gerrymandering when the input graph is a path. Next, we propose a generalization of the model studied in [AAMAS 2019], where the input consists of a graph on $n$ vertices representing the set of voters, a set of $m$ candidates $\mathcal{C}$, a weight function $w_v: \mathcal{C}\rightarrow {\mathbb Z}^+$ for each voter $v\in V(G)$ representing the preference of the voter over the candidates, a distinguished candidate $p\in \mathcal{C}$, and a positive integer $k$. The objective is to decide if it is possible to partition the vertex set into $k$ {\it districts} (i.e., pairwise disjoint connected sets) such that the candidate $p$ {\em wins} more districts than any other candidate. There are several natural parameters associated with the problem: the number of districts the vertex set needs to 
be partitioned  ($k$), the number of voters ($n$), and the number of candidates ($m$). The problem is known to be \npc  even if $k=2$, $m=2$, and $G$ is either a complete bipartite graph (in fact $K_{2,n}$, a complete bipartite graphs with one side of size $2$ and the other of size $n$)  
or a complete graph. This hardness result implies that we cannot hope to have an algorithm with running time 
$(n+m)^{f(k,m)}$ let alone $f(k,m)(n+m)^{\OO(1)}$, where $f$ is a function depending only on $k$ and $m$, as this would imply that {\sf P=NP}.   This means that in search for \FPT algorithms we need to either focus on the parameter $n$, or subclasses of forest (as the problem is \npc on $K_{2,n}$, a family of graphs that can be transformed into a forest by deleting {\em one} vertex). Circumventing these intractable results, we successfully obtain the following algorithmic results.  

\begin{itemize}
\item We design a parameterized algorithm  with respect to the parameter $k$ (an algorithm with running time $2^{\OO(k)}n^{\OO(1)}$)  in both deterministic and randomized settings, even for arbitrary weight functions.   Whether the problem is \FPT parameterized by $k$ on trees remains an interesting open problem. 

\item  We show that the problem admits a $2^n (n+m)^{\mathcal{O}(1)}$ time algorithm on general graphs. 
\end{itemize}
Our algorithmic results use sophisticated  technical tools such as representative set family and Fast Fourier transform based polynomial multiplication, and their (possibly first) application to problems arising in social choice theory and/or algebraic game theory may be of independent interest to the community.

\end{abstract}

%\newpage
%\pagenumbering{arabic}

% Paper body
% !TEX root = main.tex

\section{Introduction}
``Elections have consequences'' a now-famous adage ascribed to Barack Obama, the former President of U.S.A, brings to sharp focus the high stakes of an electoral contest. %Manipulation can be defined as the strategic exploitation of a system or rules by an agent or group of agents with the goal to improve the final outcome from their perspectives. Manipulation classically refers to misreporting of true preferences by some or coalition of agents (voters in case of elections) but several other entities could well be in a position to affect the outcome. The central authority who is conducting the election comes to mind readily. This type of strategic manipulation falls within the purview of the expansive topic called {\it control} (or {\it agenda control}) that has been studied extensively ~\cite{Faliszewski2016ControlAB,ConitzerWalsh-Ch6}.  
 %\ma{Tighten this para}
Political elections, or decision making in a large organization, are often conducted in a hierarchical fashion. Thus, in order to win the final prize it is enough to manipulate at district/division level, obtain enough votes and have the effect propagate upwards to win finally. Needless to say the ramifications of winning and losing are extensive and possibly long-term; consequently, incentives for {\em manipulation} are rife.

% the big prize. \shortversion{Such districts/divisions exist apriorily and the rules of counting the votes are established beforehand.} 
%{\color{red}This type of manipulation is effective since changing the opinion of a large number of voters may be infeasible, and under favorable conditions unnecessary given that in most high stakes real-life situations a hierarchical system of voting (such as a primary followed by a general election at federal or state level) is employed which requires that a subset of candidates clear a threshold before an eventual winner is chosen from their midst. What that manipulation entails is situation specific and can run the gamut of adding bogus votes to deleting legitimate votes, adding candidates to removing candidates, changing the vote counting mechanism to redrawing the district lines.}\todo{remove from short version?}
%The last of these is the setting of this work. 

The objective of this article is to study a manipulation or control mechanism, whereby the manipulators are allowed to create the voting ``districts''.  A well-thought strategic division of the voting population may well result in a favored candidate's victory who may not win under normal circumstances. In a more extreme case, this may result in several favored candidates winning multiple seats, as is the case with election to the US House of Representatives, where candidates from various parties compete at the district level to be the elected representative of that district in Congress. This topic has received a lot of attention in recent years under the name of {\it gerrymandering}. A New York Times article ``How computers turned gerrymandering into science'' \cite{NYT-GM} discusses how Republicans were able to successfully win 65\% of the available seats in the state assembly of Wisconsin even though the state has about an equal number of Republican and Democrat voters. 
%\Ma{With the 2019 Supreme Court ruling against the authority of Federal courts for setting limits to political gerrymandering~\cite{WSJ19}, the opportunities as of now seem limitless and as a recent article in the New York Times pointed out ``How gerrymandering will protect Republicans who challenged the election''~\cite{NYT-GM21}, the partisan incentives may well outweigh the democratic cost.}\ma{Is this ok? Or too tendentious?}~ 
The possibility for gerrymandering and its consequences have long been known to exist and have been discussed for many decades in the domain of political science, as discussed by Erikson~\cite{Erikson72} and Issacharoff~\cite{Issacharoff02}. Its practical feasibility and long-ranging implications have become a topic of furious public, policy, and legal debate only somewhat recently \cite{NYT-GM2}, driven largely by the ubiquity of computer modelling in all aspects of the election process. Thus, it appears that via the vehicle of gerrymandering the political battle lines have been drawn to (re)draw the district lines.

%{\color{red}The ubiquity of computer-based modeling in all aspects of the election process, particularly in the quest to find potential voters, makes gerrymandering an enticing proposition for political operatives.}\todo{remove this line?}

%\ma{Mathematical abstraction} 

While gerrymandering has been studied in political sciences for long, it is only rather recently that the problem has attracted attention from the perspective of algorithm design and complexity theory.  Lewenberg et al.~\cite{LewenbergLR17} and Eiben et al. \cite{EFPS20} study gerrymandering in a geographical setting in which voters must vote in the closest polling stations and thus problem is about strategic placement of polling stations rather than drawing district lines. Cohen-Zemach et al.~\cite{ZemachLR18} modeled gerrymandering using graphs, where vertices represent voters and edges represent some connection (be it familial, professional, or some other kind), and studied the computational complexity of the problem. Ito et al.~\cite{ItoKKO19} further extended this study to various classes of graphs, such as paths, trees, complete bipartite graphs, and complete graphs. 

In both the papers the following hierarchical voting process is considered: A given set of voters is partitioned into several groups, and each of the groups holds an independent election. From each group, one candidate is elected as a nominee (using the plurality rule). Then, among the elected nominees, the winner is determined by a final voting  rule (again by plurality). The formal definition of the problem, termed {\sc Gerrymandering} (\gmpvr), considered in \cite{ItoKKO19} is as follows. The input consists of an undirected graph $G$, a set of candidates \Co{C}, an approval function $a:V(G) \rightarrow \Co{C}$ where $a(v)$ represents the candidate approved by $v$, a weight function $w \colon V(G) \rightarrow \mathbb{Z}^{+}$, a distinguished candidate $p$, and a positive integer $k$. We say a candidate $q$ wins a subset $V' \sse V(G)$ if $q\in \arg\max_{q'\in \Co{C}}\left\{\sum_{v\in V', \,a(v)=q'} w(v)\right\}$, i.e., the sum of the weights of voters in the subset $V'$ who approve $q$ is not less than that of any other candidate. The objective  
%of the problem 
is to decide whether there exists a partition of $V(G)$ into $k$ non-empty parts $V_1\uplus \ldots \uplus V_k$  (called {\it districts}) such that \begin{enumerate*}[label=(\roman*)] \item the induced subgraph $G[V_i]$ is connected for each $i\in \{1,\ldots,k\}$, and \item the number of districts {\em won only by $p$} is more than the districts won by any other candidate alone or with others. 
\end{enumerate*} 

\begin{tcolorbox}[colback=gray!5!white,colframe=gray!75!black]
In this paper we continue the line of investigation done in~\cite{ZemachLR18,ItoKKO19}. Our contribution is two fold, conceptual and the other is computational. Towards the former, we offer a realistic generalization of  \gmpvr, named {\sc Weighted Gerrymandering} (\gm). Towards the latter, we resolve an open question posed by Ito et al. \cite{ItoKKO19} pertaining to the complexity status of the \gmpvr\ on path graphs, when the number of candidates is not fixed. This reduction also shows that \gm is \npc on paths.  We complement this negative result by designing some {\it fixed parameter tractable} (\FPT) algorithms with respect to natural parameters associated with the problem.
 %Thus, moving the algorithmic frontier forward on the problem.
\end{tcolorbox}

 \shortversion{The authors considered the {\em plurality voting rule}, where a candidate who gets the majority votes is a nominee in the  first stage of the election, and a nominee who wins in the most groups is the final winner. A familiar setting which closely resembles this is that of the two tier electoral process of the U.S presidential election.} %\ma{Keep this.}

%\il{Our contribution: two fold--conceptual and computational. In the conceptual we offer this generalized model the computational contributions are that we resolve an open problem, and our algorithms also work for our generalised setting }
%
%\medskip
%\noindent {\bf Our contribution} is two fold, conceptual and the other is computational. Towards the former, we offer a realistic generalization of  \gmpvr, named {\sc Weighted Gerrymandering} (\gm). Towards the latter, we resolve an open question posed by Ito et al. \cite{ItoKKO19} pertaining to the complexity status of the \gmpvr\ on path graphs, when the number of candidates is not fixed. This reduction also shows that \gm is \npc in paths. We complement these negative results by designing some {\it fixed parameter tractable} (\FPT) algorithms and moderately exponential time algorithms. 

% Additionally, we provide {\it fixed parameter tractable} deterministic and randomized algorithms for \gm with respect to the parameter $k$ for path graphs. From the work of Ito et al.~\cite{ItoKKO19} it follows that for general graphs, we cannot expect to find a fixed parameter tractable algorithm with respect to $k$. Thus, we are able to identify a tractable class of graphs which while not polynomial time solvable does yield efficient algorithms when the number of districts is small. As a final result, we show that the problem admits a $2^n(n+m)^{\OO(1)}$ time algorithm on general graphs.

\medskip
\noindent 
{\bf Our model.} A natural generalization of \gmpvr in real-life is that of a vertex representing a locality or an electoral booth as opposed to an individual citizen. In that situation, however, it is only natural that more than one candidate receives votes in a voting booth, and the number of such votes may vary arbitrarily. We can model the number of votes each candidate gets in the voting booth corresponding to booth $v$ by a weight function $w_v:\Co{C} \rightarrow \mathbb{Z}^{+}$, i.e the value $w_{v}(c)$ for any candidate $c\in \Co{C}$ represents the number of votes obtained by candidate $c$ in booth $v$. This model is perhaps best exemplified by a nonpartisan ``blanket primary'' election (such as in California) where all candidates for the same elected post regardless of political parties, compete on the same ballot against each other all at once. In a two-tier system, multiple winners (possibly more than two) are declared and they contestant the general election. The idea that one can have multiple candidates earning votes from the same locality and possibly emerging as winners is captured by \gmpvr. 
In the other paper \cite{ZemachLR18,ItoKKO19}, the vertex $v$ ``prefers'' only one candidate, and in this sense our model  generlizes (\gm) theirs (\gmpvr). 
%\ma{Read this}

%\Ma{A natural realization of \gm in real-life is that a vertex represents a locality or an voting booth. Naturally, more than one candidate may receive votes in a voting booth, and the number of such votes may vary arbitrarily; the number of votes each candidate gets in the voting booth corresponding to $v$ is modeled by a weight function $w_v:\Co{C} \rightarrow \mathbb{Z}^{+} $.}\ma{Rewrite. first motivation then problem name}

%{\color{blue} There can be more than one candidate preferred by $v$ and each citizen of this subgroup or electoral booth might not prefer all these candidates but only a subset of them. To model this, unlike the \gmpvr, we allow each candidate preferred by $v$ to have different weights.} Certain parliamentary systems around the world have a weighted voting mechanism to determine the composition of their legislative body.

Formally stated, the input to \gm consists of an undirected graph $G$, a set of candidates \Co{C}, a weight function for each vertex $v\in V(G)$, $w_v:\Co{C} \rightarrow \mathbb{Z}^{+}$, a distinguished candidate $p$, and a positive integer $k$. A candidate $q$ is said to win a subset $V' \sse V(G)$ if $q\in \arg\max_{q'\in \Co{C}}\left\{\sum_{v\in V'} w_v(q')\right\}$. The objective is to decide whether there exists a partition of the vertex set $V(G)$ into $k$ districts such that \begin{enumerate*}[label=(\roman*)]
\item $G[V_i]$ is connected for each $i\in [k]$, and \item  the number of districts {\em won only by $p$} is more than the number of districts won by any other candidate alone or with others.\end{enumerate*} 
%\ma{Need this any more?}
\gmpvr can be formally shown to be a special case of \gm since we can transform an instance $\Co{I}=(G, \Co{C},a, w, p,k)$ of \gmpvr to an instance $\Co{J}=(G, \Co{C}, \{w_v\}_{v\in V(G)}, p,k)$ of \gm as follows. For each $v\in V(G)$, let $w_v:\Co{C} \rightarrow {\mathbb Z}^+$ such that for any $q\in \Co{C}$, if $a(v)=q$, then $w_v(q)=w(v)$ and $w_v(q)=0$, otherwise.% It follows that $I$ is a \yes-instance of \gmpvr\ iff $J$ is a \yes-instance of \gm. 

% Let $(G, \Co{C},c:V(G) \rightarrow \Co{C},w: V(G)\rightarrow \mathbb{Z}^{+}, p,k)$ be an instance of \gmpvr. Then, for each $v\in V(G)$, let $w_v$ be the function from $\Co{C}$ to ${\mathbb Z}^+$, defined as follows. For any $q\in \Co{C}$, if $c(v)=q$, then $w_v(q)=w(v)$ and $w_v(q)=0$, otherwise. It is easy to see that $(G,\Co{C}, c:V(G) \rightarrow \Co{C}, w: V(G)\rightarrow \mathbb{Z}^{+}, p,k)$ is a \yes-instance of \gmpvr\ if and only if $(G, \Co{C}, \{w_v: \Co{C}\rightarrow {\mathbb Z}^{+}\}_{v\in V(G)}, p,k)$ is a \yes-instance of \gm. 

% !TEX root = main.tex

%\medskip

\smallskip

\noindent
{\bf Our results and methods.}\label{para:our contribution}   %\gmpvr is the problem of redistricting and thus it is natural to study the problem on planar graphs \cite{FleinerNT17}. A path is a restricted and simple version of planar graph. 
The main open problem mentioned in Ito et. al \cite{ItoKKO19} is  the complexity status of \gmpvr\ on paths when the number of candidates is not  fixed (for the fixed number of candidates, it is solvable in polynomial time). We begin with answering their question and show that the problem is intractable even for such simple structures.

%\todo[inline]{Say somethinng from the other paper and that planar and then this..Add something to say why we study the problem on paths?} The following result answers that question. 

\begin{theorem}\label{thm:npcpath}
\gmpvr\ is {\rm \npc} on paths. %\ma{Modify this. (DONE)}
\end{theorem}

To prove Theorem~\ref{thm:npcpath}, we give a polynomial-time reduction from \rainm on paths to \gmpvr on paths. In the \rainm problem, given a graph $G$, a coloring function on edges, $\psi\colon E(G)\rightarrow \{1, \ldots,\ell\}$, and an integer $k$; the objective is to decide if there exists a $k$-sized subset of edges that are vertex disjoint, (called a {\it matching}), such that for every pair of edges $e$ and $e'$ in the set, we have $\psi(e)\neq \psi(e')$. It is known that \rainm is \npc even when the input graph is a path~\cite{pfender2014complexity}.\label{rainm:defn}

\shortversion{We prove Theorem~\ref{thm:npcpath}, by showing a polynomial-time many-to-one reduction from \rainm  on paths to \gmpvr\ on paths. In the \rainm problem, we are given a graph $G$, a coloring function on edges, $\psi\colon E(G)\rightarrow \{1, \ldots,\ell\}$, and an integer $k$; the objective is to decide if there exists a $k$-sized subset of edges that are vertex disjoint, (called a {\it matching}), such that for every pair of edges $e$ and $e'$ in the set, we have $\psi(e)\neq \psi(e')$. It is known that \rainm is \npc even when the input graph is a path~\cite{pfender2014complexity}.}

%Next, we study the problem from the viewpoint of parameterized complexity. 

 %Due to paucity of space, the entire proof of Theorem~\ref{theorem:exact} has been presented in Appendix~\ref{sec:exact-algo}.

Next, we study the problem from the viewpoint of parameterized complexity.  The goal of parameterized complexity is to find ways of solving
\nph problems more efficiently than brute force: here aim is to
restrict the combinatorial explosion to a parameter that is hopefully
much smaller than the input size. Formally, a {\em parameterization}
of a problem is assigning an integer $\ell$ to each input instance and we
say that a parameterized problem is {\em fixed-parameter tractable
  (\FPT)} if there is an algorithm that solves the problem in time
$f(\ell)\cdot |I|^{O(1)}$, where $|I|$ is the size of the input and $f$ is an
arbitrary computable function depending on the parameter $\ell$
only. There is a long list of \nph problems that are \FPT under
various parameterizations: finding a vertex cover of size $\ell$, finding
a cycle of length $\ell$, finding a maximum independent set in a graph
of treewidth at most $\ell$, etc.
For more background, the reader is referred to the monographs \cite{ParamAlgorithms15b,DowneyFbook13,niedermeier06b}.

\smallskip

\noindent
{\bf Our choice of parameters.}  
There are several natural parameters associated with the gerrymandering problem: the number of districts the vertex set needs to  be partitioned  ($k$), the number of voters ($n$), and the number of candidates ($m$).  Ito et al.~\cite{ItoKKO19} proved that \gmpvr\ is \npc  even if $k=2$, $m=2$, and $G$ is either a complete bipartite graph (in fact $K_{2,n}$) or a complete graph. Thus, we cannot hope for an algorithm for \gm that runs in $f(k,m)\cdot n^{\OO(1)}$ time, i.e.,  an \FPT algorithm with respect to the parameter $k+m$, even on planar graphs. In fact, we cannot hope to have an algorithm with running time 
$(n+m)^{f(k,m)}$, where $f$ is a function depending only on $k$ and $m$, as that would imply  {\sf P=NP}.   This means that our search for \FPT algorithms needs to either focus on the parameter $n$, or subclasses of planar graphs (as the problem is \npc on $K_{2,n}$, which is planar). Furthermore, note that $K_{2,n}$ could be transformed into a forest by deleting a vertex, and thus we cannot even hope to have  an algorithm with running time $(n+m)^{f(k,m)}$, where $f$ is a function depending only on $k$ and $m$, on a family of graphs that can be made acyclic by 
{\em deleting  at most one vertex}. This essentially  implies that if we wish to design an \FPT algorithm for \gm with respect to the parameter $k$, or $m$, or $k+m$, we must restrict input graphs to forests.   Circumventing these intractable results, we successfully obtain 
several algorithmic results. 
%the following algorithmic results.
%
%
%
%Ito et al.~\cite{ItoKKO19} proved that \gmpvr\ is \npc even if $k=2$. 
%
%Thus, we cannot hope for an algorithm for \gm that runs in $f(k)\cdot n^{\OO(1)}$ time, i.e.,  an \FPT algorithm, with respect to the parameter $k$, on general graphs. 
We give deterministic and  randomised \FPT algorithms for \gm on paths with respect to the parameter $k$. Since \gm is a generalization of \gmpvr, the algorithmic results for the former hold for the latter as well. 
%There exists a simple $n^{\OO(k)}$ time algorithm for the problem on trees~\cite{ItoKKO19}, but we could not design an \FPT algorithm this on trees. Thus, whether the problem is \FPT parameterized by $k$ on trees remains an interesting open problem. 

%In what follows, 
%{\em $n$  and $m$ denote the number of vertices in the input graph and the number of candidates, respectively.}

%\todo{remove this OR move to section?}{\color{red}Both of our algorithms work by reducing the \gm problem  on paths to detect an appropriately labeled $(k+1)$-length path in an auxiliary directed graph. The vertices on this path (except the source and target vertices) define the $k$ districts needed for the solution, and the arc labels allow us to keep track of the number of districts won by candidates other than $p$, the distinguished candidate.} 
%In what follows, 
%{\em $n$  and $m$ denote the number of vertices in the input graph and the number of candidates, respectively.}
%In particular, we prove the following results, {\it where $n$  and $m$ denote the number of vertices in the input graph and the number of candidates, respectively.}

%\ma{Added this}

\smallskip

\noindent{\it Unique winner vs Multiple winner:}~Note that the definition of \gmpvr\ by Ito et. al \cite{ItoKKO19} or its generalization \gm put forward by us does not preclude the possibility of multiple winners in a district, only that $p$ wins more number of districts alone than any other candidate {\it alone or in conjunction with others.} The time complexity stated in Theorems~\ref{thm:detfpt} and \ref{thm:ranfpt} are achieved when only one winner emerges from each district, a condition that is attainable using a tie-breaking rule. Formally stated, for an instance $(G, \Co{C}, \{w_v\}_{v\in V(G)}, p,k)$,  we consider a tie-breaking rule $\eta$ such that for any district $U \sse V(G)$, $\eta$ declares a candidate in the set~${\rm argmax}_{q\in \Co{C}}\{\sum_{v\in U} w_{v}(q)\}$, as the winner of the district $U$. Our results hold for any tie-breaking rule, as long as it is applied uniformly whenever necessary. Notably, the following algorithms can be modified to handle the case when multiple winners emerge in some district(s). % (discussed in Appendix~\ref{sec:multiwinner}). 

%\il{TIE BREAKING: For our algorithms, we assume that {\em every district is won by a unique candidate who attains the maximum weighted vote in the district.} More precisely, for an instance $(G, \Co{C}, \{w_v\colon \Co{C}\rightarrow {\mathbb Z}^{+}\}_{v\in V(G)}, p,k)$,  we consider a tie-breaking rule $\eta$ such that for any district (i.e connected subset) $U \sse V(G)$, $\eta$ declares a candidate in the set ${\rm argmax}_{q\in \Co{C}}\{\sum_{v\in U} w_{v}(q)\}$, as the winner of the district $U$. Our result holds for any tie-breaking rule, as long as it is applied uniformly whenever necessary. }

\begin{theorem}
\label{thm:detfpt}
There is a deterministic algorithm that given an instance of \gm on paths and a tie-breaking rule $\eta$ solves the instance in time $2.619^k (n+m)^{\OO(1)}$. 

%where $n$  and $m$ denote the number of vertices in the input graph and the number of candidates, respectively.

\end{theorem}

\shortversion{\il{Need to formulate this properly... how do we quantify $t$, is it a constraint on the feasible solution. Then should we define a new problem ? }
\begin{theorem}(\Ma{TRIAL-Multiple winner})\label{thm:detfpt-multiple-winners}
There is a deterministic algorithm that given an instance of \gm on paths that solves the instance in time $2.619^{kt} (n+m)^{\OO(1)}$, 
where $t$, $n$ and $m$ denote the maximum number of winners in a district, number of vertices in the input graph and the number of candidates, respectively.
\end{theorem}

The deterministic algorithm uses the concept of representative families  to detect the existence of a desired path in the auxiliary graph. The analysis is somewhat similar to that of $k$-{\sc Path}, the problem of deciding if there exists a $k$-length path in the given graph~\cite{fomin2016efficient}. Using the same auxiliary graph, we are able to obtain an improved running time in the randomized setting as exhibited by the following result.
}

\begin{theorem}
\label{thm:ranfpt}%{\footnote{Results marked by $\clubsuit$ can be found in the Supplementary File.}}
There is a randomized algorithm that given an instance of \gm on paths and a tie-breaking rule $\eta$, 
solves the instance in time $2^k (n+m)^{\OO(1)}$ with no false positives and false negatives with probability at most $1/3$.

%where $n$  and $m$ denote the number of vertices in the input graph and the number of candidates, respectively.
\end{theorem}

\noindent
{\bf Intuition behind the proofs of \Cref{thm:detfpt} and \ref{thm:ranfpt}.} Since, the problem is on paths, it boils down to selecting $k$ appropriate vertices such that the subpaths between them form the desired districts. This in turn implies 
that each district can be identified by the leftmost vertex and the rightmost vertex appearing in the district (based on the way vertices appear on the path). Hence, there can be at most  $\OO(n^2)$ districts in the path graph. Furthermore, since we are on a path, we observe that if we know a district (identified by its leftmost and the rightmost vertices on the path), then we also know the leftmost (and rightmost) vertex of the district adjacent to it. These observations naturally lead us to consider the following graph $H$: we have a vertex for each possible district and put an edge from a district to another district, if these two districts appear consecutively on the path graph. Thus, we are looking for a path of length $k$ in $H$ such that (a) it covers all the vertices of the input path (this automatically implies that each vertex appears in exactly one district); and (b) the 
distinguished candidate wins most number of districts. This equivalence allows us to use the rich algorithmic toolkit developed for designing $2^{\OO(k)} n^{\OO(1)}$ time algorithm for finding  a $k$-length path in a given graph~\cite{monien1985find,bjorklund2017narrow,Williams09}.

\shortversion{
\begin{theorem}\label{thm:ranfpt-multiple-winners}(\Ma{Multiple winners}) There is a randomized algorithm that given an instance of \gm on paths solves the instance in time $2^{kt} (n+m)^{\OO(1)}$ with no false positives and false negatives with probability at most $1/3$, 
where $t$, $n$ and $m$ denote the maximum number of winners in a district, number of vertices in the input graph and the number of candidates, respectively.
\end{theorem}
}

\shortversion{The randomized algorithm works by detecting the existence of the desired path in the auxiliary graph by interpreting each of the labeled paths in the graph as a multivariate monomial, and then using a result by Williams~\cite{Williams09} to detect a multilinear monomial in the resulting (multivariate) polynomial. The underlying idea being that each path with the desired properties is a multilinear monomial in the polynomial thus constructed, and vice-versa. Williams~\cite{Williams09} gave an algorithm with one sided error that allows us to detect a multilinear monomial in time $\OO^{\star}(2^{d})$, where $d$ denotes the degree of the multivariate polynomial\footnote{$\OO^{\star}()$ hides factors that are polynomial in the input size.}. }

The above tractability result for paths cannot be extended to graphs with pathwidth $2$, or graphs with feedback vertex set (a subset of vertices whose deletion transforms the graph into a forest) size $1$, because \gmpvr\ is \npc on $K_{2,n}$ when $k=2$ and $|\Co{C}|=2$~(see~\cite{ItoKKO19}). Note that the pathwidth of graph $K_{2,n}$ is $2$ and has feedback vertex set size $1$.  
For trees, it is easy to obtain a $\OO(\binom{n}{k-1})$ time algorithm by ``guessing'' the  $k-1$ edges whose deletion yields the $k$ districts that constitute the solution.  However, a $f(k)n^{\OO(1)}$ algorithm for trees so far eludes us. Thus, whether the problem is \FPT parameterized by $k$ on trees remains an interesting open problem.   Finally, we consider the parameter $n$, the number of voters ($n$) and design the following algorithm for \gm parameterized by $n$. 
%We conclude our algorithmic discussion with an exact exponential-time algorithm for \gm. 
%\ma{moved here}

%It is worth pointing out that the algorithms in Theorems~\ref{thm:detfpt} and \ref{thm:ranfpt} can be adapted to the case when multiple candidates win a single district. For ease of exposition, we chose to present the result where only a single candidate may win a district. \ma{Explained elsewhere}

%\il{Moved here to the end....}

%It is noteworthy that while the value of $n$ 

\begin{theorem}
\label{theorem:exact}
There is an algorithm that given an instance of \gm\  on arbitrary graphs and a tie-breaking rule $\eta$, solves the instance in time $2^n (n+m)^{\OO(1)}$. %, where $n$  and $m$ denote the number of vertices in the input graph and the number of candidates, respectively. 
\end{theorem}

\noindent 
{\bf Intuition behind the proof of \Cref{theorem:exact}.} Suppose that we are given a \yes-instance of the problem. Of the $k$ possibilities, we first ``guess'' in a solution the number of districts that are won by the distinguished candidate $p$. Let this number be denoted by $k^\star$. Next, for every candidate $c \in \Co{C}$, we consider the family $\Co{F}_c$, the set of districts  of $V(G)$ in which  $c$ wins in each of them. These families are pairwise disjoint because each
district has a unique winner. Our goal is to find $k^\star$ disjoint sets from the family $\Co{F}_p$ and at most $k^\star-1$ disjoint sets from any other family so that in total we obtain $k$ pairwise disjoint districts that partition $V(G)$. The exhaustive algorithm to find the districts from these families would take time $\OO^\star(2^{nmk^\star})$. We reduce our problem to polynomial multiplication involving polynomial-many multiplicands, each with degree at most $\OO(2^n)$. Next, we discuss the purpose of using polynomial algebra.

\smallskip
\noindent
{\bf Why use polynomial algebra?}\label{why-poly-algebra} Every district $S$ is a subset of $V(G)$. Let $\chi(S)$ denotes the characteristic vector corresponding to $S$. We view $\chi(S)$ as an $n$ digit  binary number, in particular, if $u_i \in S$, then $i^\text{th}$ bit of $\chi(S)$ is $1$, otherwise $0$. A crucial observation guiding our algorithm is that two sets $S_1$ and $S_2$ are disjoint if and only if the number of $1$ in $\chi(S_1)+ \chi(S_2)$ (binary sum/modulo $2$) is equal to $|S_1|+|S_2|$. So, for each set $\Co{F}_c$, we make a polynomial $P_c(y)$, where for each set $S\in \Co{F}_c$, there is a monomial 
$y^{\chi(S)}$.  
Let $c_1$ and $c_2$ be two candidates, and for simplicity assume that each set in 
$\Co{F}_{c_1}$ has size exactly $s$ and each set in $\Co{F}_{c_2}$ has size  exactly $t$. Let $P^\star(y)$ be the polynomial obtained by multiplying  $P_{c_1}(y)$ and $P_{c_2}(y)$; and let $y^z$ be a monomial of 
 $P^\star(y)$. Then, the $z$  has exactly $s+t$ ones if and only if ``the sets which corresponds to $z$ are disjoint''. Thus, the polynomial method allows us to capture disjointness and hence, by multiplying appropriate subparts of polynomial described above, we obtain our result. Furthermore, note that  $\chi(S) \in \{0,1\}^n$, throughout the process as they correspond to some set in $V(G)$, and hence the decimal representation of the maximum degree of the considered polynomials is upper bounded by $2^n$. Hence, the algorithm itself is about applying  an $\OO(d \log d)$ algorithm to multiply two polynomials of degree $d$;  here $d\leq 2^n$.  Thus, we obtain an algorithm that runs in time $2^n (n+m)^{\OO(1)}$.

 %\il{\Hi{TO DO:}  Omit discussing the technique here (we are doing it in a separate section), only emphasize that keeping track of the districts won by each candidate makes it different from a simple application of subsgmet convolution }

Additionally, using our parameterized algorithms (Theorems~\ref{thm:detfpt} and~\ref{thm:ranfpt}), we can 
 improve over Theorem~\ref{theorem:exact}, when the graph is a path. That is, using Theorems~\ref{thm:detfpt}, \ref{thm:ranfpt}, and the fact that there exists an algorithm for paths that runs in time $\OO(\binom{n}{k-1})$, we obtain that for \gm\ on paths, there exists a deterministic algorithm that runs in $\max_{1\leq k \leq n}\min \{\binom{n}{k}, 2.619^k\}$ time, and a randomized  algorithm that runs in $\max_{1\leq k \leq n}\min \{\binom{n}{k}, 2^k\}$ time. Using, standard calculations 
 we can obtain the following result.

%  which is maximised at $k=0.773n$. Thus, we obtain the following result.
% % which is maximised at $k=0.6634n$, 
%\begin{theorem}
%\label{thm:exact-path}
%There is a deterministic algorithm that given an instance of \gm on paths and a tie-breaking rule $\eta$, solves the instance in time $1.894^n (n+m)^{\OO(1)}$, where, $n$  and $m$ denote the number of vertices in the input graph and the number of candidates, respectively. %where $n$ denote the number of vertices in the input graph. 
%%\todo{I think we should also have $m$ as polynomial factor in the running time as we also compute the winner in every district}
%\end{theorem}

\begin{theorem}
There is a (randomized) deterministic algorithm that given an instance of \gm on paths and a tie-breaking rule $\eta$, solves the instance in time ($1.708^n (n+m)^{\OO(1)}$) $1.894^n (n+m)^{\OO(1)}$. 
%, where, $n$  and $m$ denote the number of vertices in the input graph and the number of candidates, respectively. 
\end{theorem}

It is worth mentioning that our algorithmic results use sophisticated technical tools from parameterized complexity--representative set family and Fast Fourier transform based polynomial multiplication--that have yielded breakthroughs in improving time complexity of many well-known optimization problems. Thus, their (possibly first) application to problems arising in social choice theory and/or algebraic game theory may be of independent interest to the community.
%} \ma{If not in abstract, then here? Or both?}

\smallskip
\noindent{\bf Organization of the paper.} In Section~\ref{section:npc}, we prove \Cref{thm:npcpath}. Section~\ref{sec:gm_path} and \ref{sec:Exact Algorithm} are devoted to \FPT algorithms. Section~\ref{sec:conclusion} concludes the paper with some open questions. %    is devoted We start each of our technical sections with a brief summary of the intuition behind the result. Additionally, for the algorithmic sections, where we have used technical tools such as {\it representative family} and {\it polynomial algebra}, we have also explained the underlying connection of our problem to those techniques under "Why use representative family ?" pp~\pageref{why-rep-family}, "Why use polynomial algebra ?'' pp~\pageref{why-poly-algebra}, respectively. 

\medskip 
\noindent{\it Related work.} In addition to the result discussed earlier Ito et al.~\cite{ItoKKO19} also prove that \gmpvr\ is strongly \npc\ when $G$ is a tree of diameter four; thereby, implying that the problem cannot be solved in pseudo-polynomial time unless {\sf P} = {\sf NP}. As \gmpvr\ is a special case of \gm, each of the hardness results for \gmpvr\ carry onto \gm.  They also exhibit several positive results: \gmpvr\ is solvable in polynomial time on stars (i.e., trees of diameter two) and that the problem can be solved in polynomial time on trees when $k$ is a constant.  Moreover, when the number of candidates is a constant, then it is solvable in polynomial time on paths and is solvable in pseudo-polynomial time on trees. The running time of the algorithm on paths is $k^{2^{\vert \Co{C}\vert}} n^{\OO(1)}$, where $n$ is the number of vertices in the input graph and $\Co{C}$ is the set of the candidates. Prior to this Cohen-Zemach et al.~\cite{ZemachLR18} studied \gmpvr on graphs. \shortversion{proved that given an instance of \gmpvr, it is \npc\ to decide if there is a partition of the input graph such that each part contains at least two vertices and the distinguished candidate $p$ wins in at least $r$ parts, for a given positive integer $r$. } On the other hand, Brubach et al.~\cite{brubach2020meddling} study strategyproofness in partisan gerrymandering and the effects of banning outlier.
In addition to the papers discussed earlier, there are far too many articles to list on the subject of strategic manipulation in voting as well as on the subject of gerrymandering. %Due to space constraints we can only present a brief snapshot of the landscape which is far from comprehensive.
Some of them are \cite{PUPPE200993,FleinerNT17,Clough07,Talmon18,UnweightedCoalitionalManipulation09,Zuckerman09a,Dey-GM-bribery}. Due to space constraints we do not discuss them here. 
Parameterized complexity of manipulation has received extensive attention over the last several years,~\cite{BetzlerGN10, BETZLER20095425, FaliszewskiHHR08,FaliszewskiHHR09,DeyMN19}  are just a few examples.

\section{Preliminaries}

%\il{Move this to Our Contribution \& only leave notations and basic definitions here.}

%\paragraph*{{\bf Target gerrymandering.}} 
To prove our algorithmic result we prove the following variant of \gm\ that we call {\sc Target Weighted Gerrymandering} (in short, \targm). The input of \targm\ is an instance of \gm, and a positive integer $k^{\star}$. The objective is to test whether the vertex set of the input graph can be partitioned into $k$ districts such that the candidate $p$ wins in $k^{\star}$ districts and no other candidate wins in more than $k^{\star}-1$ districts.  The following simple lemma implies that to design an efficient algorithm for \gm\ it is enough to design an efficient algorithm for \targm.  

\begin{lemma}\label{lem:tgm-gm}
If  there exists an algorithm that given an instance $(G,\Co{C},\{w_v\}_{v\in V(G)},p,k,k^\star)$ of \targm and a tie-breaking rule $\eta$, solves the instance in $f(z)$ time, then there exists an algorithm that solves the instance $(G,\Co{C},\{w_v\}_{v\in V(G)},p,k)$ of  \gm in $f(z)\cdot k$ time under the tie-breaking rule $\eta$.
\end{lemma}
%\il{}

\noindent{\bf Notations and basic terminology.}  
%Throughout the paper, we denote the number of candidates by $m$ and the number of vertices in the graph $G$ by $n$.  
%\shortversion{For a (un)directed graph $G$, we denote the vertex set and the (edge) arc set of $G$ by $V(G)$ and $E(G)$ and $A(G)$, respectively.\ma{Need this?}
%} 
%
%%\noindent
%%{\bf Graphs.} 
In an undirected graph, we denote an edge between the vertices  $u$ and $v$ as $uv$, and $u$ and $v$ are called the {\it endpoints} of  $uv$. Let $G=(V,E)$ be an undirected graph. A graph $G$ is said to be {\em connected} if every two vertices of $G$ are connected to each other by a path in $G$. For a set $X \subseteq V(G)$, $G[X]$ denote the graph induced on $X$, that is, $G[X]$ contains all the vertices in $X$ and all the edges in $G$ whose both the endpoints are in $X$.  We say that $X$ is a connected set if $G[X]$ is a connected graph. A {\em connected component} of a graph $G$ is a maximally connected subgraph of $G$. In a directed graph $G=(V,A)$, we denote an arc (i.e., directed edge) from $u$ to $v$ by $\langle u,v \rangle$, and say that $u$ is an in-neighbor of $v$ and $v$ is an out-neighbor of $u$. For $x\in V(G)$, $N^-(x)$ denote the set of all in-neighbors of $x$, that is, $N^-(x)=\{y\in V(G)\colon \langle y,x \rangle \in A(G)\}$. The in-degree (out-degree)  of a vertex $x$ in $G$ is the number of in-neighbors (out-neighbors) of $x$ in $G$. For basic notations of graph theory we refer the reader to ~\cite{DBLP:books/daglib/0030488}. %For any $n\in \mathbb{Z}^{+}$, we denote the set $\{1,2,\ldots, n\}$ as $[n]$. 
For a function $\psi \colon A \rightarrow B$, $\psi(A)=\{\psi(a)\in B \colon a \in A\}$.

\hide{\shortversion{A {\em matching} in $G$ is a set $Y\subseteq E(G)$ such that no two distinct edges in $Y$ have a common vertex. A {\em path} $G$ is a graph whose vertices can be listed in the order $(v_1,\ldots,v_n)$ such the edge set is $\{v_iv_{i+1} \colon 1\leq i \leq n-1\}$. For the path graph $G=(v_1,\ldots,v_n)$, its {\em subpath} from $v_i$ to $v_j$, where $j\geq i$, is a path with vertex set $\{v_i,\ldots,v_j\}$ and edge set  $\{v_\ell v_{\ell+1} \colon i\leq \ell \leq j-1\}$.}}

%\shortversion{For a digraph $G$ (with parallel arcs), a path is a sequence of vertices and edges denoted by $(v_1,e_1,v_2,e_2,\ldots e_{\ell-1},v_{\ell})$, where $\ell\in {\mathbb N}$ such that $v_1,\ldots,v_{\ell}$ are distinct vertices, $e_1,\ldots,e_{\ell-1}$ are distinct arcs, and for each $i\in \{1,\ldots,\ell-1\}$, $e_i$ is an arc from $v_i$ to $v_{i+1}$.%\ma{Moved to App, where it is used.}}

%\noindent{\bf Sets and Function.} 

%\par
%
%\noindent{\bf Preliminaries for Section~\ref{sec:exact-algo}:} %In this paragraph, we define the notations and terminologies used in Section~\ref{sec:exact-algo}. 

% !TEX root = gm-main.tex

%\newcommand{\Red}[1]{\textcolor{red}{{\bf #1}}}
%\newcommand{\Red}[1]{{#1}}
\definecolor{zzttff}{rgb}{0.6,0.2,1}
\newcommand{\I}{\ensuremath{\mathcal{I}}}
\newcommand{\J}{\ensuremath{\mathcal{J}}}

\section{NP-Completeness}\label{section:npc}

We prove Theorem~\ref{thm:npcpath} here, by giving a polynomial-time reduction from \rainm on paths to \gmpvr on paths. Recall the definition of \rainm from  pp. \pageref{rainm:defn}.
For a function $\psi \colon A \rightarrow B$, we define $\psi(A)=\{\psi(a)\in B \colon a \in A\}$. 
 % In the \rainm problem, we are given a graph $G$, a coloring function on edges, $\psi\colon E(G)\rightarrow \{1, \ldots,\ell\}$, and an integer $k$; the objective is to decide if there exists a $k$-sized subset of edges that are vertex disjoint, (called a {\it matching}), such that for every pair of edges $e$ and $e'$ in the set, we have $\psi(e)\neq \psi(e')$. It is known that \rainm is \npc even when the input graph is a path~\cite{pfender2014complexity}.%Recall the definition of \rainm from page~\pageref{para:our contribution}.
\shortversion{The input to \rainm consists of a graph $G$, an edge color function $\psi \colon E(G)\rightarrow \{1,\ldots,\ell\}$, and an integer $k$. A {\em rainbow matching} is a matching $M$ whose edges have distinct colors, that is, for any two edges $e,e' \in M$, we have $\psi(e) \neq \psi(e')$. In the \rainm problem, the objective is to decide if there exists a $k$-sized rainbow matching.
It is known that \rainm is \npc even when the input graph is a path~\cite{pfender2014complexity}.}  

\smallskip

\noindent{\bf Main idea behind the reduction:}~In the following exposition, we refer the reader to \Cref{figure:hardness}, where we have shown the reduction when applied to a path of three vertices of the instance of \rainm: $(v_{1}, v_{2}, v_{3})$, whose colors are $\psi(v_{1}v_{2})$ and $\psi(v_{2} v_{3})$. The main ingredients of our reduction are as follows: \begin{enumerate*}[label=(\roman*)]
\item We create a path such that the distinguished candidate $c^\star$ can win in at most $k+2$ districts (depicted in \Cref{figure:hardness} as the red portion of the path). We think of the colors of edges in \rainm instance as candidates. Each edge of \rainm instance corresponds to a subpath, we call a {\it segment}. We ensure that a segment can yield at most $k+1$ districts that are won by the color of the edge it corresponds to. Additionally, we have two additional candidates $c$ and $\hat{c}$ whose role will be clear from the formal exposition of the weights on the vertices. Moreover, we set the value $k'=k^{2}+4k+4$, the number of districts in the instance of \gmpvr. 

%This is done by introducing two additional candidates $c$ and $\hat{c}$ who get large weights from $k+1$ vertices in the segment.
\item Given a $k$-sized rainbow matching, say $M$, each color that appear in the matching wins in $k+1$ districts. Since no color appears more than once in $M$, it cannot win more than $k+1$ districts. Our construction ensures that only $c^\star$ wins in $k+2$ districts. This gives a solution for \gmpvr. 
\item For the reverse direction, our gadget ensures that in any solution of the constructed instance of \gmpvr, no color can win more than $k+1$ districts (otherwise $c^\star$  cannot win maximum number of districts). Consequently, a color does not win in two segments corresponding to two distinct edges. We construct a matching by taking an edge whose color wins in $k+1$ districts in the corresponding segment. Then, no color appears more than once in a matching ensures the rainbow matching condition. There are two vertices between every pair of segments, and they both give large weight to two dummy candidates $\hat{c}$ and $c$. Unless there exist $k$ segments where the colors win, we will not get desired number of districts. %$\hat{c}$ wins in at least $k+2$ districts, a contradiction. (Note that $\hat{c}$ wins in $k+1$ districts marked by green in the figure.) 
If the edges corresponding to these $k$ segments do not form a matching, then $c$ or $\hat{c}$ wins in more than $k+1$ districts. These properties together ensure that given a solution to the reduced instance of \gmpvr, we will obtain a $k$-sized rainbow matching for the instance of \rainm. 
\end{enumerate*}

\begin{figure}[b!]
\centering
%\vspace{-0.5cm}
\hspace{-0.3cm}\includegraphics[scale=0.33]{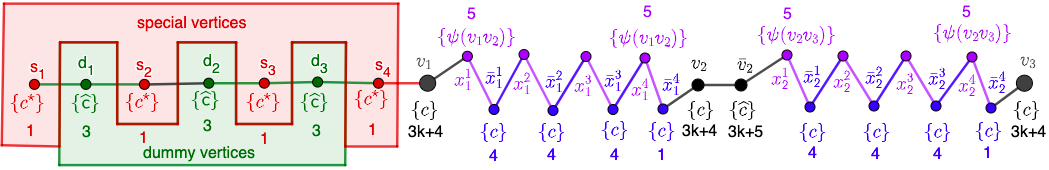}
%\vspace{-4.2cm}
\caption{An illustration of the construction of graph $H$ used in the proof of \textsf{NP}-completeness of \gmpvr for an instance of \rainm on path $(v_{1}, v_{2}, v_{3})$ and $k = 2$. %The vertices are represented in colors, 
The candidate approved by a vertex is in braces directly below the vertex; and the assigned weights appear below the candidates. % Colored  Each vertex assign non-zero weight to only the set of candidates written below it. %All the \textcolor{zzttff}{purple} colored vertices assign same weight to the candidates.}
}\label{figure:hardness}
\end{figure}

Next, we describe our reduction.

\smallskip

\begin{sloppypar}\noindent{\bf Construction.} Let $\Co{I}=(G,\psi,k)$ denote an instance of \rainm, where we assume that $k\geq 5$, or else it is a trivial yes-instance. We create an instance $\Co{J}=(H,\Co{C},a,w, p,k')$ of \gmpvr as follows. 
Let $V(G)=\{v_1,\ldots,v_{\tilde{n}}\}$. 

\paragraph{\bf Construction of the graph $H$.}
\begin{itemize}
\item Corresponding to the vertices $v_1$ and $v_{\tilde{n}}$, we add vertices $v_1$ and $v_{\tilde{n}}$ in $V(H)$. For each vertex $v_i \in V(G)$, where $i\in \{2, \ldots, \tilde{n}-1\}$, we add two vertices $v_i$ and $\bar{v}_i$ in $V(H)$. 
\item For each $i\in \{1, \ldots, \tilde{n}-1\}$, we add a path on $2k+2$ vertices, denoted by $(x_i^1,\bar{x}_i^1,x_i^2,\bar{x}_i^2,\ldots,x_i^{k+1},\bar{x}_i^{k+1})$ to $H$. 
\item We add the edge $v_1x_1^1$. Moreover, for each $i\in \{2, \ldots, \tilde{n}-1\}$, we add edges $\bar{v}_ix_i^1$ and $v_{i+1}\bar{x}_i^{k+1}$ to $E(H)$. 
\item Additionally, we add a set of {\em special vertices} of size $k+2$, denoted by $S=\{s_1,\ldots,s_{k+2}\}$, and a set of {\em dummy vertices} of size $k+1$, say $D=\{d_1,\ldots,d_{k+1}\}$, in $H$. 
\item We add edges $s_id_i$ and $d_is_{i+1}$, for each $i\in \{1,\ldots, k+1\}$, as well as edge $s_{k+2}v_1$ to $E(H)$. 
\end{itemize}
Note that graph $H$ is a path, as depicted below and in Figure~\ref{figure:hardness}.
\[s_{1} d_{1} \!\ldots s_id_i \!\ldots  d_{k+1} s_{k+2} v_{1} x_1^1\!\dots% x_{1}^{j}, \ldots,
 \bar{x}_1^{k+1} v_{2} \bar{v}_2 \!\ldots %v_{i}, \ldots,
% v_{\tilde{n}-1},\bar{v}_{\tilde{n}-1},
 x_{\tilde{n}-1}^1,\bar{x}_{\tilde{n}-1}^1\!\ldots \bar{x}_{\tilde{n}-1}^{k+1}v_{\tilde{n}} \]
\end{sloppypar}

\paragraph{\bf Weight function:}
Next, we define the weight function $w\colon V(H) \rightarrow \mathbb{Z}^+$.% as follows.

\begin{itemize}%[nosep, topsep=0pt, wide=0pt]
    \item For each $i\in \{1,\ldots,\tilde{n}\}$, we set $w(v_i)=3k+4$.
    \item For each $i\in \{2,\ldots,\tilde{n}-1\}$, we set $w(\bar{v}_i)=3k+5$.
    \item For each $i \in \{1,\ldots,\tilde{n}-1\}$, $j\in \{1,\ldots,k+1\}$, we set $w(x_i^j)=5$.
    \item For each $i \in \{1,\ldots,\tilde{n}-1\}$, $j\in \{1,\ldots,k\}$, we set  $w(\bar{x}_i^j)=4$.
    \item For each $i \in \{1,\ldots,\tilde{n}-1\}$, $j= k+1$, we set $w(\bar{x}_i^j)=1$.
    \item For each $i\in \{1,\ldots,k+2\}$, we set $w(s_i)=1$.
    \item For each $i\in\{1,\ldots,k+1\}$, we set $w(d_i)=3$.
\end{itemize}

%\il{Change of notation c^\star=p, c=y$ and $z=\bar{c}$ }

\paragraph{\bf Candidate Set.}
Next, we describe the set of candidates. For each color $i\in \{1,\ldots,|\psi(E(G))|\}$, we have a candidate $i$ in $\Co{C}$, the candidate set. We also have three additional {\em special candidates} \Red{$c^\star, c$, and $\hat{c}$} in $\Co{C}$, where $c^{\star}$ is the distinguished candidate, that is $p=c^\star$. 

\paragraph{\bf Approval function.}

Next, we describe the approval function $a\colon V(H)\rightarrow \Co{C}$.
\begin{inparaenum}[(i)]%[nosep, topsep=0pt, wide=0pt]
    \item Every special vertex $s \in V(H)$ approves the special candidate $c^\star$, that is, $a(s)=c^\star$.
    \item Every dummy vertex $d \in V(H)$ approves the candidate $\hat{c}$, that is, $a(d)=\hat{c}$.
    \item Every vertex $v_i \in V(H)$, where $i\in \{1,\ldots,\tilde{n}\}$, approves the candidate $c$, that is, $a(v_i)=c$.
    \item Every vertex $\bar{v}_i \in V(H)$, where $i\in \{2,\ldots,\tilde{n}-1\}$, approves the candidate $\hat{c}$, that is, $a(\bar{v}_i)=\hat{c}$.
    \item Every vertex $x_i^j \in V(H)$, where $i\in \{1,\ldots,\tilde{n}-1\}, j\in \{1,\ldots,k+1\}$, approves the candidate $\psi(v_iv_{i+1})$, that is, $a(x_i^j)=\psi(v_iv_{i+1})$.
    \item Every vertex $\bar{x}_i^j \in V(H)$, where $i\in \{1,\ldots,\tilde{n}-1\}, j\in \{1,\ldots,k+1\}$, approves the candidate $c$, that is, $a(\bar{x}_i^j)=c$.
\end{inparaenum}

%\il{Refer to Figure~\ref{}}

%We set target candidate $p$ as $c^\star$. 

\paragraph{\bf Number of districts.}
Next, we describe the choice for the number of districts. Intuitively speaking, we want to create $k+2$ districts each containing only special vertices, $k+1$ districts each containing only dummy vertices, $k+1$ districts containing $v_i,\bar{v}_j$, where $i\in \{1,\ldots,\tilde{n}\}, j\in \{2,\ldots,\tilde{n}-1\}$, and $k(k+1)$ some other districts. Consequently, we set  $k'=k^2+4k+4$, the number of districts. %The proof of correctness is in the Supplementary File (Section~\ref{correctness:np-hardness}). %The $k(k+1)$ districts containing vertices $\{

\medskip
\noindent{\bf Correctness.} Next, we show the equivalence between the instance $\Co{I}$ of \rainm and the instance $\Co{J}$ of \gmpvr. Formally, we prove the following:
\begin{sloppypar}
\begin{lemma}\label{lem:equivalence}
$\Co{I}=(G,\psi,k)$ is a \yes-instance of \rainm if and only if $\Co{J}=(H,\Co{C},a,w, c^\star,k')$ is a \yes-instance of \gmpvr.
\end{lemma}
\end{sloppypar}

%\begin{claims}\label{claim:y-win}
%If a district contains $v_i$, where $i\in \tilde{n}$, then it is won by \Ma{$c=y$}. 
%\end{claims}
%\begin{proof}
%\todo[inline]{add proof here}
%\end{proof}
\begin{proof}%[Proof of Lemma~\ref{lem:equivalence}]
We start the proof with the following claim which will be extensively used in the proof. % of Lemma~\ref{lem:equivalence}.

\begin{claim}\label{claim:y-wins-in-large-district}
If there exists a district $P$ such that $|P|\geq 6$ and $v_i \in P$, for some $i\in \{1,\ldots,\tilde{n}\}$, then \Red{$c$} is the unique winner in the district  $P$.
\end{claim}

\begin{proof} Let $P$ contains $v_r$ and $v_t$, where $r,t \in \{1,\ldots,\tilde{n}\}$, but does not contain $v_{t+1}$, if $t<\tilde{n}$, and $v_{r-1}$, if $r>1$. Note that since $H[P]$ is connected, it contains all the vertices in the subpath from $v_r$ to $v_t$ in $H$. We consider several cases depending on the values of $r$ and $t$. 
\begin{description}
\item [Case $r=1$.] Since $H[P]$ is connected, $P$ contains all the vertices in the subpath from $v_1$ to $v_t$ in $H$. Note that $P$ may also contain special and dummy vertices as well as vertices from the set $\{\bar{v}_t, x_t^1,\ldots,\bar{x}_t^{k+1}\}$, if $t<\tilde{n}$. 
The total weight of the candidate $c^\star$ in the district $P$ is at most $k+2$ due to the presence of the path $s_1,\ldots,s_{k+1}$.  We further consider cases depending on whether $P$ contains $x_1^1$. %, for any $i\in \{1,\ldots,t\}$ and $j \in \{1,\ldots, k+1\}$.
\begin{description}
\item [$P$ does not contain $x_1^1$.] In this case the total weight of the candidate \Red{$\hat{c}$} is at most $3(k+1)$; it is $3k+4$ for the candidate $c$; and it is $0$ for every other candidate. Thus, $c$ is unique winner in $P$. 
\item [$P$ contains $x_1^1$.] In this case the total weight of \Red{$\hat{c}$} is at most $3(k+1)+(t-1)(3k+5)$. Let $x_{t'}^j$ be a vertex in $P$, where $t'\in\{t-1,t\}$, such that $x_{t'+1}^{j'}$ is not in $P$ for any $j'\in \{1,\ldots,k+1\}$, if $t'=t-1$. If $t'=t-1$, then for $\psi(v_iv_{i+1})$, where $i\in \{1,\ldots,t'\}$, it is at most $5(t-1)(k+1)$; and for the candidate $c$ it is at least $t(3k+4)+(t-1)(4k+1)$. Thus, $c$ is the unique winner in the district $P$. If $t'=t$, then let $x_{t'}^{j'}$ be a vertex in $P$, where $j' \in \{1,\ldots,k+1\}$, such that $x_{t'}^{j'+1}$ is not in $P$, if $j'< k+1$. In this case the total weight of $\psi(v_iv_{i+1})$ is at most $5(t-1)(k+1)+5j'$; and for the candidate $c$ it is at least $t(3k+4)+(t-1)(4k+1)+4(j'-1)$. Thus, $c$ is the unique winner in the district $P$. 
 \end{description}
\item [Case $r>1$.] Clearly, in this case $P$ does not contain special and dummy vertices. %Let $P$ contains $v_r$ and $v_t$, where $r,t \in \{2,\ldots,\tilde{n}\}$ and $r\leq t$, and does not contain $v_{r-1}$ and $v_{t+1}$ if $t<\tilde{n}$. As argued above, $P$ contains all the vertices in the subpath from $v_r$ to $v_t$ of $H$. 
We further consider the cases depending on whether $r<t$ or $r=t$.
\begin{description}
\item [Case $r<t$.] Note that $P$ can contain vertices from the set $\{\bar{v}_t, x_t^1,\ldots,x_t^{k+1}\}$ if $t<\tilde{n}$, and from the set $\{\bar{v}_{r-1}, x_{r-1}^1,\ldots,\bar{x}_{r-1}^{k+1}\}$. We further consider the following cases. 
\begin{description}
\item [$P$ neither contains $x_t^1$ nor $x_{r-1}^{k+1}$.] In this case the total weight of $\hat{c}$ is at most $(t-r+1)(3k+5)$; for $\psi(v_iv_{i+1})$, where $i\in \{r,\ldots,t-1\}$, it is at most $5(t-r)(k+1)$; and for $c$ it is $(t-r+1)(3k+4)+(t-r)(4k+1)$. Thus, $c$ is the unique winner in $P$. 
\item [$P$ contains $x_t^1$, but not $x_{r-1}^{k+1}$.] Let $x_{t}^{j'}$ be a vertex in $P$, where $j' \in \{1,\ldots,k+1\}$, such that $x_{t}^{j'+1}$ is not in $P$, if $j'< k+1$. In this case the total weight of $\hat{c}$ is at most $(t-r+1)(3k+5)$; for $\psi(v_iv_{i+1})$, where $i\in \{r,\ldots,t\}$, it is at most $5(t-r)(k+1)+5j'$; and for $c$ it is at least $(t-r+1)(3k+4)+(t-r)(4k+1)+4(j'-1)$. Thus, $c$ is the unique winner in $P$. 
\item [$P$ contains $x_{r-1}^{k+1}$, but not $x_{t}^1$.] If $P$ contains $\bar{v}_{r-1}$, then clearly, due to the connectivity, $P$ contains all the vertices in $\{\bar{v}_{r-1}, x_{r-1}^1,\ldots,\bar{x}_{r-1}^{k+1}\}$. In this case the total weight of $\hat{c}$ is at most $(t-r+2)(3k+5)$; for $\psi(v_iv_{i+1})$, where $i\in \{r,\ldots,t-1\}$, it is at most $5(t-r+1)(k+1)$; and for $c$ it is $(t-r+1)(3k+4)+(t-r+1)(4k+1)$. Thus, $c$ is the unique winner in $P$. Suppose that $P$ does not contain $\bar{v}_{r-1}$. Let $x_{r-1}^{j'}$ be a vertex in $P$, where $j'\in\{1,\ldots,k+1\}$, such that $x_{r-1}^{j'-1}$ is not in $P$, if $j'>1$. In this case, the total weight of $\hat{c}$ is at most $(t-r+1)(3k+5)$; for $\psi(v_iv_{i+1})$, where $i\in \{r-1,\ldots,t-1\}$ it is at most $5(k+1)(t-r)+5(k+2-j')$; and for $c$ it is at least $(t-r+1)(3k+4)+(t-r)(4k+1)+4(k+1-j')+1$. Thus, $c$ is the unique winner in $P$. 
\item [$P$ contains both $x_{r-1}^{k+1}$ and $x_{t}^{1}$.] We first consider the case when $P$ contains $\bar{v}_{r-1}$. Let $x_{t}^{j'}$ be a vertex in $P$, where $j'\in\{1,\ldots,k+1\}$, such that $x_{t+1}^{j'+1}$ is not in $P$, if $t< k+1$. In this case, the total weight of $\hat{c}$ is at most $(t-r+2)(3k+5)$; for $\psi(v_iv_{i+1})$, where $i\in \{r-1,\ldots,t\}$, it is at most $5(k+1)(t-r+1)+5(k+2-j')$; and for $c$ it is at least $(t-r+1)(3k+4)+(t-r+1)(4k+1)+4(k+2-j')$.  Thus, $c$ is the unique winner in $P$. 
\end{description}
\item [Case $r=t$.] If $r=\tilde{n}$, then $P$ is a subpath of  $(\bar{v}_{\tilde{n}-1},x_{\tilde{n}-1}^1,\ldots,v_{\tilde{n}})$, and using the similar argument as above $c$ wins in such a district uniquely. If $r\neq \tilde{n}$, then $P$ is subpath of $(\bar{v}_{r-1}, x_{r-1}^1,\ldots,\bar{x}_{r}^{k+1})$, and using the same argument as above, $c$ wins in the district $P$ uniquely.
\end{description}
\end{description}
 
\end{proof}

Next, we move towards proving  Lemma~\ref{lem:equivalence}. 
%
%\begin{sloppypar}
($\Rightarrow$) For the forward direction, let $M$ be a solution to $\Co{I}$. We create a $k'$-partition of $V(H)$, denoted by $\Co{P}$, as follows. Let $P_S=\{\{s_1\}, \ldots, \{s_{k+2}\}\}$, $P_D=\{\{d_1\},\ldots,\{d_{k+1}\}\}$, and $P_X=\{\{x_i^1, \bar{x}_i^1\}, \ldots, \{x_i^{k+1},\bar{x}_i^{k+1}\}\colon v_i v_{i+1}\in M\}$. We add $P_S,P_D$ and $P_X$ to $\Co{P}$.  Let $\tilde{G}$ be the graph obtained from $H$ after deleting all the  special vertices, dummy vertices, and $x_i^j,\bar{x}_i^j$, for all  $v_iv_{i+1}\in M$ and $j\in \{1,\ldots,k+1\}$.  Since $|M|=k$, we have $k+1$ connected components in $\tilde{G}$. Let these connected components be denoted by $\tilde{G}_1, \ldots, \tilde{G}_{k+1}$. For each $h\in \{1,\ldots,k+1\}$, we add the set $V(\tilde{G}_h)$ to $\Co{P}$. Note that $\Co{P}$ is a partition of $V(H)$ and every set in $\Co{P}$ is connected. We observe that
%\end{sloppypar}
 \begin{itemize}[itemsep=0.25em]%[(i)]
\item the candidate $c^\star$ wins in every district in $P_S$. Hence, there are at least $k+2$ districts won by $c^\star$ in $\Co{P}$.
\item the candidate \Red{$\hat{c}$} wins in every district in $P_D$.  Therefore, there are at least $k+1$ districts won by \Red{$\hat{c}$} in $\Co{P}$.
\item for an edge $v_iv_{i+1}\in M$, the candidate $\psi(v_iv_{i+1})$ wins in every district $\{x_i^j,\bar{x}_i^j\}$ in $P_X$, where $j\in \{1,\ldots,k+1\}$, and hence  there are at least $k+1$ districts won by $\psi(v_iv_{i+1})$ in $\Co{P}$.
\end{itemize}

We next claim that for each $h\in \{1,\ldots,k+1\}$, the candidate \Red{$c$} wins in the district $V(\tilde{G}_h)$. We first observe that $|V(\tilde{G}_h)|$ is either $1$ or at least $2k+2$. This is due to the fact that $M$ is a matching, so for any $i\in \{1,\ldots,\tilde{n}-1\}$ and $j,j'\in \{1,\ldots,k+1\}$, we do not delete both $x_i^j$ and $x_{i+1}^{j'}$ to construct the graph $\tilde{G}$. We first consider the case when $|V(\tilde{G}_h)|=1$. Due to the construction of the districts, if $|V(\tilde{G}_h)|=1$, then either $V(\tilde{G}_h)$ is $\{v_1\}$  or $\{v_{\tilde{n}}\}$. Since $v_1$ and $v_n$ both approves $c$, the candidate \Red{$c$} wins in the districts $\{v_1\}$ and $\{v_{\tilde{n}}\}$ uniquely. We next consider the case when $|V(\tilde{G}_h)|\geq 2k+2$. By the construction of $\tilde{G}_h$, %where $h\in [k+1]$, %it is a subpath of $G'$, does not contain special vertices and dummy vertices, and 
it contains at least two vertices from the set $\{v_1,\ldots,v_{\tilde{n}}\}$. %, and if $x_\ell^j \in V(\tilde{G}_i)$, where $\ell \in [\tilde{n}-1], j\in [k+1]$, then $v_\ell',v_{\ell+1}'\in V(\tilde{G}_i)$. Let $\tilde{G}_i$ be a subpath \todo{define subpath} of $G$ from $v_t'$ to $v_{t+\ell}'$, where $t'\in [\tilde{n}]$ and $\ell \in [\tilde{n}-t']$. Now, we compute the weight of every candidate in the district $V(\tilde{G}_i)$. The total weight for the candidate \Ma{\Red{$c$}} in $V(\tilde{G}_i)$, that is $\sum_{v\in V(\tilde{G}_i)}w_v(y)$, is $(\ell+1)(2k+4)+\ell(k+1)=\ell(3k+5)+2k+4$. Similarly, the total weight for the candidate $z$ in $V(\tilde{G}_i)$ is $(\ell+1)(2k+5)$  which is less than the weight of \Red{$c$} in  $V(\tilde{G}_i)$ as $k\geq2$. The other candidates whose weight is non-zero in $V(\tilde{G}_i)$ are $\psi(v_hv_{h+1})$, where $t\leq h \leq t+\ell-1$. For every $t\leq h \leq t+\ell-1$, the weight of the candidate $\psi(v_hv_{h+1})$ is at most $2\ell(k+1)$ which is less than the weight of \Red{$c$} in  $V(\tilde{G}_i)$. 
Therefore, due to Claim~\ref{claim:y-wins-in-large-district}, \Red{$c$} wins in the district $V(\tilde{G}_h)$ uniquely, when $|V(\tilde{G}_h)|\geq 2k+2$. Thus, for each $h\in \{1,\ldots,k+1\}$, $c$ wins in the district $V(\tilde{G}_h)$ uniquely. 
Since $c$ wins in $V(\tilde{G}_h)$ uniquely, for each $h\in \{1,\ldots,k+1\}$, due to the above observations $c^\star$ wins in exactly $k+2$ districts, and  $\hat{c}$ and $\psi(v_iv_{i+1})$ win in exactly $k+1$ districts.
Since $c^\star$ wins in $k+2$ districts and every other candidate wins in at most $k+1$ districts in $\Co{P}$, $\Co{P}$ is a solution to $\Co{J}$.  \par

%\il{end of $\Rightarrow$}
%In the backward direction, 
($\Leftarrow$) For the reverse direction, let $\Co{P}=\{P_1,\ldots,P_{k'}\}$ be a solution to $\Co{J}$. We create a set of edges $M\subseteq E(G)$ as follows. If there are $k+1$ districts which are subpaths of $(x_i^1,\ldots, \bar{x}_{i}^{k+1})$, where $i\in \{1,\ldots,\tilde{n}-1\}$, such that $\psi(v_iv_{i+1})$ wins in these districts, then we add $v_iv_{i+1}$ to $M$. We next prove that $M$ is a solution to $\Co{I}$. We begin with proving some properties of the partition $\Co{P}$. Let $\Co{\tilde{P}} \subseteq \Co{P}$ be the set of districts that contain $v_i$ or $\bar{v}_j$, where $i\in \{1,\ldots,\tilde{n}\}$ and $j\in \{2,\ldots,\tilde{n}-1\}$. The next set of claims complete the proof.

\begin{claim}\label{c winning sets}
Every district in $\Co{\tilde{P}}$ is won by either \Red{$c$} or \Red{$\hat{c}$}.
\end{claim} %\ma{suppressed}

\begin{proof}
We first argue for the districts that contains $v_i$, where $i\in \{1,\ldots,\tilde{n}\}$, but not $\bar{v}_j$ for any $j\in \{2,\ldots,\tilde{n}-1\}$. Suppose that $P$ is such a district in $\Co{\tilde{P}}$. If $P$ contains $v_1$ or $v_2$, then clearly, $P$ is a subpath of $(s_1,\ldots,v_2)$. Note that the weight of $c^\star$ in $P$ is at most $k+2$ and for $\hat{c}$, it is at most $3(k+1)$. If $P$ does not contain $x_1^j$ for any $j\in \{1,\ldots,k+1\}$, then the total weight of $c$ is $3k+4$. For any other candidate, it is $0$, thus, $c$ wins in the district $P$. Suppose that $P$ contains $x_1^j$, where $j\in \{1,\ldots,k+1\}$, but nor $x_1^{j+1}$, if $j<k+1$. In this case the total weight of $c$ is at least $3k+4+4(j-1)$, and for $\psi(v_1v_2)$, it is $5j$. For any other candidate, it is $0$, thus, $c$ wins in the district $P$.  If $P$ contains $v_i$, where $i>2$, then, clearly $P$ is a subpath of $(x_{i-1}^1,\ldots,v_i)$, and $c$ wins in such a district. Next, we argue for the districts that contains $\bar{v}_j$, where $j\in \{2,\ldots,\tilde{n}-1\}$, but not $v_i$, for any $i\in \{1,\ldots,\tilde{n}\}$. Suppose that $P$ is such a district in $\Co{\tilde{P}}$. Note that $P$ is a subpath of $(\bar{v}_j, x_{j}^1,\ldots,\bar{x}_{j}^{k+1})$, and $\hat{c}$ wins in such a district. Next, we consider the districts in $\Co{\tilde{P}}$ that contains both $v_i$, where $i\in \{1,\ldots,\tilde{n}\}$, and $\bar{v}_j$, where $j\in \{2,\ldots,\tilde{n}-1\}$. Suppose that $P$ is such a district in $\Co{\tilde{P}}$. We consider the following cases depending on the size of $P$. 
\begin{sloppypar}
\begin{itemize} 
\item if $|P|=2$, then due to the construction of $H$, $P$ is $\{v_i,\bar{v}_i\}$, and $\hat{c}$ wins in such a district. \item if $|P|=3$, then due to the construction of $H$, $P$ is either $\{v_i,\bar{v}_i\}$ or $\{v_i,\bar{v}_i,x_{i}^1\}$ or $\{v_i,\bar{v}_i,\bar{x}_{i-1}^{k+1}\}$, and $c$ or $\hat{c}$ or both wins in such a district. 
\item if $|P|=4$, then $P$ is either $\{x_{i-1}^{k+1},\ldots,\bar{v}_i\}$ or $\{\bar{x}_{i-1}^{k+1},\ldots,x_i^1\}$ or $\{v_i,\ldots,\bar{x}_i^1\}$,  and $c$ or $\hat{c}$ or both wins in such a district.
 \item if $|P|=5$, then $P$ is either $\{\bar{x}_{i-1}^k, \ldots,\bar{v}_i\}$ or $\{x_{i-1}^{k+1},\ldots,x_i^1\}$ or $\{\bar{x}_{i-1}^{k+1},\ldots,\bar{x}_i^1\}$ or $\{v_i,\ldots, x_i^2\}$,  and $c$ or $\hat{c}$ or both wins in such a district.
%  \item if $|P|=6$, then $P$ is either $\{x_{i-1}^k,\ldots,\bar{v}_i\}$ or $\{\bar{x}_{i-1}^k,\ldots, x_i^1\}$ or $\{x_{i-1}^{k+1}, \ldots,\bar{x}_i^1\}$ or $\{\bar{x}_{i-1}^{k+1},\ldots,x_i^2\}$ or $\{v_i,\ldots,\bar{x}_i^2\}$, and in such a district either $c$ wins or $\hat{c}$ wins or both. 
%  \item if $|P|=7$, then $P$ is either $\{\bar{x}_{i-1}^{k-1}, \ldots,\bar{v}_i\}$ or $\{x_{i-1}^k, \ldots, x_i^1\}$ or $\{\bar{x}_{i-1}^k, \ldots, \bar{x}_i^1\}$ or $\{x_{i-1}^{k+1},\ldots,x_i^2\}$ or $\{\bar{x}_{i-1}^{k+1},\ldots,\bar{x}_i^2\}$ or $\{v_i,\ldots, x_i^3\}$, and in such a district either $c$ wins or $\hat{c}$ wins or both. 
 \item If $|P|\geq 6$, then due to Claim~\ref{claim:y-wins-in-large-district}, $c$ wins in the district $P$.
 \end{itemize}  
 \end{sloppypar}\end{proof}

\begin{claim}
\label{claim:k-blocks-contains-sol-edges}
The size of the set $\Co{\tilde{P}}$ is at most $k+1$.
\end{claim}

\begin{proof}
Suppose that $|\Co{\tilde{P}}|\geq k+2$. Due to Claim~\ref{c winning sets}, every district in $\Co{\tilde{P}}$ is won by either \Red{$c$} or \Red{$\hat{c}$}. Let $n_c$ and $n_{\hat{c}}$ be the number of districts won by \Red{$c$} and $\hat{c}$, respectively, in $\Co{\tilde{P}}$. Clearly, $n_c+n_{\hat{c}} \geq k+2$ as $|\Co{\tilde{P}}|\geq k+2$. Note that $c^\star$ can win in at most $k+2$ districts as only these many vertices approve $c^\star$. Since $c^\star$ is the distinguished candidate, \Red{$\hat{c}$} can win at most $k+1$ districts. Note that if a district contains only dummy vertices and special vertices, then it is won by \Red{$\hat{c}$}, by the construction.  Let $\Co{P}_{\hat{c}}$ denote the set of all districts in $\Co{P}\setminus \Co{\tilde{P}}$ that contain at least one dummy vertex. Every district in $P_{\hat{c}}$ is won by $\hat{c}$ because either they contain only a dummy vertex or dummy and special vertices.

Thus, it follows that  $1\leq |\Co{P}_{\hat{c}}|+n_{\hat{c}}\leq k+1$, since \Red{$\hat{c}$} can only win at most $k+1$ districts. By the construction of the graph $H$, there are at most $|\Co{P}_{\hat{c}}|+1$ districts containing only special vertices. Therefore, there are at most $|\Co{P}_{\hat{c}}|+1$ districts won by $c^\star$ as $c^\star$ can only win a district which contains only special vertices. Thus, there are at most $k+2-n_{\hat{c}}$ districts won by $c^\star$. Since $n_c+n_{\hat{c}} \geq k+2$, we have that the number of districts won by $c$ is at least $k+2-n_{\hat{c}}$, a contradiction to the fact that $\Co{P}$ is a solution to $\Co{J}$.
\end{proof}

%Due to Claim~\ref{claim:k-blocks-contains-sol-edges}, we have that there are at most $k(k+2)$ districts that contain only $x_i^j$, where $i\in [\tilde{n}], j\in [k+1]$. 

\begin{claim}\label{claim:non-emptyM}
The set $M$ is non-empty. 
\end{claim}

\begin{proof}
For the sake of contradiction, suppose that $M=\emptyset$. Then, due to the construction of the set $M$, we know that  for each $i\in \{1,\ldots,\tilde{n}\}$, there are at most $k$ districts which are subpaths of $(x_i^1,\ldots, \bar{x}_{i}^{k+1})$ that are won by $\psi(v_iv_{i+1})$. Suppose that $n_c$ and $n_{\hat{c}}$ be the number of districts in $\Co{\tilde{P}}$ that are won by $c$ and $\hat{c}$, respectively. Then, there can be at most $k+1-n_c$ districts of type $\{\bar{x}_i^j\}$, where $i\in \{1,\ldots,\tilde{n}-1\}$, $ j\in \{1,\ldots,k+1\}$, as these districts are also won by $c$ and $c$ wins at most $k+1$ districts, since the distinguished candidate $c^\star$ can win at most $k+2$ districts. Since $|\Co{\tilde{P}}|\leq k+1$  (Claim~\ref{claim:k-blocks-contains-sol-edges}) and for each $i\in \{1,\ldots,\tilde{n}-1\}$, there are at most $k$ districts which are subpaths of $(x_i^1,\ldots, \bar{x}_{i}^{k+1})$ that are won by $\psi(v_iv_{i+1})$, it follows that there are at most $k^2$ districts that contains $x_i^j$ but not $v_i$ or $v_{i+1}$. Let $\Co{P}_{\hat{c}}$ denote the set of all districts in $\Co{P}\setminus \Co{\tilde{P}}$ that contain at least one dummy vertex. Using the same argument in Claim~\ref{claim:k-blocks-contains-sol-edges}, every district in $P_{\hat{c}}$ is won by $\hat{c}$. Thus, $|\Co{P}_{\hat{c}}|\leq k+1-n_{\hat{c}}$ and there are at most $k+2-n_{\hat{c}}$ districts won by $c^\star$. Therefore, the total number of districts in $\Co{P}$ is at most
\begin{equation*}
    \begin{split}
        & n_c+n_{\hat{c}}+(k+1-n_c)+k^2+(k+1-n_{\hat{c}})+(k+2-n_{\hat{c}}) \\
        = & k^2+3k+4-n_{\hat{c}} \\
        < & k'
    \end{split}
\end{equation*}
 a contradiction to the fact that $\Co{P}$ is a solution to $\Co{J}$. 
 \end{proof}

\begin{claim}
\label{cor}
Candidate $c^\star$ wins in $k+2$ districts. Moreover, in $\Co{P}$  there are $k+2$ districts containing only special vertices and $k+1$ districts containing only dummy vertices.
\end{claim}

\begin{proof}Since $M\neq \emptyset$, by the construction of $M$, there exists at least one $i\in \{1,\ldots,\tilde{n}-1\}$ such that there are $k+1$ districts which are subpaths of $(x_i^1,\ldots, x_{i}^{k+1})$ that are won by $\psi(v_iv_{i+1})$. Since $\Co{P}$ is a solution to the instance $\Co{J}$, $c^\star$ must win in at least $k+2$ districts. Since there are only $k+2$ vertices who approve $c^\star$, it can win in at most $k+2$ districts. Consequently, there are $k+2$ districts in $\Co{P}$ containing only special vertices (districts won by $c^{\star}$) and additional $k+1$ districts in $\Co{P}$ containing only dummy vertices. 
\end{proof}

Due to Claim~\ref{cor}, we have the following:
\begin{corollary}\label{cor:c wins in large districts}
Every district in $\Co{\tilde{P}}$ is won by the candidate $c$.
\end{corollary}

%Next, we prove that $M$ is a solution to $\Co{I}$.
% which concludes the proof. 
\begin{claim}\label{sizeM}\label{claim:rainbow-matching}
Set $M$ is a rainbow matching of size $k$.  
\end{claim}

\begin{proof} First, we show that $|M| = k$.
Since $|\Co{\tilde{P}}|\leq k+1$ (\Cref{claim:k-blocks-contains-sol-edges}), due to the construction of $M$, we know that $|M|\leq k$. Suppose that $|M|<k$. Then, for at most $k-1$ $i$s, where $i\in \{1,\ldots,\tilde{n}-1\}$, $\Co{P}$ contains $k+1$ districts that are subpaths of $(x_i^1,\ldots,\bar{x}_i^{k+1})$ and won by $\psi(v_iv_{i+1})$. Moreover, since $|\Co{\tilde{P}}|\leq k+1$, there are at most $(k+1)(k-1)+k$ districts in $\Co{P}$ that are subpaths of some $(x_i^1,\ldots,\bar{x}_i^{k+1})$ and won by $\psi(v_iv_{i+1})$, where $i\in \{1,\ldots,\tilde{n}-1\}$. Let $|\Co{\tilde{P}}|= \tilde{k}$. Due to Corollary~\ref{cor:c wins in large districts}, we know that there are at most $k+1-\tilde{k}$ districts of type $\bar{x}_i^j$, where $i\in \{1,\ldots,\tilde{n}-1\}$, $ j\in \{1,\ldots,k+1\}$, as $c$ wins in these districts as well and the distinguished candidate wins in $k+2$ districts (Claim~\ref{cor}). Thus, the total number of districts in $\Co{P}$ is at most $(2k+3)+(k^2-1+k)+(k+1) = k^2+4k+3 < k'$, a contradiction.

Now we prove that $M$ is a rainbow matching. We first prove that $M$ is a matching. Suppose not, then for some $i\in \{1,\ldots,\tilde{n}-2\}$, there are $k+1$ districts that are subpaths of $(x_i^1,\ldots,\bar{x}_i^{k+1})$ and $(x_{i+1}^1,\ldots,\bar{x}_{i+1}^{k+1})$, and won by $\psi(v_iv_{i+1})$ and $\psi(v_{i+1}v_{i+2})$, respectively.  Thus, either there is a district $\{\bar{v}_{i+1}\}$ or $\{v_{i+1},\bar{v}_{i+1}\}$ or $\{\bar{x}_i^{k+1},v_{i+1},\bar{v}_{i+1}\}$ in $\Co{P}$. In all these cases, \Red{$\hat{c}$} wins. Due to Claim~\ref{cor}, there are $k+1$ districts in $\Co{P}$ containing only dummy vertices. Therefore, there are $k+2$ districts won by \Red{$\hat{c}$}, a contradiction, because only the distinguished candidate $c^\star$ wins in $k+2$ districts. %In the latter case $c$ wins. Since $|\Co{\tilde{P}}|=k+1$, $c$ wins in $k+2$ districts, a contradiction.
\par
\begin{sloppypar}
We next prove that if edges $v_iv_{i+1},v_hv_{h+1} \in M$, where $i,h\in \{1,\ldots,\tilde{n}-1\}$, $i \neq h$, then $\psi(v_iv_{i+1})\neq \psi(v_h v_{h+1})$. Towards the contradiction, suppose that $\psi(v_iv_{i+1}) = \psi(v_h v_{h+1})$. Due to the construction of the edge set $M$, there are $k+1$ districts that are subpaths of $(x_i^1,\ldots,\bar{x}_i^{k+1})$ and won by $\psi(v_iv_{i+1})$ and $k+1$ districts that are subpaths of $(x_j^1,\ldots,\bar{x}_j^{k+1})$ and won by $\psi(v_jv_{j+1})$. Thus, there are $2k+2$ districts won by $\psi(v_iv_{i+1})$, a contradiction as the target candidate $c^\star$ wins in $k+2$ districts.  \end{sloppypar}
\end{proof}

Due to Claim~\ref{sizeM}, we can conclude that $M$ is a solution to $(G,\psi,k)$. 
\end{proof}

% !TEX root = main.tex

\section{FPT Algorithms for Path}\label{sec:gm_path}
In this section, we prove \Cref{thm:detfpt} and \Cref{thm:ranfpt}, that is,  we present a deterministic and a randomized \FPT algorithm parameterized by $k$ for \targm when the input  is a path. 
 Let $(G,\Co{C},\{w_v\colon \Co{C}\rightarrow \mathbb{Z}^+\}_{v \in V(G)},p,k, k^\star)$ 
be the input instance of \targm, where $G$ is the path $(u_1,\ldots,u_n)$. We begin with a simple observation.

%We first guess the number of districts the target candidate $p$ wins. Let it be denoted by $k^{\star}$, where $2\leq k^{\star} \leq k$. 

\begin{observation}
Given a path $G$ on $n$ vertices, there are $\OO(n^2)$ distinct connected sets.
\end{observation}

%Let $G$ be the input graph and $n=\vert V(G)\vert$. Let $G$ be the path $(v_1,\ldots,v_n)$. 
Based on the above observation we create an auxiliary directed graph $H$ with parallel arcs on ${n \choose 2}+n+2$ vertices,  where we have a vertex for each connected set  of $G$. 
 Note that $G$ is a path on $n$ vertices. For $\{i,j\} \sse [n]$, $i\leq j$, let $P_{i,j}$ denote the subpath of $G$ starting at the $i^{\text{th}}$ vertex and ending at the $j^{\text{th}}$ vertex. That is $P_{i,j}$ is the subpath $(u_i,\ldots,u_j)$ of $G$. 
 Formally, we define the auxiliary graph $H$ as follows.
 
\begin{enumerate}[wide=0pt]
 \item For each $\{i,j\} \sse\{1,\ldots, n\}$ such that $i\leq j$,   create a vertex $v_{i,j}$ corresponding to the subpath  $P_{i,j}$.
 \begin{sloppypar} \item We do the following for each $\{i,j\} \sse \{1,\ldots,n\}$. Let $c$ denote the candidate that wins the district $P_{i,j}$, where $i\leq j$. 
 If $c \neq p$, then we do the following. For each 
  $r \in \{j+1,\ldots,n\}$, we add $k^{\star} -1$ arcs 
 $ \eLabel{v_{i,j},v_{j+1,r},1} , \eLabel{v_{i,j},v_{j+1,r},2 }, \ldots, \eLabel{v_{i,j},v_{j+1,r},k^{\star}-1 }$
from vertex $v_{i,j}$ to $v_{j+1,r}$. We label the $k^{\star}-1$ arcs from $v_{i,j}$ to $v_{j+1,r}$ with $ \eLabel{ c,1} , \eLabel{ c,2 }, \ldots, \eLabel{ c,k^{\star}-1 }$.  That is, for each $k'\in \{1,\ldots,k^{\star}-1\}$, the arc \eLabel{v_{i,j},v_{j+1,r},k'}  is labeled with \eLabel{ c,k'}. 
 If $c = p$, then we do the following. For each 
 $r \in \{j+1,\ldots,n\}$, we add an unlabeled arc 
 from  $v_{i,j}$ to $v_{j+1,r}$.\end{sloppypar}
 \item Finally, we add two new vertices $s$ and $t$. Now we add arcs incident to $s$.  For each $i \in \{1,\ldots,n\}$, we add an unlabeled arc  
 from the vertex $s$ to $v_{1,i}$. Next we add arcs incident to $t$. 
We do the following for each $i \in \{1,\ldots,n\}$. Let $c$ denote the candidate that wins in $P_{i,n}$.
If $c \neq p$, then we add $k^{\star} -1$ arcs $ \eLabel{v_{i,n},t,1} , \eLabel{v_{i,n},t,2 }, \ldots, \eLabel{v_{i,n},t,k^{\star}-1 }$ from $v_{i,n}$ to $t$ and label them with $ \eLabel{ c,1 } ,  \eLabel{ c,2 }, \ldots, \eLabel{ c,k^{\star}-1 }$, respectively. If $c =p$, then we add an unlabeled arc 
%\eLabel{ v_{i,n},t} 
from $v_{i,n}$ to $t$. 
 \end{enumerate}
 
 As the in-degree of $s$ is $0$ and the out-degree of $t$ is $0$, there is no cycle in $H$ that contains either $s$ or $t$.  Since the direction of arcs in $H \setminus \{s,t\}$ is from $v_{i,j}$ to $v_{j+1,r}$, where $\{i,j\} \sse \{1,\ldots,n\}$, $r\in \{j+1,\ldots,n\}$, $H \!\sm\! \{s,t\}$ must be acyclic. This yields the following simple observation. 
 
 \begin{observation}\label{obs:H-is-DAG}
 $H$ is a directed acyclic graph.
 \end{observation}

The following results is the backbone of our deterministic and randomized algorithms.  

\begin{lemma}\label{lem:stpath}
There is a path on $k+2$ vertices from $s$ to $t$ in $H$ such that the path has $k -k^{\star}$ labeled arcs with distinct labels and $k^\star +1$ unlabeled arcs  if and only if $V(G)$ can be partitioned into $k$ districts such that $p$ wins in $k^{\star}$ districts and any other candidate wins in at most $k^{\star}-1$ districts.
\end{lemma}

\begin{proof}
Recall that $G$ is the path $(u_1,u_2, \ldots, u_n)$.
From the construction of $H$, each vertex in $V(H)$ corresponds to a connected set in $G$, that is, each vertex corresponds to a subpath of $(u_1,u_2, \ldots, u_n)$. We observe the following three properties of $H$. 

\begin{enumerate}[wide=0pt]
\item The vertices of $H$ that are connected to $s$ correspond to the subpaths starting at $u_1$. 
That is, for each arc from $s$ to $z$ in $A(H)$, $z=v_{1,j}$ for some $j\in \{1, \ldots, n\}$. 
\item There is an arc from a vertex corresponding to a subpath $P$ of $G$ to a vertex corresponding to a subpath $P'$ of $G$ if $P$ ends at a vertex $u_i$ and $P'$ starts from the next vertex $u_{i+1}$, for some $i \in \{1, \ldots, n\}$. 
\item For each arc from $w$ to $t$, $w=v_{i,n}$ for some $i\in \{1, \ldots, n\}$. 
%If there is an arc $(w,t)\in A(H)$, then the corresponding subpath of the vertex $w$ ends at $v_n$. That is,   $w=v_{i,n}$ for some $i\in [n]$. 
\end{enumerate}

For the digraph $H$, a path is a sequence of vertices and edges denoted by $(v_1,e_1,v_2,e_2,\ldots e_{\ell-1},v_{\ell})$, where $\ell\in {\mathbb N}$ such that $v_1,\ldots,v_{\ell}$ are distinct vertices, $e_1,\ldots,e_{\ell-1}$ are distinct arcs, and for each $i\in \{1,\ldots,\ell-1\}$, $e_i$ is an arc from $v_i$ to $v_{i+1}$. %}\ma{Added here}

($\Rightarrow$) We first prove the forward direction of the lemma. Let $X$ denote a  path from $s$ to $t$ on $k+2$ vertices such that it has $k -k^{\star}$ labeled arcs with distinct labels and $k^\star +1$ unlabeled arcs. 
Therefore, the set of $k$ vertices in $V(X)\setminus \{s,t\}$ correspond to $k$ subpaths in $G$. 
Let these subpaths of $G$ be denoted by $P_1, P_2, \ldots, P_k$. Due to the above three properties $\bigcup_{i=1}^k V(P_i) = V(G)$. Due to the second property, and Observation~\ref{obs:H-is-DAG}, $V(P_i) \cap V(P_j) = \emptyset$, for every pair of integers $i,j$, $i \neq j$. Hence, the connected sets corresponding to $V(X)\setminus \{s,t\}$ forms a $k$-sized partition of $V(G)$, that is, the sets form $k$ pairwise disjoint districts.

%Let $X$ be a path in $H$. 
From the construction of $H$, if there is an unlabeled arc $\eLabel{u,v}$ in $X$, $u \neq s$, then $p$ wins in the district corresponding to $u$. Hence, if there are $k^\star$ unlabeled arcs in $X$ excluding the arc from $s$, then $p$ wins in $k^\star$ districts among the $k$ districts that correspond to the vertices of $X$. Now, we show that any other candidate wins at most $k^\star -1$ districts. For a candidate $c$, let  $label(c)$ denote the set $\{\eLabel{c,1}, \eLabel{c,2} ,\ldots, \eLabel{c,k^\star-1}\}$. Note that for each candidate $c$ there are $k^\star -1$ distinct labels. Since the labels on the arcs of $X$ are distinct, there are at most $k^\star-1$ arcs that are labeled with a label from $label(c)$, for each candidate $c$. From the construction of $H$, if an arc from $z \in V(H)$ is labeled with an element from $label(c)$, then $c$ wins in the district corresponding to the vertex $z$. Hence, each candidate $c \in \Co{C}\setminus \{p\}$ wins in at most $k^\star-1$ districts among $\{P_1,\ldots, P_k\}$ because all the labels are distinct in the path $X$.

($\Leftarrow$) For the reverse direction, suppose that $V(G)$ can be partitioned into $k$ pairwise disjoint districts such that $p$ wins in $k^\star$ districts and any other candidate wins in at most $k^\star -1$ districts. Let  $Y_1, Y_2, \ldots, Y_k$ be the set of these districts  (i.e., each $Y_i$ is a subpath of $G$) such that for every $i \in \{1, \ldots, k-1\}$,  $Y_{i+1}$ begins at the unique  out-neighbor of the last vertex of the subpath $Y_i$  in $G$. 
Moreover, $Y_1$ is a subpath starting at $u_1$ and $Y_k$ is a supath ending at $u_n$. 
%Let $q_1,\ldots,q_k$ be the candidates win in the districts $Y_1,\ldots,Y_k$, respectively. 
Let the vertices in $H$ corresponding to $Y_1, Y_2, \ldots, Y_k$ be $y_1,y_2, \ldots y_k$, respectively. Let $y_{k+1}=t$. 
For each $i\in \{1, \ldots, k\}$ such that $p$ wins in $Y_i$, let $e_i$ be the unique arc in $H$ from $y_i$ to $y_{i+1}$. 
Notice that such arcs are unlabeled and the number of such arcs is $k^{\star}$ as $p$ wins in $k^{\star}$ districts in $\{Y_1,\ldots,Y_k\}$. For any $i\in \{1, \ldots, k\}$ such that the winner in the district $Y_i$ is a candidate $c$ other than $p$, we define the arc $e_i$ from $y_{i}$ to $y_{i+1}$ as follows. Let $j$ be the number of districts won by $c$ in the set of districts $\{Y_1,\ldots,Y_i\}$. Then, $e_i$ is the arc \eLabel{ y_{i},y_{i+1},j} and as the number of districts won by $c$ in $\{Y_1,\ldots,Y_k\}$  is at most $k^{\star}-1$, $e_i$ is well defined. Moreover, $e_i$ is labeled with 
\eLabel{c,j}. No label appear more than once among the arcs $\{e_1,\ldots,e_k\}$. 
Let $e_0$ be the arc \eLabel{s,y_1}. Notice that $e_0$ is an unlabeled arc. From the definition of $e_1,\ldots,e_k$, the number of unlabeled arcs in $\{e_1,\ldots,e_k\}$ is $k^{\star}$ because the number of districts won by $p$ in $\{Y_1,\ldots, Y_k\}$ is $k^{\star}$. Thus, there are $k^{\star}+1$ unlabeled arcs in $\{e_0,\ldots,e_{k+1}\}$. Again by the definition of $e_1,\ldots,e_k$, all the labels of the labeled arcs in $\{e_1,\ldots,e_k\}$ are distinct and the number of labeled arcs is $k-k^{\star}$.  Therefore $(s,e_0,y_1,e_1,y_2, \ldots,e_{k-1},y_k,e_k,t)$ is the required path. 
This completes the proof of the lemma. %\qed
\end{proof}

Thus, our problem reduces to finding a path on $k+2$ vertices from $s$ to $t$ in $H$ such that there are $k^{\star}+1$ unlabeled arcs, and $k -k^{\star}$ labeled arcs with distinct labels.

% !TEX root = main.tex

\subsection{Deterministic Algorithm on Paths}
\label{sec:detpath}

%In this section, we present a deterministic FPT algorithm for \targm, when parameterized by $k$.  That is, we prove the following theorem. 

In this section, we will prove \Cref{thm:detfpt}. Due to \Cref{lem:tgm-gm} 
it is sufficient to prove the following.

\begin{theorem}
\label{thm:dettrgm}
There is an algorithm that given an instance $\mathcal{I}$ of \targm and a tie-breaking rule, runs in time $2.619^{k-k^{\star}} \vert \mathcal{I}\vert ^{\OO(1)}$, and solves the instance $\mathcal{I}$. 
\end{theorem}

Towards proving Theorem~\ref{thm:dettrgm}, we design a dynamic programming algorithm using the concept of  {\em representative family}.

\medskip
\noindent{\bf Why use representative family?}~The method is best explained by applying it to finding a $k$-sized path in a graph between two vertices $s$ and $t$.
% Let $\mathcal{F}(s^\star,t^\star,k^\star)$ be   all possible subsets $X$ of size $k^\star$ such that there is a path between $s^\star$ and $t^\star$ containing all the vertices in $X$.  Observe that there could be two distinct paths corresponding to the same set $X$. In particular,   
%$|\mathcal{F}|\leq {n \choose k^\star}$ while the total number of $k^\star$-sized paths between $s^\star$ and 
%$t^\star$  could be as many as $ {n \choose k^\star} k^\star!$. 
Let $\mathcal{F}$ denote the set of all paths of size $k$ between $s$ and $t$. Observe that 
$ |\mathcal{F}|\leq {n \choose k} k!$. Let ${\sf Prefix}$  denote the subset of vertices  of size  $k/2$  that appear as a prefix on a path of size $k$ between $s$ and $t$.  That is, a set $X$ belongs to ${\sf Prefix}$, if there is a path $P\in \mathcal{F}$ such that $X$ appears among the {\em first} $k/2$ vertices on $P$. Similarly, define the set ${\sf Suffix}$ as the subset of vertices of size $k/2$ that appear as a suffix on a path of size $k$ between $s$ and $t$. That is, a set $X$ belongs to ${\sf Suffix}$, if there is a path $P\in \mathcal{F}$ such that $X$ appears among the {\em last} $k/2$ vertices on $P$. 
Clearly, $|{\sf Prefix}|$ could be ${n \choose k/2}$. A representative set is a subfamily ${\sf Prefix}^\star \subseteq {\sf Prefix} $ such that for every $Q\in {\sf Suffix}$, there is a $P^\star \in {\sf Prefix}^\star $ such that $P^\star \cap Q=\emptyset$. That is, if  there is a path in ${\sf Prefix}$ which along with $Q$ yields a path between $s$ and $t$, then the same holds with the smaller subfamily 
${\sf Prefix}^\star$. One can show that there exists a ${\sf Prefix}^\star$ of size ${k \choose k/2}$ and in fact, this can be computed very efficiently in an iterative fashion. This is the core of the method of representative family. 

We note that the representative family method has led to improvements in %running time if not 
the best known running times for deterministic algorithms for many problems beyond that of finding a path of length $k$, some related problems being {\sc Long Directed Cycle}--Decide whether the input digraph contains a cycle of length at least $k$,  etc.
Representative family improves on color coding based method which uses randomization and dynamic programming, separately.   
%Representative family combines these two techniques 
 We refer to Cygan et.al~\cite{ParamAlgorithms15b} for a detailed exposition. 

%\il{Any other significant problem worth mentioning please write here.}

We begin our  formal discussion by defining representative families~\cite{tcsMarx09,ParamAlgorithms15b} and stating some well-known results. 
% in Section~\ref{subsubsec:repset}. %; and then prove Theorem~\ref{thm:dettrgm} in Section~\ref{subsubsec:detalg}. 
%{\bf Representative Family}\label{subsubsec:repset}
%\begin{defn}\label{def:repset}{\rm \cite{tcsMarx09,ParamAlgorithms15b}}
%\todo[inline]{Add more explanation for Rep set?}
Let $\mathcal{S}$ be a family of subsets of a universe $U$; and let $q\in{\mathbb N}$. 
A subfamily $\widehat{\mathcal{S}} \subseteq \mathcal{S}$ is said to \emph{$q$-represent} $\mathcal{S}$ if the following holds. For every set $B$ of size $q$, if there is a set $A \in \mathcal{S}$ such that $A\cap B=\emptyset$, then there is a set $A' \in \widehat{\mathcal{S}}$  such that $A'\cap B=\emptyset$. If $\widehat{\mathcal{S}}$ $q$-represents $\mathcal{S}$, then we call $\widehat{\mathcal{S}}$ a \emph{$q$-representative of $\mathcal{S}$}.
%\end{defn}

%We use the following results in the design and analysis of our algorithm. 

\begin{proposition}{\rm \cite{fomin2016efficient}}
\label{thm:fastRepUniform}
Let $ \Co{S} = \{S_1,\ldots, S_t\}$ be a family of sets of size $p$ over a universe of size $n$ and let $0<x<1$.
% be a 
%fixed constant. 
For a given $q \in \mathbb{N}$, a $q$-representative family $\widehat{\mathcal{S}}\subseteq {\mathcal{S}}$  for $ {\mathcal{S}}$  with at most  
$ {x^{-p}(1-x)^{-q}} \cdot 2^{o(p+q)} $ sets can be computed in time  
$\OO((1-x)^{-q} \cdot 2^{o(p+q)}\cdot t \cdot \log{n})$.
% There is an algorithm that given a $p$-family ${\mathcal{A}}$ of sets  over a universe $U$ of size $n$,  an integer $q$, 
%and a non-negative weight function    \awf{} with maximum value at most $W$,
%computes in time 
%$\cO(|{\mathcal{A}}| \cdot  (1-x)^{-q}\cdot 2^{o(p+q)} \cdot \log n +|{\mathcal{A}}|\cdot \log |{\mathcal{A}}| \cdot  \log{W} )$ 
%%\todo{should we write the factor $n$}  
%%$\cO(|{\mathcal{A}}|\cdot (\frac{p+q}{q})^q  \cdot  \log n +|{\mathcal{A}}|\cdot \log |{\mathcal{A}}| \cdot  \log{W})$ 
%a subfamily 
%$\widehat{\mathcal{A}}\subseteq \mathcal{A}$ such that 
%$|{\mathcal{A}}'| \leq  {x^{-p}(1-x)^{-q}} \cdot 2^{o(p+q)} $
%and 
% \minrep{A}{q} (\maxrep{A}{q}).
\end{proposition}

We introduce the  definition of subset convolution on set families which will be used to capture the idea of ``extending'' a partial solution, a central concept when using representative family. 

\begin{definition}
For two families of sets $\mathcal{A}$ and $\mathcal{B}$, we define $\mathcal{A} \ast \mathcal{B}$ as \[\{A \cup B  \colon A \in \mathcal{A}, B \in \mathcal{B}, A\cap B = \emptyset\}.\]

\end{definition}
%We use the following properties of representative family in our analysis. 

%\il{Q: Should we move the \Cref{prop:rep-set1,prop:rep-set3} to the Appendix? They are used only in the proof of \Cref{lem:rec-detpath}. So if we move that proof to Appendix, then we can move the propositions to the Appendix, as well.}

\begin{proposition}{\rm \cite[Lemma 12.27]{ParamAlgorithms15b}}\label{prop:rep-set2}
Let ${\mathcal{A}}$ be a family of sets over a universe. If ${\mathcal{A}}'$ $q$-represents $\widehat{\mathcal{A}}$ and $\widehat{\mathcal{A}}$ $q$-represents ${\mathcal{A}}$, then ${\mathcal{A}}'$ $q$-represents ${\mathcal{A}}$. 
\end{proposition}

%\il{Can we move these 2 to Appendix? If so, then remove mentions of them in the ``Proof sketch of Lemma~\Cref{lem:rec-detpath}''.}

\shortversion{
\begin{proposition}[Lemma~12.26~\cite{ParamAlgorithms15b}]\label{prop:rep-set1}
Let $q$ and $h$ be two  non-negative integers. If ${\mathcal{A}}_1$ and ${\mathcal{A}}_2$ are both $h$-families over a universe, ${\mathcal{A}}'_1$ $q$-represents ${\mathcal{A}}_1$ and ${\mathcal{A}}'_2$ $q$-represents ${\mathcal{A}}_2$, then ${\mathcal{A}}'_1 \cup {\mathcal{A}}'_2$ $q$-represents ${\mathcal{A}}_1 \cup {\mathcal{A}}_2$.
\end{proposition}

\begin{proposition}[Lemma 12.28~\cite{ParamAlgorithms15b}]\label{prop:rep-set3}
Let $p_1,p_2,k$ be non-negative integers and $k\geq p_1+p_2$. 
Let ${\mathcal{A}}_1$ be a $p_1$-family and ${\mathcal{A}}_2$ be a $p_2$-family over a universe. Suppose ${\mathcal{A}}'_1$ $(k - p_1)$-represents ${\mathcal{A}}_1$ and ${\mathcal{A}}'_2$ $(k - p_2)$-represents ${\mathcal{A}}_2$. Then, ${\mathcal{A}}'_1 \ast {\mathcal{A}}'_2$  $(k- p_1 - p_2)$-represents ${\mathcal{A}}_1 \ast {\mathcal{A}}_2$.
\end{proposition}
}

\begin{proof}[Proof of Theorem~\ref{thm:dettrgm}.]\label{subsubsec:detalg}
\begin{sloppypar}An instance of \targm is given by  $\Co{I}=\!(G,\Co{C},\{w_v\}_{v \in V},p,k,k^\star)$.
% be an instance.
 %of \targm. 
 Additionally, recall the construction of the labeled digraph $H$ with parallel arcs from $\Co{I}$. 
In order to prove Theorem~\ref{thm:dettrgm}, due to Lemma~\ref{lem:stpath},
it is enough to decide whether there exists a path on $k+2$ vertices from $s$ to $t$ in $H$ that satisfies the following properties: \begin{enumerate*}[label={\bf (P\Roman*)}]\item there are  $k^{\star}+1$ unlabeled arcs, and
\item the remaining $k -k^{\star}$ arcs have distinct labels.
\end{enumerate*} 
\end{sloppypar}
%\hide{Our algorithm is similar to the algorithm for {\sc $k$-Path} from~\cite{fomin2016efficient}, which also uses representative families.} %\ma{Need this here ? Already stated in Our Contribution.}

Before presenting our algorithm, we first define some notations. 
For $i \in\{1, \ldots, k+1\}$ and $r\in \{1, \ldots, k^{\star}+1\}$, a path $P$ starting from $s$ on $i+1$ vertices is said to satisfy $\Prop{i}{r}$ if there are $r$ unlabeled arcs (including the arc from $s$ in $P$), and the remaining $i -r$ arcs have distinct labels. For a subgraph $H'$ of $H$, we denote the set of labels in the graph $H'$ by $\Lset(H')$. Recall that each vertex $v \in V(H)\sm\{s,t\}$ corresponds to a subpath (i.e., a district) of the path graph $G$. Hence, for each $v \in V(H) \setminus \{s,t\}$, we define $\win{v}$ to denote the unique candidate that wins\footnote{We may assume that a unique candidate wins each district because of the application of the tie-breaking rule.} in the district denoted by $v$. Equivalently, we say that the candidate $\win{v}$ \emph{wins} the district $v$ in $G$. For each vertex $v \in V(H)$, and a pair of integers $i \in\{1, \ldots, k+1\}$, $r\in \{1, \ldots, \min\{i, k^{\star}+1\}\}$, we define a set family $\setfamily{F}{i}{r}{v} = $
%\begin{align*}
$\{ P ~\colon \!\!\textit{~$P$ is a $s$ to $v$ path in $H$ on $i\!+\!1$ vertices satisfying } \Prop{i}{r}\}.$
%\end{align*}
 % \textit{where $r$ arcs are unlabeled and the other $i - r$ arcs have distinct labels}\}
 
The following family contains the arc labels on the path in the aforementioned family \setfamily{F}{i}{r}{v}. 
\begin{align*}
& \setfamily{Q}{i}{r}{v} = \big\{\Lset(P) \colon P \in \setfamily{F}{i}{r}{v} \big\} .
\end{align*}

Note that for each value of $i \in \{1, \ldots, k+1\}$, $r $ defined above and $v \in V(H)$, each set in $\setfamily{Q}{i}{r}{v}$ is actually a subset of $\Lset(H)$ of size $i-r$. If there is a path from $s$ to $t$ on $k+2$ vertices with $k-k^\star$ arcs with distinct labels, then $ \setfamily{Q}{k+1}{k^{\star}+1}{t}\neq \emptyset$ and vice versa. 
That is, $\setfamily{Q}{k+1}{k^{\star}+1}{t}\neq \emptyset$ if and only if $\setfamily{F}{k+1}{k^{\star}+1}{t}\neq \emptyset$.  
Hence, to solve our problem, it is sufficient to check if $\setfamily{Q}{k+1}{k^{\star}+1}{t}$ is non-empty.  To decide this, we design a dynamic programming algorithm using representative  families over $\Lset(H)$. %At each step of our dynamic programming routine we use representative families. 
In this algorithm, for each value of $i \in \{1,\ldots, k+1\}$, $r\in\{1,\ldots,\min\{i, k^{\star}+1\}\}$, and $v \in V(H)$, we compute a $(k-k^\star-(i-r))$ representative family of $\setfamily{{Q}}{i}{r}{v}$, denoted by $\setfamily{\widehat{Q}}{i}{r}{v}$, using Proposition~\ref{thm:fastRepUniform}, where $x=\frac{i-r}{2(k-k^{\star})-(i-r)}$. Here, the value of $x$ is set with the goal to optimize the running time of our algorithm, as is the case for the algorithm for {\sc $k$-Path} in~\cite{fomin2016efficient}. At the end our algorithm outputs ``\yes'' if and only if $\setfamily{\widehat{Q}}{k+1}{k^{\star}+1}{t}\neq \emptyset$. 

\medskip
%\smallskip
\noindent
{\bf The ``big picture''.} The big picture behind our algorithm is that essentially we want to decide if $\setfamily{Q}{k+1}{k^{\star}+1}{t} \neq \emptyset$, but computing that as part of the dynamic program would require a table with $(k+1) (k^{\star} +1)n^{2}$ entries and each entry may contain $\OO(n^{i})$, where $i\leq k+1$ ``partial solutions''. Using a $(k-k^\star-(i-r))$ representative family of \setfamily{Q}{i}{r}{v} instead of the family itself allows us to save on the size of each entry significantly to $2^{k-k^{\star}}$ as follows: The set $\setfamily{Q}{i}{r}{v}$ contains the labels of a path from $s$ to $v$ on $i+1$ vertices with $i-r$ distinct labels. For $\setfamily{Q}{k+1}{k^{\star}+1}{t} \neq \emptyset$, there must exist some $i$, $r$ and $v$ such that $\setfamily{Q}{i}{r}{v}$ contains the set of labels that appear on an $s$ to $v$ path, denoted by $P_{0}$, on $i$ vertices with $i-r$ distinct labels. Moreover, there exists a path $P$ from $v$ to $t$ on $k+1-i$ vertices with $k-k^{\star}- (i-r)$ distinct labels such that $(P_{0}, P)$ is an $s$ to $t$ path on $k+1$ vertices with $k-k^{\star}$ distinct labels. From the definition of representative family it follows that there exists a set $S' \in \setfamily{\widehat{Q}}{i}{r}{v}$, where $S'$ is the set of labels on a path, denoted by $P'$  
from $s$ to $v$ on $i+1$ vertices with $i-r$ distinct labels, such that $(P', P)$ is an $s$ to $t$ path on $k+1$ vertices with $k-k^{\star}$ distinct labels. Moreover, the size of \setfamily{\widehat{Q}}{i}{r}{v} is at most ${k-k^{\star}\choose k-k^{\star}- (i-r)}\leq 2^{k-k^{\star}}$.

% over uniform matroid $U_{\Lset(H), k-k^{\star}}$ where the ground set is $\Lset(H)$, and the rank of the matroid is $k-k^{\star}$. 
%That is, a set $A \in Q[i,r,v]$ is independent in $U_{\Lset(H), k-k^{\star}}$ if $|A| \leq k-k^\star$.
%We show that the family of sets we store are enough to recover an independent set in $U_{\Lset(H), k-k^{\star}}$, if a solution exists. 
%We will use need the following proposition to compute representative families.

%
%\begin{proposition}\cite[Lemma12.31]{ParamAlgorithms15b}\label{prop:rep-set-time}
%Let $M$ be a uniform matroid and $\mathcal{A}$ be a $m$-family of independent sets of $M$. There is an algorithm that given $\mathcal{A}$, a rational number $0<x<1$, and an integer $q$, computes a $q$-representative family $\mathcal{A} \subseteq^q_{rep} \mathcal{A}$ of size at most $x^{-m}(1-x)^{-q}2^{o(m+q)}$ in 
% time $|\mathcal{A}|(1-x)^{-q}2^{o(m+q)}$.
% \end{proposition}

\smallskip

\noindent
{\bf Algorithm.} We now formally describe how we recursively compute the family $\setfamily{\widehat{Q}}{i}{r}{v}$, for each $i \in \{1, \ldots, k+1\}$, $r\in \{1, \ldots, \min\{i, k^{\star}+1\}\}$, and $v \in V(H)$. 

\noindent
{\bf Base Case:} We set $\setfamily{\widehat{Q}}{1}{r}{v}=\setfamily{{Q}}{1}{r}{v} $ \begin{align}\label{eq:base-detpath} 
 & = \left\{ \begin{array}{c l}
\{\emptyset \} &  \mbox{ if }  \eLabel{s,v} \text{ is an arc in $H$} \mbox{ and } r=1 \\
\emptyset & \mbox{otherwise}  
\end{array}\right.
\end{align}
For each $i\in\{1, \ldots, k+1\}$, $r\in \{1, \ldots, k-k^{\star}\} \cup \{0\}$, and $v \in V(H)$, we set
\begin{align}\label{eq:base-detpathsecond}
\quad \setfamily{\widehat{Q}}{i}{r}{v}= \setfamily{{Q}}{i}{r}{v}=\emptyset \mbox{ if } r=0 \mbox{ or } r>i. 
\end{align}

%\il{Why are we allowing $r=0$, given that $i\geq 1$. We will only consider paths with 2 or more vertices starting from $s$, hence $r\geq 1$ }

Note that \eqref{eq:base-detpathsecond} is defined so that the recursive definition 
(\Cref{eq:rec-detpath}) has a simple description.

%e a simple equation in the recursive step (i.e., \eqref{eq:rec-detpath}).}\ma{Not clear what this means.}

%\begin{equation} \label{eq:base-detpath}
%\begin{split}
%Q[i,r,v] & = \{\emptyset\}& \textit{ if } (s,v) \in \mathcal{A}(H), i = r = 1\\
%& = \emptyset & \textit{if } (s,v) \notin \mathcal{A}(H), i = r = 1\\
% & = \emptyset & \textit{ if } 1 \leq i \leq k+1, r =0 \textit{ or } r \geq i
%\end{split}
%\end{equation}

%For each $v \in V(H)$, For each $1 \leq i \leq k+1$, $0 \leq r \leq k^\star +1$, we set $\widehat{Q}[i,r,v]^{k-k^\star - (i-r)} = Q[i,r,v]$.

%For two families of sets $\mathcal{A}$ and $\mathcal{B}$, we define the operation $\mathcal{A} \ast \mathcal{B}$ as $\{A \cup B  \colon A \in \mathcal{A}, B \in \mathcal{B}, A\cap B = \emptyset\}$. For every $2\leq i \leq k+1$, $1\leq r \leq \min\{i, k^\star+1\}$, we will compute $\tilde{Q}[i,r,v]$ recursively. Then, we compute  $k-k^\star -(i-r)$-representative of $\tilde{Q}[i,r,v]$. This is denoted by $\widehat{Q}[i,r,v]^{k-k^\star-(i-r)}$. We store this family. Since the value of $k-k^\star - (i-r)$ is well defined by the values of $i$ and $r$, for simplicity of notation, we drop the superscript and write $\widehat{Q}[i,r,v]$.

 \noindent
{\bf Recursive Step:} 
 For each $i \in \{2, \ldots, k+1\}$, $r \in \{1, \ldots, \min\{i, k^\star+1\}\}$, and $v \in V(H)$,  we compute 
 $\setfamily{\widehat{Q}}{i}{r}{v}$ as follows. We first compute  $\setfamily{{Q'}}{i}{r}{v}$ 
 from the previously computed families and then  we compute a $(k-k^{\star}-(i-r))$-representative family $\setfamily{\widehat{Q}}{i}{r}{v}$ of $\setfamily{{Q'}}{i}{r}{v}$. The family $\setfamily{{Q'}}{i}{r}{v}$ is computed using the representative family: %as follows. 

\vspace{-1em}
\begin{align} \label{eq:rec-detpath}
\nonumber
 \setfamily{{Q'}}{i}{r}{v} = &
\bigg(\bigcup_{ \substack{w \in N^{-}(v),\\ \win{w}=p}}    \setfamily{\widehat{Q}}{i-1}{r-1}{w} \bigg)  \bigcup \\ 
\bigg(\! \bigcup_{\substack{ w \in N^{-}(v),\\ \win{w}\neq p}} & \!\!\setfamily{\widehat{Q}}{i-1}{r}{w} \ast \{\{\eLabel{\win{w},j}\}: 1\!
\leq j< \!k^{\star}\} \bigg)  %\hspace{-4em}
\end{align}
%\vspace{-2em}

%\il{}
%\begin{align}
%\nonumber \setfamily{{Q'}}{i}{r}{v} & =
%(\cup_{ \substack{w \in N^{-}(v),\\ \win{w}=p}}    \setfamily{\widehat{Q}}{i-1}{r-1}{w} )    
% \bigcup  
% (\!\cup_{\substack{ w \in N^{-}(v),\\ \win{w}\neq p}} \setfamily{\widehat{Q}}{i-1}{r}{w} \ast \{\{\eLabel{\win{w},j}\}: 1
%\leq j< \!k^{\star}\}) 
%\end{align}

% $\tilde{Q}[i,r,v]$.
 
% Let $(w,v) \in \mathcal{A}(H)$ and $\win{w} \neq p$. We define $\widehat{J}(w) = \{  j \in [k^\star -1] \colon  \textit{ for all } A \in \widehat{Q}[i-1,r,w], \eLabel{\win{w},j} \notin  A\}$. If $\widehat{J}(w) \neq \emptyset$, then we fix an arbitraty $j \in \widehat{J}(w)$ and define $\widehat{Q}'_{w,j} = \{ A \in \widehat{Q}[i-1,r,w] \colon \eLabel{\win{w},j} \notin  A\}$.
 
%\begin{equation} \label{eq:rec-detpath}
%\begin{split}
%\setfamily{{Q'}}{i}{r}{v} & =
% \bigcup_{ \substack{(w,v) \in \mathcal{A}(H),\\ \win{w}=p}}    \setfamily{{Q'}}{i-1}{r-1}{w}      
%& \cup \bigcup_{\substack{ (w,v) \in \mathcal{A}(H),\\ \win{w}\neq p,\\ \widehat{J} \neq \emptyset}} \widehat{Q}'_{w,j} \ast \{\{\eLabel{\win{w},j}\}\}  \\
%& &\textit{ if } 2 \leq i \leq k+1, 1\leq r \leq \min\{i, k^\star+1\}\\\[1ex]
%& = \bigcup_{ \substack{(w,v) \in \mathcal{A}(H),\\ \win{w}=p}} \widehat{Q}[i-1,i-1,w] 
%& \textit{ if } 2 \leq i \leq k^\star+1, r = i &.
%\end{split}
%\end{equation}

Next, we compute a $(k-k^{\star}-(i-r))$-representative family $\setfamily{\widehat{Q}}{i}{r}{v}$ of $\setfamily{{Q'}}{i}{r}{v}$ using Proposition~\ref{thm:fastRepUniform}, where $x=\frac{i-r}{2(k-k^{\star})-(i-r)}$.  
Our algorithm (call it ${\mathscr A}$) to decide if the desired $s-t$ path exists in $H$ works as follows:  we compute 
$\setfamily{\widehat{Q}}{i}{r}{v}$ 
using  Equations~\eqref{eq:base-detpath}--\eqref{eq:rec-detpath}, and Proposition~\ref{thm:fastRepUniform}. At the end  ${\mathscr A}$ outputs ``\yes'' if and only if $\setfamily{\widehat{Q}}{k+1}{k^{\star}+1}{t}\neq \emptyset$. 
%
%\noindent
%{\bf Algorithm:} To summirize the algorithm executes the following steps.
%\begin{enumerate}[label=(Step \Roman*),wide=0pt]
%\item Compute the base case using Equation~\ref{eq:base-detpath}, and set $\widehat{Q}[i,r,v] = Q[i,r,v]$, for each of the entries $Q[i,r,v]$ computed in Equation~\ref{eq:base-detpath}.
%\item For each $2\leq i \leq k+1$, $1\leq r \leq \min\{i, k^\star+1\}$, and $v \in V(H)$,  compute $\tilde{Q}[i,r,v]$.
%
%\item Compute $\widehat{Q}[i,r,v]$, a $k-k^\star - (i-r)$ representative of $\tilde{Q}[i,r,v]$ over $U_{\Lset(H), k-k^\star}$.
%
%\item if $\widehat{Q}[k+1,k^\star+1,t] \neq \emptyset$ , then return Yes, otherwise return No.
%\end{enumerate}
%
%\subparagraph*

\smallskip
\noindent 
{\bf Correctness proof and running time analysis.} 
We prove that  for every $ i \in \{1, \ldots, k+1\}$, $r\in \{1, \ldots, \min\{i, k^{\star}+1\}\}$, and $v \in V(H)$,  $\setfamily{\widehat{Q}}{i}{r}{v}$ is indeed a $(k-k^{\star} - (i-r))$ representative family of $\setfamily{Q}{i}{r}{v}$, and not just that of $\setfamily{Q'}{i}{r}{v}$. From the definition of $0$-representative family of  $\setfamily{{Q}}{k+1}{k^{\star}+1}{t}$, we have that  $\setfamily{{Q}}{k+1}{k^{\star}+1}{t}\neq \emptyset$ if and only if $\setfamily{\widehat{Q}}{k+1}{k^{\star}+1}{t}\neq \emptyset$. Thus, to prove the correctness of the algorithm it is enough to prove the following. 

\begin{lemma}%[$\clubsuit$]
\label{lem:rec-detpath}For each $ i \in \{1, \ldots, k+1\}$, $r\in \{1, \ldots, \min\{i, k^{\star}+1\}\}$, and $v \in V(H)$, family
$\setfamily{\widehat{Q}}{i}{r}{v}$ is a $(k-k^\star-(i-r))$-representative of $\setfamily{{Q}}{i}{r}{v}$.
\end{lemma}

%
%\begin{lemma}\label{lem:rec-detpath}
%For each $2\leq i \leq k+1$, $1\leq r \leq k^{\star}+1$ and $v \in V(H)$, Equation~\ref{eq:rec-detpath} correctly computes a $(k-k^\star-(i-r))$-representative of $Q[i,r,v]$. \todo{remove or modify}
%\end{lemma}

To prove Lemma~\ref{lem:rec-detpath}, we first prove that the following recurrence for $\setfamily{{Q}}{i}{r}{v}$ is correct.  
 \begin{align} \label{eq:rec-detpath1}
\nonumber
 \setfamily{{Q}}{i}{r}{v} & =
 \bigg(\bigcup_{ \substack{w \in N^{-}(v),\\ \win{w}=p}}  \!\!\!  \setfamily{{Q}}{i-1}{r-1}{w} \bigg)     
 \bigcup \\ \bigg(\bigcup_{\substack{ w \in N^{-}(v),\\ \win{w}\neq p}} & \!\!\! \setfamily{{Q}}{i-1}{r}{w} \ast \{\{\eLabel{\win{w},j}\}:1\leq j< k^{\star}\}\bigg)  
\end{align}
 
%\begin{equation} \label{eq:rec-detpath1}
%\begin{split}
%Q[i,r,v] & =
% \bigcup_{ \substack{(w,v) \in \mathcal{A}(H),\\ \win{w}=p}}Q[i-1,r-1,w] 
%& \cup \bigcup_{\substack{ (w,v) \in \mathcal{A}(H),\\ \win{w}\neq p,\\ J \neq \emptyset}} Q'_{w,j} \ast \{\{\eLabel{\win{w},j}\}\}  \\
%& &\textit{ if } 2 \leq i \leq k+1, 1 \leq r \leq \min\{i, k^\star+1\}\\\[1ex]
%& = \bigcup_{ \substack{(w,v) \in \mathcal{A}(H),\\ \win{w}=p}} Q[i-1,i-1,w] 
%& \textit{ if } 2 \leq i \leq k^\star+1, r = i &.
%\end{split}
%\end{equation}

%The following claim concludes the proof of \Cref{lem:rec-detpath}, and can be proved by showing subset containment on both sides.
 
\begin{claim}\label{lem:rec-detpath1}Equations \eqref{eq:base-detpath}, 
\eqref{eq:base-detpathsecond}, and \eqref{eq:rec-detpath1} correctly compute $\setfamily{{Q}}{i}{r}{v}$, for each $ i \in \{1, \ldots, k+1\}$, $r\in \{1, \ldots, \min\{i, k^{\star}+1\}\}$, and $v \in V(H)$. 
\end{claim}

%\shortversion{
\begin{proof}
Recall that $\setfamily{Q}{i}{r}{v} = \{\Lset(P) \colon P \in \setfamily{F}{i}{r}{v}\}$. 
We prove the lemma using induction on $i$. 
When $i =1$ and $r=1$ (the base case), we are looking for paths on $2$ vertices from $s$ to $v$ with no labeled arcs.
% i.e., $(s,v) \in A(H)$. 
  Moreover, notice that all the arcs incident with $s$ are unlabeled. Hence, for $i =1$ and $r=1$, \eqref{eq:base-detpath} correctly computes $\setfamily{{Q}}{1}{1}{v}$ for all $v\in V(H)$.  Also note that when $r>i$ or ($i=1$ and $r=0$), 
$\setfamily{\widehat{Q}}{i}{r}{v}=\setfamily{{Q}}{i}{r}{v}=\emptyset$ for any $v\in V(H)$.
%the corresponding sets are empty. 

%Since $r$ is the number of unlabeled arcs, $r\leq i$. Hence, for each $i \in[k+1]$, if $r \geq i$, then $Q[i,r,v]$ must be $\emptyset$, since there is no such path. Since the arc from $s$ is unlabeled, $r\geq 1$. Hence, for each $i \in [k+1]$, if $r =0$, then $Q[i,r,v]$ must be $\emptyset$. 
Now, consider the induction step. 
For $ i \in \{2,\ldots, k+1\}$, $r\in \{1 \ldots, \min\{ i,k^\star +1\}\}$, and $v\in V(H)$, we compute $\setfamily{{Q}}{i}{r}{v}$ using \eqref{eq:rec-detpath1}. We show that the recursive formula is correct.
By induction hypothesis, we assume that for each $i' < i$, 
$r'\in \{1, \ldots, \min\{ i,k^\star +1\}\}$,  and $v \in V(H)$, \eqref{eq:base-detpath},\eqref{eq:base-detpathsecond},  and \eqref{eq:rec-detpath1} correctly computed $\setfamily{{Q}}{i'}{r'}{v}$. 
%Moreover, $\setfamily{{Q}}{i'}{r'}{v}=\emptyset$

First, we show that $\setfamily{{Q}}{i}{r}{v}$ is a subset of the R.H.S.  of  \eqref{eq:rec-detpath1}. 
Recall that $\setfamily{{Q}}{i}{r}{v}$ contains label sets of paths from $s$ to $v$ on $i+1$ vertices satisfying $\Prop{i}{r}$. Let $X\in \setfamily{{Q}}{i}{r}{v}$. Then, there exists 
a path $P\in \setfamily{F}{i}{r}{v}$ and  $X=\Lset(P)$. That is, 
$P$ is  a path  from $s$ to $v$ on $i+1$ vertices satisfying $\Prop{i}{r}$ 
%That is, $P\in \setfamily{F}{i}{r}{v}$ and  
and $X=\Lset(P) \in \setfamily{Q}{i}{r}{v}$.
%Let the in-neighbor of $v$ in $P$ be $w$. Let $P-v$ denote the path we get after the vertex $v$ from $P$. 
Let $P$ be denoted by $(s,e_0,v_1,e_1,v_2,\ldots,v_{i-1},e_{i-1},v)$. Let the subpath $(s,e_0,v_1,e_1,v_2, \ldots, e_{i-2},v_{i-1})$ be denoted by $P-v$.  Since $P$ satisfy $\Prop{i}{r}$, there are exactly $r$ unlabeled arcs including the arc from $s$. Recall that due to construction of $H$, there is an unlabeled arc from a vertex $u \in V(H)\setminus \{s,t\}$ if $\win{u}=p$. Therefore, there are $r-1$ vertices in $\{v_1, v_2, \ldots, v_{i-1}\}$ where the target candidate $p$ wins.

\medskip
\noindent 
 {\bf Case 1:} Suppose that $\win{v_{i-1}} = p$. 
 Therefore, the arc $e_{i-1}$ is unlabeled. Hence, $P-v$ has $r-1$ unlabeled arcs, and it is a path on $i$ vertices from $s$ to $v_{i-1}$.
Therefore, $P-v$ satisfy $\Prop{i-1}{r-1}$. Hence, $\Lset(P-v)$ must be in 
$\setfamily{{Q}}{i-1}{r-1}{v_{i-1}}$. Moreover, $\Lset(P-v)=\Lset(P)$ and 
$\setfamily{{Q}}{i-1}{r-1}{v_{i-1}}$ is a subset of the R.H.S. of \eqref{eq:rec-detpath1}. 
This implies that $\Lset(P)$ is a subset of the R.H.S. of \eqref{eq:rec-detpath1}.

%$Q[i-1,r-1,v_{i-1}]$. 
%Note that if $r = i$, then all the arcs of $P$ must be unlabeled, since $P$ has $i$ arcs. Precisely, $(v_{i-1},v)$ must be unlabeled. Hence, the next case  does not occur when $r=i$.

\smallskip
\noindent
{\bf Case 2:} Suppose that $\win{v_{i-1}} \neq p$. 
Therefore, the arc $e_{i-1}$ is labeled with $\eLabel{\win{v_{i-1}}, j}$, where $j \in \{1, \ldots, k^{\star} -1\}$. Since the arcs of $P$ has distinct labels, $\eLabel{\win{v_{i-1}}, j} \notin \Lset(P-v)$. That is, $|\Lset(P-v)| = |\Lset(P)| -1$. Recall that $P$ is a path on $i+1$ vertices from $s$ to $v$ and $P-v$ is a path on $i$ vertices from $s$ to $v_{i-1}$. Therefore, the number of unlabeled arcs in the path $P-v$ is the same as in $P$. That is $P-v$ has $r$ unlabeled arcs.
Hence, $P-v$ satisfy $\Prop{i-1}{r}$ implying that $\Lset(P-v) \in \setfamily{{Q}}{i-1}{r}{v_{i-1}}$.
Hence, $\Lset(P)$ is in R.H.S of \eqref{eq:rec-detpath1}. 

For the other direction, we show that R.H.S of \eqref{eq:rec-detpath1} is a subset of $\setfamily{{Q}}{i}{r}{v}$. 
Let $X$ be a set that belongs to the R.H.S. of \eqref{eq:rec-detpath1}. Since it is a disjoint union, we have the following two cases.

\medskip
\noindent
{\bf Case 1: $X\in \bigcup_{ \substack{w \in N^{-}(v),\\ \win{w}=p}}    \setfamily{{Q}}{i-1}{r-1}{w}$.} That is, there exists $w\in V(H)$ and a path $P\in  \setfamily{{F}}{i-1}{r-1}{w}$ such that $X=\Lset(P)$, $w \in N^{-}(v)$, and $\win{w}=p$. Let $e$ be the unique arc in $H$ from $w$ to $v$ (because $\win{w}=p$). 
Let $P'=(P,e,v)$ denote the path obtained  by connecting $w$ to $v$ in $P$ using the arc $e$. Due to Observation~\ref{obs:H-is-DAG}, $P'$ is a path in $H$. 
Since $P \in \setfamily{{F}}{i-1}{r-1}{w}$, $P$ satisfy $\Prop{i-1}{r-1}$. Hence, $P$ has $r-1$ unlabeled arcs. Therefore, there are $r$ unlabeled arcs in $P'=(P,e,v)$. 
Note that $P'$ is a path on $i+1$ vertices from $s$ to $v$.   Hence, $P'$ satisfy $\Prop{i}{r}$.
Therefore, $P'\in \setfamily{{F}}{i}{r}{v}$. Since $e$ is an unlabeled arc, $\Lset(P')=\Lset(P)=X$. 
Hence, $X=\Lset(P')\in \setfamily{{Q}}{i}{r}{v}$

\smallskip
\noindent
{\bf Case 2: $X\in \bigcup_{\substack{ w \in N^{-}(v),\\ \win{w}\neq p}} \setfamily{{Q}}{i-1}{r}{w} \ast \{\{\eLabel{\win{w},j}\}\colon j\in \{1, \ldots, k^{\star}-1\}\}$.}
That is, there exist $w\in V(H)$ and a path $P\in  \setfamily{{F}}{i-1}{r}{w}$ such that 
$w \in N^{-}(v)$, $\win{w}\neq p$, and $X\in \{\Lset(P)\}\ast \{\{\eLabel{\win{w},j}\}\colon j\in \{1, \ldots, k^{\star}-1\}\}$. 
That is, there exits $j\in \{1, \ldots, k^{\star}-1\}$ such that $\eLabel{\win{w},j}\notin \Lset(P)$ and $X=\Lset(P)\cup \{\eLabel{\win{w},j}\}$. 
Let $e$ be the  arc in $H$ from $w$ to $v$ that is labeled with $\eLabel{\win{w},j}$. 
Let $P'=(P,e,v)$ denote the path obtained  by connecting $w$ to $v$ in $P$ using the arc $e$. 
Due to Observation~\ref{obs:H-is-DAG}, $P'$ is a path in $H$. 
Since $P \in \setfamily{{F}}{i-1}{r}{w}$, $P$ satisfy $\Prop{i-1}{r}$. Hence, $P$ has $r$ unlabeled arcs. Therefore, there are $r$ unlabeled arcs in $P'=(P,e,v)$. 
Note that $P'$ is a path on $i+1$ vertices from $s$ to $v$ 
and $\eLabel{\win{w},j}\notin \Lset(P)$. Hence, $P'$ satisfy $\Prop{i}{r}$.
Therefore, $P'\in \setfamily{{F}}{i}{r}{v}$.  Also, 
since $\Lset(P')=\Lset(P)\cup \{\eLabel{\win{w},j}\}=X$, 
we have that $X=\Lset(P')\in \setfamily{{Q}}{i}{r}{v}$ 

This competes the proof of the claim. %\qed
\end{proof}

%\subsection{Proof of Lemma~\ref{lem:rec-detpath}}
Next, we state some of the properties of representative family in order to prove Lemma~\ref{lem:rec-detpath}.
\begin{proposition}{\rm \cite[Lemma~12.26]{ParamAlgorithms15b}}\label{prop:rep-set1}
If $A_1$ and $A_2$ are both 
$m$-families, $A'_1$ $q$-represents $A_1$ and $A'_2$ $q$-represents $A2$, then $A'_1 \cup A'_2$ $q$-represents $A_1 \cup A_2$.
\end{proposition}

%\begin{proposition}\cite[Lemma 12.27]{ParamAlgorithms15b}\label{prop:rep-set2}
%Let $A$ be a family of sets. If $A'$ $q$-represents $\widehat{A}$ and $\widehat{A}$ $q$-represents $A$, then $A'$ $q$- represents $A$.
%\end{proposition}

\begin{proposition}{\rm \cite[Lemma~12.28]{ParamAlgorithms15b}}\label{prop:rep-set3}
Let $A_1$ be a $p_1$-family and $A_2$ be a $p_2$-family. Suppose $A'_1$ $(k - p_1)$-represents $A_1$ and $A'_2$ $(k - p_2)$-represents $A_2$. Then $A'_1 \ast A'_2$  $(k- p_1 - p_2)$-represents $A_1 \ast A_2$.
\end{proposition}

\begin{proof}[Proof of Lemma~\ref{lem:rec-detpath}]
We prove the lemma using induction on $i$. Base case is given by $i=1$. When $i=1$, $r=1$. 
By \eqref{eq:base-detpath}, $\setfamily{\widehat{Q}}{1}{1}{v}=\setfamily{{Q}}{1}{1}{v}$ for all $v\in V(H)$. Thus,  $\setfamily{\widehat{Q}}{1}{1}{v}$ is a $(k-k^{\star})$-representative of $\setfamily{{Q}}{1}{1}{v}$ for any $v\in V(H)$. 
Notice that  when $r>i$ or ($i=1$ and $r=0$), $\setfamily{\widehat{Q}}{i}{r}{v}=\setfamily{{Q}}{i}{r}{v}=\emptyset$ for any $v\in V(H)$.

Now consider the induction step. That is, $i>1$. By induction hypothesis, we have that for any 
$r'\in \{1,\ldots,\min \{i-1,k^{\star}+1\}\}$ and any $v\in V(H)$, $\setfamily{\widehat{Q}}{i-1}{r'}{v}$ is a $(k-k^{\star})-(i-1-r')$-representative of $\setfamily{{Q}}{i-1}{r'}{v}$. Next, we fix an arbitrary $r\in \{1,\ldots,\min\{i,k^{\star}+1\}\}$ and $v\in V(H)$, and prove that  $\setfamily{\widehat{Q}}{i}{r}{v}$ is a $((k-k^{\star})-(i-r))$-representative of $\setfamily{{Q}}{i}{r}{v}$. 

Consider \Cref{eq:rec-detpath}. Let $i'=i-1$ and $r'=r-1$.  By induction hypothesis and Proposition~\ref{prop:rep-set1}, 

\begin{itemize}
\item[$(a)$] $\bigcup_{ \substack{w \in N^{-}(v),\\ \win{w}=p}}    \setfamily{\widehat{Q}}{i'}{r'}{w}$ is $((k-k^{\star})-(i-r))$-representative of $\bigcup_{ \substack{w \in N^{-}(v),\\ \win{w}=p}}    \setfamily{{Q}}{i'}{r'}{w}$. 
\end{itemize}

By induction hypothesis, Proposition~\ref{prop:rep-set1} and Proposition~\ref{prop:rep-set3}, 

\begin{itemize}
\item[$(b)$] $\bigcup_{\substack{ w \in N^{-}(v),\\ \win{w}\neq p}} \setfamily{\widehat{Q}}{i'}{r}{w} \ast \{\{\eLabel{\win{w},j}\}:1\leq j< \!k^{\star}\}$ is $(k-k^{\star})-(i-r)$-representative of \\$\bigcup_{\substack{ w \in N^{-}(v),\\ \win{w}\neq p}} \setfamily{{Q}}{i'}{r}{w} \ast \{\{\eLabel{\win{w},j}\}:1\leq j< \!k^{\star}\}$
\end{itemize}

%By Lemma~\ref{lem:rec-detpath1}, we know that, $(c)$ \eqref{eq:rec-detpath1} is correct. 
Statements $(a)$ and $(b)$, Lemma~\ref{lem:rec-detpath1}, Proposition~\ref{prop:rep-set1}, and \Cref{eq:rec-detpath} imply that  $\setfamily{{Q}'}{i}{r}{v}$ is a $((k-k^{\star})-(i-r))$-representative of $\setfamily{{Q}}{i}{r}{v}$.
By construction, we have that $\setfamily{\widehat{Q}}{i}{r}{v}$ is a $((k-k^{\star})-(i-r))$-representative of $\setfamily{{Q}'}{i}{r}{v}$. 
Thus, by Proposition~\ref{prop:rep-set2}, 
$\setfamily{\widehat{Q}}{i}{r}{v}$ is a $((k-k^{\star})-(i-r))$-representative of $\setfamily{{Q}}{i}{r}{v}$. 
This completes the proof. 
\end{proof}

Next we analyse the  running time of algorithm ${\mathscr A}$.  %The proof is by induction on $i$. Using [8, Claim 12.34] and standard calculations we get the following lemma.

%There is an algorithm that given $H$, $k$ and $k^\star$ computes for every vertex $v \in V(H)$ and integers $1\leq i \leq k+1$ and $0 \leq r \leq k^\star$, a family $\widehat{Q}[i,r,v]$ that$(k-k^\star-(i-r))$-represents $Q[i,r,v]$. 

\begin{lemma}\label{lem:time-detpath}
For each $i\in \{1, \ldots, k+1\}$, $r\in\{1, \ldots, \min\{i, k^{\star}+1\}\}$, and $v\in V(H)$, the cardinality of $\setfamily{\widehat{Q}}{i}{r}{v}$ is at most 
%  \begin{align*}
%& &\left(\frac{d+2q}{d}\right)^{d} \left(\frac{d+2q}{2q}\right)^{q+1}2^{o(k-k^\star)}
%\end{align*}
  \begin{align*}
 &\left(\frac{2(k-k^{\star})-(i-r)}{i-r}\right)^{i-r}\! \left(\frac{2(k-k^{\star})-(i-r)}{2((k-k^{\star})-(i-r))}\right)^{(k-k^{\star})-(i-r)}\!\!2^{o(k-k^\star)}
\end{align*}

where $d = i-r$ and $q = k-k^\star - (i-r)$, 
and the algorithm ${\mathscr A}$ takes time $2.619^{k-k^\star} n^{\OO(1)}$.
\end{lemma}

%\shortversion{
\begin{proof}
We prove the lemma using induction on $i$. 
 First, the algorithm computes the base case as described in \Cref{eq:base-detpath}. 
 %For $i =1, r = 1$ and 
 By \Cref{eq:base-detpath} the cardinality of $\setfamily{\widehat{Q}}{1}{1}{v}$ is at most $1$ for 
 each $v \in V(H)$ and their computation takes  polynomial time. Also note that when $r>i$ or ($i=1$ and $r=0$), $\vert \setfamily{\widehat{Q}}{i}{r}{v}\vert=0$.
 
 Now we fix  integers $i\in \{2,\ldots k+1\}$ and $r\in \{1,\ldots,\min \{i,k^{\star}+1\}\}$, and $v\in V(H)$. Next,  we compute the cardinality of $\setfamily{\widehat{Q}}{i}{r}{v}$ and the time taken to compute it. 
 Let $d=i-r$ and $q=(k-k^{\star})-(i-r)$. 
 For any two positive integers $a$ and $b$, let $x_{a,b}=\frac{a}{a+2b}$ and $s_{a,b}=(x_{a,b})^{-a}(1-x_{a,b})^{-b}$.  By Proposition~\ref{thm:fastRepUniform}, the cardinality of $\setfamily{\widehat{Q}}{i}{r}{v}$ is at most 
 $x_{d,q}^{-d}(1-x_{d,q})^{-q} \cdot 2^{o(d+q)}$. By substituting the values for $d,q$ and $x_{d,q}$, we have that  $ \vert \setfamily{\widehat{Q}}{i}{r}{v}\vert$
 
 %where $x=\frac{i-r}{2(k-k^{\star})-(i-r)}$. That is, 
 
 \begin{eqnarray*}
 \leq \left(\frac{2(k-k^{\star})-(i-r)}{i-r}\right)^{i-r}\!\!\left(\frac{2(k-k^{\star})-(i-r)}{2((k-k^{\star})-(i-r))}\right)^{(k-k^{\star})-(i-r)}2^{o(k-k^\star)}
 \end{eqnarray*}
 Next, we compute the running time to compute $\setfamily{\widehat{Q}}{i}{r}{v}$. Towards that we first need to bound the cardinality of $\setfamily{{Q}'}{i}{r}{v}$. By \Cref{eq:rec-detpath} and induction hypothesis, 
 the cardinality of $\setfamily{{Q}'}{i}{r}{v}$ is bounded by  $(s_{d,q}+s_{d-1,q+1}) \cdot 2^{o(k-k^{\star})}n^2$.  
 
\begin{claim}\cite[Claim~12.34]{ParamAlgorithms15b}\label{claimLspq}
For any $d\geq 3$ and $q\geq 1$, $s_{d-1,q+1} \leq e^2\cdot d \cdot s_{d,q}$.
 \end{claim}
 
 Thus, by \Cref{claimLspq}, when $d\geq 3$, $\vert \setfamily{{Q}'}{i}{r}{v}\vert \leq s_{d,q} \cdot 2^{o(k-k^{\star})}n^2$. Then, by Proposition~\ref{thm:fastRepUniform}, 
 when $d\geq 3$, the running time to compute $\setfamily{\widehat{Q}}{i}{r}{v}$ is upper bounded by 
 
 \begin{equation}
 \label{eqnruntime}
 s_{d,q} (1-x_{d,q})^{-q} \cdot 2^{o(d+q)} n^{\OO(1)} \leq \left(\frac{d+2q}{d}\right)^{d} \left(\frac{d+2q}{2q}\right)^{2q}2^{o(k-k^\star)} n^{\OO(1)}
 \end{equation}
 When $d\leq 3$, by Proposition~\ref{thm:fastRepUniform}, the running time to compute $\setfamily{{Q}'}{i}{r}{v}$ is $n^{\OO(1)}$. 
 
 Now for the total running time of the algorithm, the value for $d$ and $q$ in \eqref{eqnruntime}
varies as follows:   $0\leq d\leq k-k^{\star}$ and $q=k-k^{\star}-d$. The 
R.H.S. of \Cref{eqnruntime} is maximized when $d = (1-\frac{1}{\sqrt{5}})(k-k^\star)$, and it is upper bounded by  
 $2.619^{k-k^\star} n^{\OO(1)}$. Therefore, the total running time of the algorithm is $2.619^{k-k^\star} n^{\OO(1)}$. 
 \end{proof}

Thus, Theorem~\ref{thm:dettrgm} is proved.
\end{proof}
%Hence we proved the following theorem.
%
%\begin{theorem}\label{th:det-path}
%There is a $2.619^{k-k^\star}n^{\OO(1)}$ time deterministic algorithm for \targm when the input graph $G$ is a path.
%\end{theorem}

%\input{multiwinner-district}
% !TEX root = main.tex

\subsection{Randomized Algorithm on Paths}\label{sec:rand-fpt}

\
In this section, we will prove \Cref{thm:ranfpt}. The randomized algorithm works by detecting the existence of the desired path in the auxiliary graph by interpreting each of the labeled paths in the graph as a multivariate monomial, and then using a result by Williams~\cite{Williams09} to detect a multilinear monomial in the resulting (multivariate) polynomial. The underlying idea being that each path with the desired properties is a multilinear monomial in the polynomial thus constructed, and vice versa. Williams~\cite{Williams09} gave an algorithm with one sided error that allows us to detect a multilinear monomial in time $\OO^{\star}(2^{d})$, where $d$ denotes the degree of the multivariate polynomial\footnote{$\OO^{\star}()$ hides factors that are polynomial in the input size.}. 
%\todo{To fit in the space, do we just keep a intuitive idea here? }
 Due to \Cref{lem:tgm-gm} 
it is sufficient to prove the following. % result.\todo{stop sentence after following?} 

\begin{theorem}\label{lemma:rand-algo}%Let $s(n)$ denote the function that bounds the size of an arithmetic circuit that represents a multivariate polynomial on $n$ variables and degree at most $k$. Then, 
There is a one-sided error randomized algorithm that given an instance $\mathcal{I}$ of \targm, runs in time $\OO^{\star}(2^{k-k^{\star}})$, outputs
 ``yes'' with high probability (at least $2/3$) if $\mathcal{I}$ is a \yes-instance, and always outputs ``no'' if $\mathcal{I}$ is a \no-instance.
%with false negatives that solves \targm on paths  that runs in time $\OO^{\star}(2^{k-k^{\star}})$ time. %, where $m$ denotes the number of candidates. % and $s(n)$ denotes the function that bounds the size of an arithmetic circuit that represents a multivariate polynomial on $n$ variables. %and degree at most 
\end{theorem}

%Recall the definition of $\setfamily{F}{i}{r}{v}$ and $\setfamily{Q}{i}{r}{v}$  from Section~\ref{subsubsec:detalg}. 

\begin{supress}
\il{}
\begin{multline*}
\setfamily{F}{i}{r}{v} = \{ P : \text{$P$ is a path from $s$ to $v$ on $i+1$ vertices,} \\ \hspace{2cm} \text{where $r$ arcs are unlabeled and the other $i - r$ arcs have distinct labels}\}, %\text{ and} 
 \end{multline*}
%\begin{align*}
%\mbox{and} 
and 
%\qquad & &
$\setfamily{Q}{i}{r}{v} = \{ \Co{L}(P): V(P) \in \setfamily{F}{i}{r}{v}\}$,  where 
%\end{align*}
$\Co{L}(P)$ denotes the set of labels that appear in the path $P$.
\il{}
\end{supress}

Recall the definition of $\setfamily{F}{i}{r}{v}$ and $\setfamily{Q}{i}{r}{v}$  from Section~\ref{subsubsec:detalg} and that $\setfamily{F}{k+1}{k^{\star}+1}{t}\neq \emptyset$ if and only if $\setfamily{Q}{k+1}{k^{\star}+1}{t} \neq \emptyset$. 
In essence, our randomized algorithm will decide if for the given instance of \targm, the family $\setfamily{Q}{k+1}{k^{\star}+1}{t} \neq \emptyset$. That in conjunction with \Cref{lem:stpath} will allow us to decide if the given instance is a \yes-instance.

\medskip
\noindent{\bf Interpreting \setfamily{Q}{i}{r}{v} as a multivariate polynomial:} 
Note that there are exactly $(k^{\star}-1)(m-1)$ distinct arc labels in $H$, where $m$ denotes the number of candidates. We will associate each of these labels (denoted by \eLabel{c, j} for some candidate $c\in \Co{C}\setminus \{p\}$ and $j \in \{1, \ldots, k^{\star}-1\}$) with a distinct variable; and use $\Co{\widehat{L}}(P)$ to denote the monomial associated with the labels appearing in the path $P$ in $H$. We assume that the unlabeled arcs contribute $1$ to the monomial. 

For any $i\in \{1, \ldots, k+1\}$, $r\in \{1, \ldots, \min\{i, k^{\star} +1\}\}$, and vertex $v\in V(H)\sm\{s\}$, we will define a polynomial function $\polynf[i,r,v]$ that will contain all the monomials that correspond to each of the paths in $\setfamily{F}{i}{r}{v}$. Towards that, we define a {\it helper} function $\polynf'[i,r, v]$ as follows. For any $v\in V(H) \setminus \{s\}$,  $ i\in \{1, \ldots, k+1\}$, and $r\in \{0, \ldots, \min\{i, k^{\star} +1\}\}$ such that $r>i$ or $r\leq 1$, we define
\begin{align} 
\label{tag:function1base}  
 \polynfpath[i,r,v] = \left\{ \begin{array}{c l}
1 & \mbox{if }  \langle s,v\rangle \text{ is an arc in $H$} \mbox{ and } i=r=1 \\
0 & \mbox{otherwise}  
\end{array}\right.
\end{align}

%\il{For both \polynf and \polynfpath, why don't we define base case just as $\polynf[1,1,v]$ instead of splitting into 2 cases like above. We will never use \polynf[0,0,v] }

\begin{align}
 \polynfpath[i,r, v] ~ = \sum_{\substack{\text{Path }P :\\ \Co{L}(P) \in \setfamily{Q}{i}{r}{v}}}\!\Co{\widehat{L}}(P) %\tag{I}
\label{tag:function1}
\end{align}

%\vspace{-1em}

We will prove that the function $\polynf[i,r,v]$, defined below, contains all the monomials of $\polynfpath[i,r,v]$.

For any $v\in V(H) \setminus \{s\}$,  $ i\in \{1, \ldots, k+1\}$, and $r\in \{0, \ldots, \min\{i, k^{\star} +1\}\}$ such that $r>i$ or $r\leq 1$, we define 
\begin{align} %\tag{II}
\label{tag:function2base}  
 \polynf[i,r,v] = \left\{ \begin{array}{c l}
1 & \mbox{if }  \langle s,v\rangle \text{ is an arc in $H$} \mbox{ and } i=r=1 \\
0 & \mbox{otherwise}  
\end{array}\right.
\end{align}

%\il{Why do we allow $r=0$ ? Since $i\geq 1$, we know we only deal with paths with 2 or more vertices. Thus, 1 or more unlabelled arcs.  }

\noindent
For any $i\in \{2, \ldots, k+1\}$, $r\in \{1, \ldots \min\{i, k^{\star} +1\}\}$, and  $v\in V(H) \setminus \{s\}$,
\begin{align}
 \polynf[i,r, v] &= \sum_{\substack{w \,\in N^{-}(v),\\ \win{w}\ne p}} %\hspace{-1cm}
\polynf[i-1,r,w]\times \left(\sum_{j=1}^{k^{\star}-1}\eLabel{\win{w}, j}\right) + 
\sum_{\substack{w \in N^{-}(v),\\\win{w}=p} } \polynf[i-1,r-1,w] 
\label{tag:function2}  
\end{align}
%\nolinenumber

%
%\begin{equation} \label{eq:base-detpath}
%\setfamily{\widehat{Q}}{1}{1}{v}= \setfamily{{Q}}{1}{1}{v} = \left\{ \begin{array}{c l}
%\{\emptyset \} & \mbox{if }  (s,v) \in A(H) \\
%\emptyset & \mbox{otherwise}  
%\end{array}\right.
%\end{equation}

A {\it multilinear monomial} in a multivariate polynomial is defined to be a monomial in which every variable has degree at most one. The next result establishes the connection between \polynf\ and \polynfpath.

\begin{lemma}
\label{lem:function-equivalence}
%\Cref{tag:function2} correctly computes the relation defined in \cref{tag:function1}.
%\il{Instead show the following}

For each value of $i\in \{1, \ldots, k+1\}$, $r \in \{1, \ldots, \min\{i, k^{\star}+1\}\}$, and $v\in V(H) \setminus \{s\}$, we have the following properties:
each monomial in $\polynfpath[i,r,v]$ is also a monomial in $\polynf[i,r,v]$, and every multilinear monomial in $\polynf[i,r,v]$ is a monomial in $\polynfpath[i,r,v]$.

%\Ma{Moreover, every  multilinear monomial in $\polynf[i,r,v]$ corresponds to a path in $\setfamily{F}{i}{r}{v}$.}

%\il{Moreover, there is a multilinear monomial in $\polynf[i,r,v]$ if and only if there is a path in $\setfamily{F}{i}{r}{v}$. }

%\il{Consequently, we can conclude that there is a multilinear monomial in $\polynf[i,r,v]$ if and only if there is an $s$ to $t$ path of length $i+1$ with $i-r$ distinct arc labels.}

\end{lemma}

\begin{proof}
We will first prove that every monomial in $\polynf'[i,r,v]$ is also a monomial in $\polynf[i,r,v]$ for any  $i\in \{1, \ldots, k+1\}$, $r \in \{1, \ldots, \min\{i, k^{\star}+1\}\}$, and $v\in V(H)$. We will prove this by induction on the value of $i$.  
We say that the entry $[i,r,v]$ is {\it correct} if each monomial in $\polynfpath[i,r,v]$ is also a monomial in $\polynf[i,r,v]$.

The base case of the recursive definition ensures that the base case of the induction holds as well. Suppose that for some value of $i' \in \{1, \ldots, k\}$ the induction hypothesis holds for all entries $[i,r,v]$ where $i \in \{1, \ldots, i'\}$, $r \in \{1, \ldots, \min\{i, k^{\star}+1\}\}$, and $v\in V(H)$.

We want to show that entry $[i'+1, r, v]$ is correct for every value of $r\in \{1, \ldots, \min\{i'+1, k^{\star}+1\}\}$ and vertex $v\in V(H)$. To this end, consider an arbitrary path $P\in \setfamily{F}{i'+1}{r}{v}$, denoted by $P=(P_{0},e,v)$,  where $P_{0}$ denotes the prefix of the path $P$ that ends at the penultimate vertex, denoted by $\bar{w}$ 
and $e$ is an arc from $\bar{w}$ to $v$. Thus, based on whether the arc $e$ is labeled, we have two cases: $P_{0}\in \setfamily{F}{i'}{r}{\bar{w}}$ or $P_{0} \in \setfamily{F}{i'}{r-1}{\bar{w}}$, respectively. The two cases pinpoint the summand in \Cref{tag:function2} to which the monomial $\Co{\widehat{L}}(P)$ belongs.

\medskip
\noindent{\bf Case 1: Path $P_{0}\in \setfamily{F}{i'}{r}{\bar{w}}$}. In this case $e$ is labeled and $\win{\bar{w}}\neq p$. Since arc $e$ is labeled, path $P$ (has $i'+1$ arcs) can contain at most $i'$ unlabeled arcs, so $r\leq \min\{i', k^{\star}+1\}$. 

By induction hypothesis, we know that the entry $[i',r,\bar{w}]$ is correct. Thus, $\polynf[i',r,\bar{w}]$ contains the term $\Co{\widehat{L}}(P_{0})$, the monomial that corresponds to the labels on the arcs in the path $P_{0}$. The label of the arc $e$ in the path $P$ is $\eLabel{\win{\bar{w}}, j}$ where $j\in \{1, \ldots, k^{\star}-1\}$. Thus, monomial $\Co{\widehat{L}}(P)=\Co{\widehat{L}}(P_{0})\times \eLabel{\win{\bar{w}}, j}$; and this monomial is present in the summation $\sum_{\substack{w \,\in N^{-}(v),\\ \win{w}\ne p}} \polynf[i',r,w]\times \left(\sum_{j=1}^{k^{\star}-1}\eLabel{\win{w}, j}\right)$. Thus, the monomial $\Co{\widehat{L}}(P)$ is in $\polynf[i'+1,r,v]$.

\smallskip
\noindent{\bf Case 2: Path $P_{0}\in \setfamily{F}{i'}{r-1}{\bar{w}}$}. 
In this case $e$ is unlabeled and $\win{\bar{w}}= p$.
Since arc $e$ is unlabeled, we know that  $\Co{\widehat{L}}(P)=\Co{\widehat{L}}(P_{0})$. 
 By induction hypothesis, we know that the entry $[i',r-1,w]$ is correct, and so $\Co{\widehat{L}}(P_{0})$ is part of the summation $\!\sum_{\substack{w \in N^{-}(v),\\\win{w}=p} } \polynf[i',r-1,w]$. 
 %Since $\Co{\widehat{L}}(P)=\Co{\widehat{L}}(P_{0})$, so is $\Co{\widehat{L}}(P)$. 
 Thus, the monomial $\Co{\widehat{L}}(P)$ is in the polynomial $\polynf[i'+1,r,v]$.

 This completes the inductive argument and we can conclude that each of the entries $[i, r, v]$ is correct, i.e, every monomial in $\polynfpath[i,r,v]$ is also a monomial in $\polynf[i,r,v]$.

%\il{Thus, we can conclude that for every path $P\in \setfamily{F}{i}{r}{v}$, we have a monomial ($\Co{\widehat{L}}(P)$) in $\polynf[i,r,v]$. }\ma{Need this?}

%\il{Reworded the following para}
Next, we will argue that every multilinear monomial in $\polynf[i,r,v]$ is a monomial in $\polynfpath[i,r,v]$. We begin by noting that a simple induction on the value of $i$ and the fact that the summation enumerates over all the in-neighbors of $v$ (the arc may be labeled or unlabeled) yields the property that every monomial in $\polynf[i,r,v]$ corresponds to a path on $i+1$ vertices from $s$ to $v$ with $r$ unlabeled arcs (including the first), where the labels need not be distinct. In fact, a multilinear monomial in $\polynf[i,r,v]$, corresponds to a path, denoted by $P$, on $i+1$ vertices from $s$ to $v$ with $r$ unlabeled arcs and $i-r$ {\it distinctly labeled} arcs. Therefore, the path $P \in \setfamily{F}{i}{r}{v}$, and so $\Co{L}(P)$ is a monomial in $\polynfpath[i,r,v]$. This completes the proof. 
\end{proof}

Next, we establish a correspondence between the multilinear monomials and distinctly labeled paths. We infer the following result due to \Cref{lem:function-equivalence} and the fact that an existence of a path $P$ in $\setfamily{F}{i}{r}{v}$ implies that there is a monomial $\Co{\widehat{L}}(P)$ in $\polynfpath[i,r,v]$. 

\begin{corollary}\label{multilinear-term-is-path}For each $ i\in \{1, \ldots, k+1\}$, $r\in \{1, \ldots, \min\{i, k^{\star}+1\}\}$, and $v\in V(H)\setminus \{s\}$, we can conclude that there is a multilinear monomial in $\polynf[i,r,v]$ if and only if $\setfamily{F}{i}{r}{v} \neq \emptyset$.

\end{corollary}

\begin{proof} By \Cref{lem:function-equivalence}, we know that every multilinear monomial in $\polynf[i,r,v]$ is a monomial in $\polynfpath[i,r,v]$. Thus, if there is a multilinear monomial in $\polynf[i,r,v]$, then $\setfamily{F}{i}{r}{v}\neq \emptyset$. Conversely, if there exists a path $P$ in \setfamily{F}{i}{r}{v}, then there is a monomial $\Co{\widehat{L}}(P)$ in $\polynfpath[i,r,v]$, and by  \Cref{lem:function-equivalence}, $\Co{\widehat{L}}(P)$ exists in $\polynf[i,r,v]$.
\end{proof}

%\todo[inline]{we may move above proof to appendix and say here that due to \Cref{lem:function-equivalence} and the fact that   if there exists a path $P$ in \setfamily{F}{i}{r}{v}, then there is a monomial $\Co{\widehat{L}}(P)$ in $\polynfpath[i,r,v]$, we have the following: (DONE)}

%Additionally, we note that since the summation enumerates over all in-neighbors of $v$ (denoted by $w$, labeled or unlabeled), each summand in the formula yields a path of length $i+1$ from $s$ to $v$. Specifically, a summand which is a multilinear monomial yields a path of length $i+1$ from $s$ to $v$ with $i-r$ distinct labels. 
%That is, it contains $i-r$ distinct variables, each of which corresponds to distinct labels on arcs in a path of length $i+1$ from $s$ to $v$. %Moreover, for every path $P\in \setfamily{F}{i}{r}{v}$, we know that $\Co{\widehat{L}}(P)$ contains $i-r$ distinct variables, so $\Co{\widehat{L}}(P)$ is a multilinear monomial in the polynomial $\polynf[i,r,v]$. 

Our randomized algorithm uses Corollary~\ref{multilinear-term-is-path} in the following manner: It tests if the polynomial $\polynf[k+1,k^{\star}+1, t]$ contains a 
 a multilinear monomial, as that would be a sufficient condition to conclude that \setfamily{F}{k+1}{k^{\star}+1}{t} is non-empty, i.e., there is an $s$ to $t$ path on $k+2$ vertices in which exactly $k-k^{\star}$ arcs have distinct labels. 
% 
%
%
%%Using an arithmetic circuit for the polynomial $\polynf[k+1,k^{\star}+1, t]$, it tests if it contains a multilinear monomial, as that would be a sufficient condition to conclude that \setfamily{F}{k+1}{k^{\star}+1}{t} is non-empty, i.e., there is an $s$ to $t$ path on $k+2$ vertices in which exactly $k-k^{\star}$ arcs have distinct labels. 
%
%
%
%
%% and other arcs are unlabeled. 
%
%\il{TO DO: Do we analyze the complexity to compute $\polynf[i,r,v]$, the table has size $\OO(kk^{\star}n)$ ?}
%
%we define $\Co{\widehat{L}}$ to be the labels that appear in the arcs that 
%  \begin{displaymath} \Phi(P)=\Pi_{(v_{ij},v_{jt}) \in E(P)}\phi_{P}(v_{ij} v_{jt}), \text{ where } \end{displaymath}
%
%
%
Towards this, we will construct an {\it arithmetic circuit} for the polynomial $\polynf[k+1,k^{\star}+1, t]$ and use a result by Williams~\cite{Williams09} to test if it has a multilinear monomial. We begin by formally defining an arithmetic circuit.

\begin{definition}
 An {\em arithmetic circuit} $C$ over a commutative ring $R$ is a simple labeled directed acyclic graph with 
its  internal nodes are labeled by $+$ or $\times$ and leaves (in-degree zero nodes) are labeled from 
$X\cup R$, where $X = \{x_1,\ldots,x_n\}$, a set of variables. There is a node of out-degree zero, 
called the root node or the output gate. 
\end{definition}

%We use the following result by Williams~\cite{Williams09} as a subroutine.  

\begin{proposition}\label{prop:Williams}{\rm \cite[Theorem~3.1]{Williams09}}
Let $P(x_{1},\ldots,x_{n})$ be a polynomial of degree at most $d$, represented by an arithmetic circuit of size $s(n)$ with $+$ gates (of unbounded fan-in), $\times$ gates (of fan-in two), and no scalar multiplications. There is a randomized algorithm that on every $P$ runs in $\OO^{\star}(2^{d}s(n))$ time, outputs
 ``yes'' with high probability (at least $2/3$) if there is a multilinear term in the sum-product expansion of $P$, and always outputs ``no'' if there is no multilinear term.
\end{proposition}

%\noindent{\bf Description of the circuit that represents $\polynf[k+1,k^{\star}+1, t]$:}
The number of variables in $\polynf[k+1,k^{\star}+1,t]$ is $(k^{\star}-1)(m-1)$. We prove that the size, denoted by  $s((k^{\star}-1)(m-1))$, of an arithmetic circuit that represents $\polynf[k+1,k^{\star}+1,t]$ is bounded by a  polynomial function in $n$ and $m$, and that it can be constructed in time $(n+m)^{\OO(1)}$.  
%in the number of the inputs, which in our case is $(k^{\star}-1)(m-1)$ variables. 
%Formally, we prove the following result. 

\begin{lemma}
\label{circuit-size}
For each $ i\in\{1, \ldots, k+1\}$, $r\in \{1, \ldots, \min\{i, k^{\star}+1\}\}$, and $v\in V(H) \setminus \{s\}$, the polynomial $\polynf[i,r,v]$ is represented by an arithmetic circuit whose size is bounded by ${\rm poly}(n+m)$ and can be constructed in that time. Moreover, for each $ i\in \{1, \ldots, k+1\}$, $r\in \{1, \ldots, \min\{i, k^{\star}+1\}\}$, 
the degree of each monomial in $\polynf[i,r,v]$ is $i-r$. 
\end{lemma}

\begin{proof}
The recursive formula given in \Cref{tag:function2} actually describes the circuit representing the function $\polynf[i,r,v]$. We can inductively describe the construction as follows.  

In the base case, the arithmetic circuits have just one (input) gate which is either $0$ or $1$ (See \Cref{tag:function2base}).
%we have the circuit representing polynomial $\polynf[0,1,s]$ whose input gate only consists of 1 variable. The set of circuits representing polynomials \{$\polynf[1,1,v_{1j}]: 1\leq j\in n\}$ also consist of 1 input variable because the set of arcs $\{(s, v_{1j}):1\leq j\in n\}$ are all unlabeled.  

Suppose that we have a family of polynomial sized arithmetic circuits representing each of the polynomials in $\Co{A}_{1}=\{\polynf[i-1,r,w] : w \in N^{-}(v), \win{w}\ne p\}$, the set of polynomials $\polynf[i-1,r,w]$ such that the arc $(w,v)$ in $H$ is labeled, as well as those in $\Co{A}_{2}=\{\polynf[i-1,r,w] : w \in N^{-}(v), \win{w}= p\}$, the set of polynomials $\polynf[i-1,r,w]$ such that the arc $(w,v)$ is unlabeled. Then, the circuit representing the polynomial $\polynf[i,r,v]$ can be constructed from the circuits representing the polynomials in $\Co{A}_{1}$ and $\Co{A}_{2}$ as follows: For each  circuit in $\Co{A}_{1}$, we take additional input variables $\{\eLabel{\win{w}, j}: 1\leq j\leq k^{\star}-1\}$ and use a $+$  gate to add these variables followed by a $\times$ gate that multiplies the obtained sum with the output of the $\polynf[i-1,r,w]$-circuit. After doing this for each element of $\Co{A}_{1}$, we use additional $+$ gates to add the outputs so obtained to the outputs of the circuits in $\Co{A}_{2}$. Since there are $\OO(k^{\star}m)$ arc labels in total, and the in-degree of any vertex in $H$ is $\OO(n^2)$ (where $n=\vert V(G)\vert$), the number of new gates added to construct a circuit for $\polynf[i,r,v]$ from the previously computed 
circuits for $i'=i-1$, is upper bounded by $\OO(k^{\star}m n^2)$. As the number of choices for $i,r$ and $v$ is bounded by $\OO(k^2n^2)$, the total number of gates we created is bounded by a polynomial function in $m+n$, since $k\leq n$. 
%, the size of the circuit for $\polynf[k+1,k^{\star}+1,t]$ is upper bounded by $\OO(k k^{\star}m n^2)$
Therefore, the circuit for $\polynf[i,r,v]$ has size polynomial in $n+m$ and is constructable in time $(n+m)^{\OO(1)}$.  

The second statement of the lemma can be proved using a straightforward induction on $i$ and the proof is omitted. 
\end{proof}

Hence, for the purpose of our algorithm, we may assume that we have a circuit that represents the polynomial $\polynf[k+1,k^{\star}+1,t]$. Our algorithm can be described in a snapshot as follows: it uses the circuit for $\polynf[k+1,k^{\star}+1,t]$ to decide if there exists a multilinear monomial in the polynomial, and returns an answer accordingly. %The details are deferred to Section~\ref{rad:finalproof} of Supplementary File.

%\longversion{For completeness the proof of \Cref{lemma:rand-algo} is deferred to Section~\ref{rad:finalproof} in the Appendix. } 

%Theorems~\ref{thm:dettrgm}, and \ref{lemma:rand-algo} implies \Cref{coro:twocand}. \todo{may cut after Section~\ref{rad:finalproof}. It'll save a line}

\hide{
\begin{supress}
\il{From here the rest appears in Appendix. I think we leave it here}
}

Our algorithm constructs the graph $H$ as described at the beginning of \Cref{sec:gm_path}, pp~\pageref{sec:gm_path}. Then, it constructs the arithmetic circuit that represents the polynomial $\polynf[k+1,k^{\star}+1,t]$, as described in \Cref{circuit-size}. Following that we apply the algorithm described in \Cref{prop:Williams}, to detect the existence of a multilinear term in $\polynf[k+1,k^{\star}+1,t]$, and return the answer accordingly.
Proof of \Cref{lemma:rand-algo} completes the analysis. Now we are ready to prove \Cref{lemma:rand-algo}. 

\begin{proof}[Proof of \Cref{lemma:rand-algo}] Let $\Co{I}=(G,\Co{C}, \{w_{v}\}_{v\in V},\eta, p, k^{\star}, k)$ denote the given instance of \targm. By \Cref{circuit-size}, we know that the size of the circuit representing polynomial $\polynf[k+1, k^{\star}+1,t]$ is polynomial, hence the $s(\cdot)$ function in \Cref{prop:Williams} is bounded by a polynomial in $(n+m)$. Thus, we can conclude that our algorithm runs in time $\OO^{\star}(2^{k-k^{\star}})$. Next, we will prove that our randomized algorithm has one-sided error. 

Suppose that $\Co{I}$ is a \yes-instance of \targm. We will prove that with high probability, our algorithm will return ``yes''.  Since $\Co{I}$ is a \yes-instance, by the definition of $\setfamily{F}{k+1}{k^{\star}+1}{t}$ and Lemma~\ref{lem:stpath}, we have that $\setfamily{F}{k+1}{k^{\star}+1}{t} \neq \emptyset$. Thus, by  Corollary~\ref{multilinear-term-is-path}, $\polynf[k+1, k^{\star}+1,t]$ contains a multilinear monomial. By   \Cref{circuit-size}, 
we know that the degree of each monomial in $\polynf[k+1, k^{\star}+1,t]$  is $k-k^{\star}$.  Thus, 
by \Cref{prop:Williams}, our algorithm outputs ``yes'' with probability   at least $1/3$. 
Suppose that our algorithm returns ``yes''. Then, by  Corollary~\ref{multilinear-term-is-path}, $\setfamily{F}{k+1}{k^{\star}+1}{t} \neq \emptyset$. 
By Lemma~\ref{lem:stpath} and the definition of $\setfamily{F}{k+1}{k^{\star}+1}{t}$, $\Co{I}$ is a \yes-instance of \targm. 
\end{proof}

\hide{\il{}
\end{supress}}

%\il{Add more if space available after adding reduction.}

%\il{Can we add a figure?} 

%\il{TO DO: Analyze $s(n)$ the size of the circuit, expressed in terms of the number of variables $\OO(k^{\star}m)$ ...DONE}

%Thus, \Cref{thm:ranfpt} is proved.

%\begin{theorem}There is a randomized algorithm \gm on paths can be solved in time $\OO^{\star}(2^{k-k^{\star}})$ on 

%\end{theorem}

% !TEX root = main.tex

%\il{Do we want to claim multiple winners tweak works for the rand algo also? If so, then move this to end of Path-algo section, maybe Section 4.}

\subsubsection{{\bf Multiple winners in a district}}\label{sec:multiwinner}
In the presence of multiple winners in a district, our above analysis can be modified to yield a result analogous to Theorem~\ref{thm:detfpt} without the application of a tie-breaking rule.
\begin{figure}
\definecolor{ffvvqq}{rgb}{1,0.3333333333333333,0}
\definecolor{qqwuqq}{rgb}{0,0.39215686274509803,0}
\definecolor{ududff}{rgb}{0.30196078431372547,0.30196078431372547,1}
\begin{center}
\includegraphics[scale=0.355]{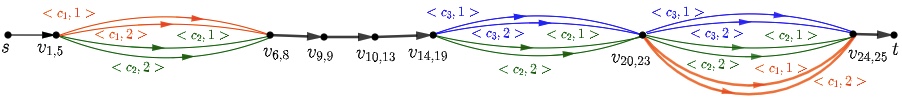}
\caption{Depicts $H'$ where $n=25$, $k=7$, $k^\star=3$, $\Co{C}=\{p,c_1,c_2,c_3\}$, $\chi(c_1) =$ {\color{ffvvqq}orange}, $\chi(c_2) =${\color{qqwuqq}green}, and $\chi(c_3) =${\color{ududff}blue}. Candidates ${\color{ffvvqq}c_1}$ and ${\color{qqwuqq}c_2}$ win in district $P_{1,5}$; ${\color{qqwuqq}c_2}$ and ${\color{ududff}c_3}$ win in $P_{14,19}$; ${\color{ffvvqq}c_1}$, ${\color{qqwuqq}c_2}$and ${\color{ududff}c_3}$ win in $P_{20,23}$.
Candidate $p$ wins in districts $P_{6,8}$, $P_{9,9}$,$P_{10,13}$, and $P_{24,25}$.}
\label{fig:multi-winner}
\end{center}
\end{figure}

Specifically, the auxiliary graph $H$ (\Cref{sec:gm_path}) will have labeled arcs that have different colors, where each color represents a single candidate who wins the district represented by the tail vertex of the arc. Formally stated, there exists a coloring function $\chi:\!\Co{C}\!\rightarrow \{1, \ldots, m\}$. Suppose that for some $\{i,j\} \sse \{1, \ldots, n\}$ there are at least two winners in district $P_{i,j}$ other than $p$. Then consider all the out-going arcs from vertex $v_{i,j}$ in $H$. For any $r\in \{j+1, \ldots, n\}$, every arc from  $v_{i,j}$ to $v_{j+1, r}$ must be labelled.  For a candidate $c\in \Co{C} \sm \{p\}$ such that $c$ wins the district $P_{i,j}$, every arc from $v_{i,j}$ to $v_{j+1,r}$ is colored by $\chi(c)$ and labeled as before $\eLabel{c,1},\ldots, \eLabel{c,k^{\star}-1}$. 
(Note that we do not store the information that $p$ also won this district.) 

The goal now is to decide if there exists a subgraph $H'$ on $k+2$ vertices (i.e a spanning subgraph of $H$) with possible parallel arcs, that contains a directed (spanning) path from $s$ to $t$ (i.e passing through every vertex in $V(H')$) such that the path has $k^{\star} +1$ unlabeled arcs and all other arcs must have distinct labels (as depicted in Figure~\ref{fig:multi-winner}). Moreover, the parallel arcs in $H'$ must be of the following form: for the vertex $ v_{i,j} \in V(H')$ and $r\in \{j+1, \ldots, n\}$, such that $ v_{i,j}$ has out-going labeled arcs to $v_{j+1,r}$ in $H$ with more than one color, we require that in $H'$, the set of edges leaving $v_{i,j}$. contain one arc of each color that leaves $v_{i,j}$ in $H$. 

By modifying the definition of the set families $\setfamily{F}{i}{r}{v}$ and $\setfamily{Q}{i}{r}{v}$, and accordingly the representative family $\setfamily{\widehat{Q}}{i}{r}{v}$, we obtain algorithms that can deal with multiple winners in a district.  The increase in the number of arcs in the graph $H$ and the possibly parallel arcs in the $s$ to $t$ path-like subgraph $H'$, leads to an increase in the size of the labels stored in $\setfamily{Q}{i}{r}{v}$, an increase by a factor of $\alpha$ where $\alpha$ is the maximum number of winners (besides $p$) in a district. Arguing as in the proof of \Cref{thm:dettrgm}, yields the property that the size of $\setfamily{\widehat{Q}}{i}{r}{v}$ is at most ${(k-k^{\star})\alpha \choose (k-k^{\star})\alpha - (i-r)\alpha} \leq 2^{k\alpha} $. 
We note that interpreting $\setfamily{Q}{i}{r}{v}$ as a multivariate polynomial yields a polynomial with at most $k^{\star}m\alpha$ variables and the degree of each monomial is at most $(i-r)\alpha$, where $\alpha$ is the maximum number of winners (excluding $p$) in any district. Thus, by modifying appropriately, we obtain a randomized algorithm with desired properties. Since $\alpha\leq m$, the time complexities are $2.619^{km} (n+m)^{\OO(1)}$ and $2^{km}(n+m)^{\OO(1)}$, for deterministic and randomized algorithms respectively.

% !TEX root = main.tex

\newcommand{\bin}[1]{\ensuremath{{\rm binary}(#1)}\xspace}

\section{FPT Algorithm For General Graphs}\label{sec:Exact Algorithm}

In this section we will prove \Cref{theorem:exact}. %We start with a very brief idea of the algorithm first. 
Towards the proof of Theorem~\ref{theorem:exact}, we use polynomial algebra that carefully keeps track of the number of districts won by each candidate so that nobody wins (if at all possible) more than $p$. 
Due to \Cref{lem:tgm-gm} it is sufficient to prove the following. 
\begin{theorem}\label{lem-exact-algo}
 There is an algorithm that given an instance $I$ of \targm and a tie-breaking rule $\eta$, runs in time $2^n |I|^{\OO(1)}$, and solves the instance $I$.
 \end{theorem}

Before we discuss our algorithm, we must introduce some notations and terminologies. The {\em characteristic vector} of a set $S\subseteq U$, denoted by $\chi(S)$, is an $|U|$-length vector whose $i^{\text {th}}$ bit is $1$ if $u_i \in S$, otherwise $0$.  Two binary strings $S_1, S_2 \in \{0,1\}^{n}$ are said to be disjoint if for each $i\in \{1,\ldots,n\}$, the $i^{th}$ bit of $S_{1}$ and $S_{2}$ are different.  The {\em Hamming weight} of a binary string $S$, denoted by $\Co{H}(S)$, is defined to be the number of $1$s in the string $S$. %A monomial $x^i$ is said to have Hamming weight $h$, if the binary representation of the value $i$, denoted by \bin{i}, has Hamming weight $h$. %\Ma{With a slight abuse the notation, we denote $ \Co{H}(i) = h$.}\ma{need this?}

\begin{observation}\label{obs:disjoint-binary-vectors}
Let $S_1$ and $S_2$ be two binary vectors, and let $S=S_1+S_2$. If $\Co{H}(S)=\Co{H}(S_1)+\Co{H}(S_2)$, then $S_1$ and $S_2$ are disjoint binary vectors. 
\end{observation}

\begin{proposition}{\rm \cite{DBLP:journals/tcs/CyganP10}}\label{prop:disjoint-set}
Let $S=S_{1} \cup S_{2}$, where $S_1$ and $S_2$ are two disjoint subsets of the set $V=\{v_{1}, \ldots, v_{n}\}$. Then, $\chi(S)=\chi(S_1)+\chi(S_2)$ and $\Co{H}(\chi(S))=\Co{H}(\chi(S_1))+\Co{H}(\chi(S_2))=|S_1|+|S_2|$.
\end{proposition}

A monomial $x^i$, where $i$ is a binary vector, is said to have Hamming weight $h$, if $i$ %the binary representation of the value $i$, denoted by \bin{i}, 
has Hamming weight $h$.
The {\em Hamming projection} of a polynomial $P(x)$ to $h$, denoted by $\Co{H}_h(P(x))$, is the sum of all the monomials of $P(x)$ which have Hamming weight $h$. We define the representative polynomial of $P(x)$, denoted by $\Co{R}(P(x))$, as the sum of all the monomials that have non-zero coefficient in $P(x)$ but have coefficient $1$ in $\Co{R}(P(x))$, i.e., it %ignores the actual coefficients and 
only remembers whether the coefficient is non-zero. We say that a polynomial $P(x)$ {\it contains the monomial} $x^i$ if the coefficient of $x^{i}$ is non-zero. In the zero polynomial, the coefficient of each monomial is $0$. 

%\ma{Need this?}
\medskip

\noindent{\bf Algorithm.} Given an instance $I=(G,\Co{C},\{w_v\}_{v \in V(G)},p,k, k^\star)$  of \targm,  
we proceed as follows. We assume that $k^\star \geq 1$, otherwise $k=1$ and it is a trivial instance.

For each candidate $c_i$ in \Co{C}, we construct a family $\Co{F}_{i}$ that contains all possible districts won by $c_i$. Due to the application of tie-breaking rule, we may assume that every district has a unique winner. Without loss of generality, let $c_1=p$, the distinguished candidate. Note that we want to find a family $\mathcal{S}$ of $k$ districts, that contains $k^\star$ elements of the family $\Co{F}_{1}$ and at most $k^\star -1$ elements from each of the other family $\Co{F}_{i}$, where $i>1$. Furthermore, the union of these districts gives $V(G)$ and any two districts in $\Co{S}$ are pairwise disjoint.  To find such $k$ districts, we use the method of polynomial multiplication appropriately. Due to Observation~\ref{obs:disjoint-binary-vectors} and Proposition~\ref{prop:disjoint-set}, we know that subsets $S_{1}$ and $S_{2}$ are disjoint if and only if the Hamming weight of the monomial $y^{\chi(S_1)+\chi(S_2)}$ is $|S_{1}|+|S_{2}|$. %\Ma{Moreover, $S_{1}\cup S_{2}$ is a disjoint union if $\Co{H}(\chi(S_{1}\cup S_{2})) =|S_{1}|+|S_{2}|$.} 

We use the following well-known result about polynomial multiplication.
\begin{proposition}{\rm \cite{moenck1976practical}}\label{prop:poly-multiplication}
There exists an algorithm that multiplies two polynomials of degree $d$ in $\OO(d \log d)$ time.
\end{proposition}
%\shortversion{\sout{Now, given the families $\Co{F}_i$, where $i\in \{1,\ldots,m\}$, and an integer $k^\star$, we proceed as follows.}}

For every $i\in \{1,\ldots,m\}$, $\ell \in \{1,\ldots,n\}$, if $\Co{F}_i$ has a set of size $\ell$, then we construct a polynomial $P_i^\ell(y) = \sum_{\substack{Y \in \Co{F}_i \\ \lvert Y \rvert = \ell}}y^{\chi(Y)}$.  Next, using polynomials $P_1^\ell(y)$, where $\ell \in \{1,\ldots,n\}$, we will create a sequence of polynomials $Q_{1,j}^s$, where $j \in \{1,\ldots,k^\star -1\}$, $s \in \{j+1,\ldots,n\}$, in the increasing order of $j$, such that every monomial in the polynomial $Q_{1,j}^s$ has Hamming weight $s$. For $j=1$, we construct $Q_{1,1}^s$ 
by summing all the polynomials obtained by multiplying $P_1^{s'}$ and $P_1^{s''}$, for all possible values of $s',s'' \in \{1,\ldots,n\}$ such that $s'+s''=s$, and then by taking the representative polynomial of its Hamming projection to $s$. Formally, we define \[Q_{1,1}^s= \Co{R}\Big(\Co{H}_s\Big(\sum_{\substack{1\leq s',s'' \leq s \\ s'+s'' =s}}P_1^{s'} \times P_1^{s''}\Big)\Big).\]

Thus, if $Q_{1,1}^s$ contains a monomial $x^{t}$, then there exists a set $S\sse V(G)$ of size $s$ such that $t=\chi(S)$ and $S$ is formed by the union of two districts won by $c_1$. Next, for $ j \in \{2,\ldots,k^\star -1\}$ and $s\in \{j+1,\ldots,n\}$, we create the polynomial $Q_{1,j}^s$ similarly, using $Q_{1, (j-1)}^{s''}$ in place of $P_{1}^{s''}$. Formally, %we define
 \[Q_{1,j}^s=\Co{R}\Big(\Co{H}_s\Big(\sum_{\substack{1\leq s',s'' \leq s \\ s'+s'' =s}}P_1^{s'} \times Q_{1,(j-1)}^{s''}\Big)\Big).\] Thus, if $Q_{1,j}^s$ contains a monomial $x^{t}$, then there exists a set $S\sse V(G)$ of size $s$ such that $t=\chi(S)$ and $S$ is formed by the union of $j+1$ districts won by $c_1$. In this manner, we can keep track of the number of districts won by $c_1$. Next, we will take account of the wins of the other candidates.

%\il{\Hi{To Do}: (1) Explain what polynomial $T_{k^{\star}+h}$ encodes....See old notes. }

Towards this we create a family of polynomials $\Co{T}=\{T_{k^\star},\ldots,T_{k}\}$ such that the polynomial $T_{k^\star+\ell}$, where $\ell \in \{0,\ldots,k-k^\star\}$, encodes the following information: the existence of a monomial $x^t$ in $T_{k^\star+\ell}$ implies that there is a subset $X\sse V(G)$ such that $t=\chi(X)$ and $X$ is the union of $k^{\star}+\ell$  districts in which $c_1$ wins in $k^{\star}$ districts and every other candidate wins in at most $k^{\star}-1$ districts. Therefore, it follows that if $T_{k}$ contains the monomial $y^{\chi(V(G))}$ (the all 1-vector) then our algorithm should return ``\yes'', otherwise it should return ``\no''. We define $T_{k^\star+\ell}$ recursively, with the base case given by $T_{k^{\star}}=\sum_{s=k^\star}^n Q_{1,(k^\star-1)}^s$. If $T_{k^\star}=0$, then we return ``\no''. 

We initialize $T_{k^\star+\ell}=0$, for each $\ell \in \{1,\ldots,k-k^\star\}$. For each $ i \in \{2,\ldots,m\}$, we proceed as follows in the increasing order of $i$.

\begin{tcolorbox}[colback=gray!5!white,colframe=gray!75!black]
 \begin{itemize}
 \item For each $j \in \{1,\ldots,\min\{k-1,k-k^\star\}\}$
 \begin{itemize}
 \item For each $\ell \in \{j,\ldots,k-k^\star\}$ and $s \in \{k^\star+1,\ldots,n\}$
% \begin{itemize}
%\item For each $s \in \{1,\ldots,n\}\setminus [k^\star]$
\begin{itemize}
 \item Computer the polynomial $Q_\ell^s = \sum\limits_{\substack{1\leq s',s''\leq s \\ s'+s'' =s}}P_i^{s'}\times \Co{H}_{s''}(T_{k^\star+\ell-1})$
 \item Compute the Hamming projection of $Q_{\ell}^s$ to $s$, that is,  $Q_{\ell}^s = \Co{H}_s(Q_{\ell}^s)$
 \end{itemize}
%   \end{itemize}
      \item For each $\ell \in \{j,\ldots,k-k^\star\}$
      \begin{itemize}
      \item Set $T_{k^\star+\ell}=\Co{R}(T_{k^\star+\ell}+\sum_{s=k^\star+1}^n Q_{\ell}^s)$
      \end{itemize}
 \end{itemize}
 \end{itemize}
\end{tcolorbox} 

The range of $j$ is dictated by the fact that since $c_1$ wins $k^{\star}$ districts, all other candidates combined can only win $k-k^{\star}$ districts and each individually may only win at most $k^{\star}-1$ districts. Thus, overall candidate $c_{i}$, for any $i\geq 2$ can win at most $\min\{k^{\star}-1, k-k^{\star}\}$ districts. The range of $\ell$ is dictated by the fact that (assuming that first $k^\star$ districts are won by $c_1$) $j^\text{th}$ district won by $c_i$ is either $(k^\star+j)^\text{th}$ district, or $(k^\star+j+1)^\text{th}$ district, ..., or $k^\text{th}$ district. The range of $s$ is dictated by the fact that the number of vertices in the union of all the districts is at least $k^\star+1$ as $c_1$ wins $k^\star$ districts.  

Note that $Q_{\ell}^s$ is a non-zero polynomial if there exists a subset of vertices of size $s$ that are formed by the union of $k^{\star}+\ell$ pairwise disjoint districts, $k^{\star}$ of which are won by $c_1$ and every other candidate wins at most $k^{\star}-1$. Thus, the recursive definition of $T_{k^\star+\ell}$ is self explanatory. 
% !TEX root = main.tex
%
%\hide{\section{Analysis of the Exact Algorithm}\label{app:exact-algo}}
%
%\hide{In this section, we give the correctness of the exact exponential-time algorithm for \gm presented in \Cref{sec:Exact Algorithm}.}
%
Next, we prove the correctness of the algorithm. In particular, we prove \Cref{lem-exact-algo}.
%
%an exact algorithm for \targm that runs in $2^n(n+m)^{\OO(1)}$ time, where $n$ and $m$ are the number of vertices in the given graph and the number of candidates, respectively. 
%
%Recall that for a polynomial $P(x) = \sum_{i=1}^\ell a_ix^i$, its Hamming projection to $h$, where $h\in [\ell]$, is defined as follows.
%\begin{align*}
% & \Co{H}_h(P(x)) = \sum_{i=1}^\ell b_ix^i  \quad \textrm{where} \quad b_i = \begin{cases} a_i \quad \textrm{if} \quad \Co{H}(i) = h \\
%0 \quad \textrm{otherwise}
%\end{cases}
%\end{align*} 
%
%Additionally, the representative polynomial of $P(x)$ is defined as follows.
%\begin{align*}
% & \Co{R}(P(x)) = \sum_{i=1}^\ell b_ix^i  \quad \textrm{where} \quad b_i = \begin{cases} 1 \quad \textrm{if} \quad a_ i \neq 0 \\
%0 \quad \textrm{otherwise}
%\end{cases}
%\end{align*}
%
%
% \begin{theorem}\label{lem-exact-algo}
% There is an algorithm that solves \targm in $2^n n^{\OO(1)}$ time.
%%Given an instance $(G,\Co{C},W=\{w_v\colon \Co{C}\rightarrow \mathbb{N} \colon v \in V(G)\},p,k,K^\star)$ of \targm, there exists an algorithm  is an algorithm that outputs $k$ districts such that the target candidate $p$ wins in $k^\star$ districts and all the other candidates win in at most $k^\star -1$ districts in $\OO^\star(2^n)$ time,  where $n$ is the number of vertices in the given graph, if it exists, otherwise returns \no.
%\end{theorem}
%
\begin{supress}

 \begin{theorem}\label{lem-exact-algo}
 There is an algorithm that given an instance $I$ of \targm and a tie-breaking rule $\eta$, runs in time $2^n |I|^{\OO(1)}$, and solves the instance $I$.\ma{App}
 \end{theorem}

Before we present the proof of this theorem, we wish to introduce some terminology and notations.

%\il{Move to exact algo section}

Let $U=\{u_1,\ldots,u_\ell\}$ be a set of $\ell$ elements. The {\em characteristic vector} of $S\subseteq U$, denoted by $\chi(S)$, is an $\ell$ length vector whose $i^{\text {th}}$ bit is $1$ if $u_i \in S$, otherwise $0$. We say that two binary strings $S_1$ and $S_2$ are disjoint if there is no $i$ such that $i^{\text {th}}$ bit of both $S_1$ and $S_2$ is $1$. The {\em Hamming weight} of a binary string $t$, denoted by $\Co{H}(t)$, is the number of $1$s in the the string $t$. 
\begin{observation}\label{obs:disjoint-binary-vectors}
Let $S_1$ and $S_2$ be two binary vectors, and let $S=S_1+S_2$. If $\Co{H}(S)=\Co{H}(S_1)+\Co{H}(S_2)$, then $S_1$ and $S_2$ are disjoint binary vectors.
\end{observation}
\begin{proposition}{\rm \cite{DBLP:journals/tcs/CyganP10}}\label{prop:disjoint-set}\ma{Main body} 
Let $S_1$ and $S_2$ be two disjoint subsets of $U=\{u_1,\ldots,u_\ell\}$, and let $S =S_1 \uplus S_2$. Then

\begin{itemize} 
\item $\chi(S)=\chi(S_1)+\chi(S_2)$
\item $\Co{H}(\chi(S))=\Co{H}(\chi(S_1))+\Co{H}(\chi(S_2))$, that is, $\Co{H}(\chi(S))=|S_1|+|S_2|$
\end{itemize}
\end{proposition}
A monomial $x^i$ is said to have Hamming weight $h$, if the binary representation of $i$ has Hamming weight $h$. Here, we slightly abuse the notation and write $ \Co{H}(i) = h$. For any $h \in [\ell]$, the {\em Hamming projection} of a polynomial $P(x)= \sum_{i=1}^\ell a_ix^i$ to $h$ is defined by the following polynomial.

\begin{align*}
 & \Co{H}_h(P(x)) = \sum_{i=1}^\ell b_ix^i  \quad \textrm{where} \quad b_i = \begin{cases} a_i \quad \textrm{if} \quad \Co{H}(i) = h \\
0 \quad \textrm{otherwise}
\end{cases}
\end{align*} 
That is, $\Co{H}_h(P(x))$ is the sum of all the monomials in $P$ which has Hamming weight $h$. % whose degree have Hamming weight $h$.
For a polynomial $P(x) = \sum_{i=1}^\ell a_ix^i$, we define its {\em representative polynomial} $\Co{R}(P)$ as follows.\ma{Appendix}
\begin{align*}
 & \Co{R}(P(x)) = \sum_{i=1}^\ell b_ix^i  \quad \textrm{where} \quad b_i = \begin{cases} 1 \quad \textrm{if} \quad a_ i \neq 0 \\
0 \quad \textrm{otherwise}
\end{cases}
\end{align*}
That is, the representative polynomial $\Co{R}(P(x))$ of a polynomial $P(x)$ ignores the actual coefficients and only remembers whether the coefficient of a monomial is non-zero.
We say that a polynomial $P(x) = \sum_{i=1}^\ell a_ix^i$ contains a monomial $x^i$ if $a_i \neq 0$. The zero polynomial is a polynomial in which the coefficient of every monomial is $0$. 
\par

\begin{supress}

%Using Lemma~\ref{lem-exact-algo}, and the fact that we invoke Lemma~\ref{lem-exact-algo} at most $k$ times, we obtained the following result. 
%
%\begin{theorem}\label{thm-exact-algo}
%\gm can be solved in $\OO^\star(2^n)$ time, where $n$ is the number of vertices in the given graph. 
%\end{theorem}
%
%Without loss of generality, in this section, we assume that $p=c_1$. 

%\Cref{theorem:exact} is a corollary of \Cref{lem-exact-algo}. 
%\todo[inline]{shall we replace above statement with following: Due to~\Cref{lem:tgm-gm} and~\Cref{theorem:exact}, we obtained~\Cref{lem-exact-algo}} 

\Ma{Due to~\Cref{lem:tgm-gm} and~\Cref{lem-exact-algo}, we obtain~\Cref{theorem:exact}.} In the following subsection, we describe the algorithm specified in Theorem~\ref{lem-exact-algo}. Then, we prove its correctness in Subsection~\ref{subsec:correctness} and show that the algorithm runs in specified time in Subsection~\ref{subsec:time}.

\begin{supress}

\subsection{Algorithm}\label{subsec:exact-algo}
\il{Most of this will be in main body, so we can omit.}

Given an instance $I=(G,\Co{C},\{w_v\colon \Co{C}\rightarrow \mathbb{Z}^+\}_{v \in V(G)},p,k, k^\star)$  of \targm,  
we proceed as follows. \Ma{We assume that $k^\star \geq 1$, otherwise $k=1$ and it is the trivial instance. Specifically,   it is a {\sf Yes}-instance only if $V(G)$ is the district won by the candidate $p$.}\ma{Reword}
% or it is a \no-instance of the problem.  
 For each candidate $c_i$ in \Co{C}, we construct a family $\Co{F}_{i}$, that contains all possible districts won by $c_i$. 
Notice that every district has a unique winner because of the tie-breaking rule $\eta$. 
 %Formally,  $\Co{F}_{i}=\{X\subseteq V(G) \colon \arg\max_{v\in X}w(v) =\{c_i\}\}$. 
Without loss of generality, let $c_1$ be the distinguished candidate $p$. Note that we want to find a family $\cal S$ of $k$ districts, that contains $k^\star$ elements of the family $\Co{F}_{1}$ and at most $k^\star -1$ elements from any other family $\Co{F}_{i}$, where $i \in \{2,\ldots, m\}$. Further, the union of these districts  is $V(G)$ and any two districts in $\cal S$ are pairwise disjoint.  To find such $k$ districts, we use the method of polynomial multiplication appropriately. Due to Observation~\ref{obs:disjoint-binary-vectors} and Proposition~\ref{prop:disjoint-set}, we have that % we will use the fact that if $S_1$ and $S_2$ are two disjoint sets, then $\chi(S_1\cup S_2)=\chi(S_1)+\chi(S_2)$, and vice versa. 
 $S_1$ and $S_2$ are two disjoint sets if and only if the Hamming weight of the monomial $y^{\chi(S_1)+\chi(S_2)}$ is $|S_{1}|+|S_{2}|$. 
 
 %Using this observation, given a family $\Co{F}_i$, for every candidate $c_i$, we construct polynomially many polynomials and use polynomial multiplication appropriately to find $k$ required districts, if it exists. %reduce our problem to polynomially many instances of polynomial multiplication. 
We use the following known result for polynomial multiplication.
\begin{proposition}{\rm \cite{moenck1976practical}}\label{prop:poly-multiplication}
There exists an algorithm that multiplies two polynomials of degree $d$ in $\OO(d \log d)$ time.
\end{proposition}
Now, given the families $\Co{F}_i$, where $i\in \{1,\ldots,m\}$, and an integer $k^\star$, we proceed as follows.

For every $i\in \{1,\ldots,m\}$, $\ell \in \{1,\ldots, n\}$, if $\Co{F}_i$ has a set of size $\ell$, then we construct a polynomial $P_i^\ell$ in $y$ as follows. 
\begin{align*}
 P_i^\ell(y) = \sum_{\substack{Y \in \Co{F}_i \\ \lvert Y \rvert = \ell}}y^{\chi(Y)}
\end{align*}

Next, using polynomials $P_1^\ell(y)$, where $\ell \in \{1,\ldots,n\}$, we will create a sequence of polynomials $Q_{1,j}^s$, where $j \in \{1,\ldots, k^\star -1\}$, $s \in \{j+1,\ldots,n\}$, in the increasing order of $j$, such that every term in the polynomial $Q_{1,j}^s$ has Hamming weight $s$. For $j=1$, we create polynomials as follows. %$Q_{1,1}^1 = \Co{P}_1^1$. 
%For $s\geq 2$, we construct the polynomials as follows.
%\todo[inline]{mention that $k^\star \geq 2$, otherwise $k=1$ which is a trivial case.}
 \begin{align*} 
  Q_{1,1}^s= \Co{R}\Big(\Co{H}_s\Big(\sum_{\substack{1\leq s',s'' \leq s \\ s'+s'' =s}}P_1^{s'} \times P_1^{s''}\Big)\Big)
  \end{align*}
  That is, we first construct a polynomial by summing all the polynomials obtained by multiplying polynomials $P_1^{s'}$ and $P_1^{s''}$, for all $s',s'' \in \{1,\ldots,n\}$ such that $s'+s''=s$, and then construct representative polynomial of its Hamming projection to $s$. %Note that $Q_{1,1}^s$ has a monomial of degree whose Hamming weight is $s$ if and only if we have monomials in $P_1^{s'}$ and $P_1^{s''}$ whose degrees are disjoint. % $s'$ and $s''$ are disjoint binary strings, that is, there is no $\ell \in \{1,\ldots,n\}$ such that $s'[\ell]=s''[\ell]=1$.  %We next compute the representative polynomial of $\Co{H}_s(Q_{1,1}^s)$ and denote it again by $Q_{1,1}^s$. %That is,
%  \begin{equation*}
%  Q_{1,1}^s = \Co{R}(\Co{H}_s(Q_{1,1}^s))
%  \end{equation*}
  %
  %We next compute representative polynomial of $Q_s$ and denote it again by $Q_s$. 
  Next, for $ j \in \{2,\ldots, k^\star -1\}$ and $s\in \{j+1,\ldots, n\}$, we create the polynomial $Q_{1,j}^s$ as follows. 
 \begin{align*}
  Q_{1,j}^s=\Co{R}\Big(\Co{H}_s\Big(\sum_{\substack{1\leq s',s'' \leq s \\ s'+s'' =s}}P_1^{s'} \times Q_{1,(j-1)}^{s''}\Big)\Big)
 \end{align*}
% We next compute the representative polynomial of $\Co{H}(Q_{1,j}^s)$ and denote it again by $Q_{1,j}^s$. That is,
%  \begin{equation*}
%  Q_{1,j}^s = \Co{R}(\Co{H}_s(Q_{1,j}^s))
%  \end{equation*}

 %We next compute representative polynomial of $Q_s^\ell$ and denote it again by $Q_s^\ell$.
 %We rename the polynomial $Q_s^{k^\star}$ as $R_s$.
% If $Q_{1,(k^\star-1)}^s$ does not have monomials % whose degree has Hamming weight $s$, 
% for any $k^\star \leq s\leq n$, then return \no. Note that a monomial in $Q_{1,k^\star-1}^s$ indicates that there $k^\star$ districts won by $c_1$. 
\par
  We next create a family of polynomials $\Co{T}=\{T_{k^\star},\ldots,T_{k}\}$ such that each monomial $x^q$ in the polynomial $T_{k^\star+h}$, where $h \in \{0, \ldots, k-k^\star\}$, %contains polynomials that 
 indicates that there are 
  $k^\star+h$ districts in which $c_1$ wins in $k^\star$ districts and any other candidate wins in at most $k^\star-1$ districts, and the characteristic vector of the set  of vertices used in these districts is $q$.  
  Therefore, if $T_{k}\neq 0$ and it contains a monomial $y^{\chi(V(G))}$, then we return \yes, otherwise \no. We next discuss the construction of the family $\Co{T}$.

We construct the set $T_{k^\star}$ as follows: 
\begin{align*}
  T_{k^\star}  =  \sum_{s=k^\star}^n Q_{1,(k^\star-1)}^s  %\{Q_{1,(k^\star-1)}^s\colon k^\star\leq s\leq n, Q_{1,(k^\star-1)}^s \textrm{ is a non-zero polynomial}\}
\end{align*}
If $T_{k^\star}=0$, then we return \no. %Note that every monomial in $R_{k^\star}$ corresponds  $k^\star$ districts in which 
% Note that $|R_{k^\star}|\leq n$ and every monomial in a polynomial in $R_{k^\star}$ has same degree. In every set $R_{j}$, $k^\star \leq j \leq k$, we denote the polynomial with degree whose Hamming weight is $s$ as $R_i^s$. 
Initially, let $T_{k^\star+h}=0$, for all $h \in \{1,\ldots,k-k^\star\}$. Next, we update polynomials $T_{k^\star+1},\ldots,T_k$. For every $ i \in \{2,\ldots, m\}$, we proceed as follows in the increasing order of $i$.  %each pair of integers $2\leq i \leq m$ and $1\leq j \leq \min\{k^\star-1,k-k^\star\}$ we proceed as follows in the lexicographic ordering of $(i,j)$.
% Otherwise, we proceed as follows. We create polynomials $T_{k^\star+1},\ldots,T_{k}$ as follows. 
\begin{tcolorbox}[colback=red!5!white,colframe=red!75!black]
 \begin{itemize}
 \item For each $j \in \{1,\ldots,\min\{k^\star-1,k-k^\star\}\}$
 \begin{itemize}

 \item For each $\ell \in \{1,\ldots, k-k^\star\}$ 
 \begin{itemize}
\item For each $s \in \{k^\star+1,\ldots, n\}$
\begin{itemize}
 \item Computer the polynomial $Q_\ell^s = \sum_{\substack{1\leq s',s''\leq s \\ s'+s'' =s}}P_i^{s'}\times \Co{H}_{s''}(T_{k^\star+\ell-1})$
 % 
% compute \begin{equation}Q_\ell^s = \begin{cases}\sum_{\substack{1\leq s',s''\leq s \\ s'+s'' =s}}P_i^{s'}\times \Co{H}_{s''}(T_{k^\star+\ell-1}) \quad \text{ if } \quad T_{k^\star+\ell-1} \neq 1 \\
% 1 \quad \text{otherwise}
% \end{cases}
%  \end{equation}
 \item Compute the Hamming projection of $Q_{\ell}^s$ to $s$, that is,  $Q_{\ell}^s = \Co{H}_s(Q_{\ell}^s)$
 \end{itemize}
   \item Set $T_{k^\star+\ell}=\Co{R}(T_{k^\star+\ell}+\sum_{s=k^\star+1}^n Q_{\ell}^s)$
   \end{itemize}
 \end{itemize}
 \end{itemize}
\end{tcolorbox} 

\end{supress}

\smallskip
\noindent{\bf Correctness.}
 %In this section, we prove the correctness of the algorithm presented in the previous section. 
 % 
  In the following lemma, we prove the completeness of the algorithm. 
  \begin{lemma}\label{lem:algo-complete}
  If $(G,\Co{C},\{w_v\colon \Co{C}\rightarrow \mathbb{Z}^+\}_{v \in V(G)},p,k, k^\star)$ is a \yes-instance of \targm\ under a tie-breaking rule, then 
  the algorithm in \Cref{sec:Exact Algorithm} returns ``\yes''. 
    \end{lemma}
  \begin{proof}
Suppose that $V_1,\ldots,V_k$ is a solution to $(G,\Co{C},\{w_v\colon \Co{C}\rightarrow \mathbb{Z}^+\}_{v \in V(G)},p,k, k^\star)$. Recall that we assumed that $p=c_1$. Let $\Co{V}_i \subseteq \{V_1,\ldots,V_k\}$ be the set of districts won by the candidate $c_i$. Due to the application of a tie-breaking rule, $\Co{V}_i$s are pairwise disjoint. %Let $V(\Co{V}_i)$ denote the set of vertices in the districts in the set $\Co{V}_i$. 
Without loss of generality, let $\Co{V}_1 = \{V_1,\ldots,V_{k^\star}\}$. We begin with the following claim that enables us to conclude that polynomial $T_k$ has monomial $y^{\chi(V(G))}$.
\begin{claim}\label{clm:algo-correct}
\begin{sloppypar}For each $i \in \{1,\ldots,m\}$, polynomial $T_{\sum_{|\Co{V}_1|+\ldots+|\Co{V}_i|}}$ contains the monomial $y^{\chi(\cup_{Y\in \Co{V}_1\cup \ldots \cup \Co{V}_i}Y)}$. \end{sloppypar}
\end{claim}

\begin{proof} The proof is by induction on $i$. 
\begin{description}[wide,labelindent=0pt]
\item[Base Case:] $i=1$. We first note that each of the districts $V_1,\ldots,V_{k^\star}$ belong to the family $\Co{F}_1$ as all these districts are won by $c_1$ uniquely. Clearly, every polynomial $P_1^{|V_i|}$, where $i\in \{1,\ldots,k^\star\}$, contains the monomial $y^{\chi(V_i)}$. Since $k^\star\geq 2$, %Recall that we compute $Q_{1,(k^\star-1)}$ iteratively. 
due to the construction of the polynomial $T_{k^\star}$, it is sufficient to prove that $Q_{1,(k^\star-1)}^{s}$ has the monomial $y^{\chi(V_1\cup\ldots\cup V_{k^\star})}$, for some $s \in \{k^\star,\ldots,n\}$. Towards this, we prove that for every $\tilde{\ell} \in \{2,\ldots,k^\star\}$,  $Q_{1,(\tilde{\ell}-1)}^{|V_1|+\ldots+|V_{\tilde{\ell}}|}$ has the monomial $y^{\chi(V_1\cup\ldots\cup V_{\tilde{\ell}})}$.  We prove it by induction on $\tilde{\ell}$.  Observe that these polynomials are computed in the algorithm because $\tilde{\ell}-1 \in \{1,\ldots,k^\star-1\}$ and $|V_1|+\ldots+|V_{\tilde{\ell}}| \in \{\tilde{\ell},\ldots,n\}$. 
%\hspace{4cm}
\begin{description}[labelindent=0.75cm]
\item[Base Case:] $\tilde{\ell} =2$. By definition, we  
%Clearly, we 
consider the multiplication of 
polynomials $P_1^{|V_1|}$ and $P_1^{|V_2|}$ 
to construct the polynomial $Q_{1,1}^{|V_1|+|V_2|}$. Since $V_1 \cap V_2 = \emptyset$, using Proposition~\ref{prop:disjoint-set}, we have that 
$\chi(V_1)+\chi(V_2) = \chi (V_1\cup V_2)$ and $\Co{H}(\chi (V_1\cup V_2))=|V_1|+|V_2|$. Thus, $Q_{1,1}^{|V_1|+|V_2|}$ has the monomial $y^{\chi(V_1\cup V_2)}$. 
\item[Induction Step:] Suppose that the claim is true for $\tilde{\ell} =h-1$. We next prove it for $\tilde{\ell} =h$. To construct the polynomial $Q_{1,h-1}^{|V_1|+\ldots + |V_{h}|}$, we consider the multiplication of polynomials $Q_{1,h -2}^{|V_1|+\ldots + |V_{h-1}|}$ and $P_1^{|V_h|}$. By inductive hypothesis, $Q_{1,h-2}^{|V_1|+\ldots + |V_{h-1}|}$ has the monomial $y^{\chi(V_1\cup\ldots \cup V_{h-1})}$. Since $V_i\cap V_j =\emptyset$, for all $i,j \in \{1,\ldots,h\}$, using the same arguments as above, the polynomial $Q_{1,h-1}^{|V_1|+\ldots + |V_{h}|}$ has the monomial $y^{\chi(V_1\cup \ldots \cup V_h)}$. 
\end{description}

\item[Induction Step:] Suppose that the claim is true for $i = h-1$. We next prove it for $i=h$. If $\Co{V}_h = \emptyset$, then using the inductive hypothesis, polynomial $T_{|\Co{V}_1|+\ldots+|\Co{V}_{h}|}$ has the monomial $y^{\chi(\cup_{Y\in \Co{V}_1\cup \ldots \cup \Co{V}_h}Y)}$. Next, we consider the case when $\Co{V}_h \neq \emptyset$. Without loss of generality, let $\Co{V}_1\cup \ldots \cup \Co{V}_{h-1} = \{V_1,\ldots,V_q\}$ and $\Co{V}_h=\{V_{q+1}, \ldots, V_{q+|\Co{V}_h|}\}$. %Note that it is sufficient to
To prove our claim, we prove that for $t\in \{1,\ldots,|\Co{V}_h|\}$, the polynomial $T_{q+t}$ has the monomial $y^{\chi(V_1\cup \ldots \cup V_{q+t})}$. We again use induction on $t$.
\begin{description}[labelindent=0.75cm]
\item[Base Case:] $t=1$. Note that $q<k$ because $\Co{V}_h \neq \emptyset$ and $q\geq k^\star$ because $|\Co{V}_1|=k^\star$. In the algorithm, for $i=h,j=1$, we compute $Q_\ell^s$, where $\ell=q-k^\star+1$ (as $1\leq q-k^\star+1 \leq k-k^\star$), and $s=|V_1|+\ldots+|V_{q+1}|$. For constructing polynomial $Q_{q-k^\star+1}^{|V_1|+\ldots+|V_{q+1}|}$, we consider the multiplication of the polynomials $P_h^{|V_{q+1}|}$ and $\Co{H}_{|V_1|+\ldots+|V_q|}(T_{q})$ in the algorithm. Using the inductive hypothesis, $T_q$ has the monomial $y^{\chi(V_1\cup \ldots \cup V_q)}$ and as argued above $\Co{H}(\chi(V_1\cup \ldots \cup V_q))=|V_1|+\ldots+|V_q|$. %Since  $V_h \cap V_{h'} = \emptyset$, for all $h,h' \in [q]$, the Hamming weight of monomial $y^{\chi(V_1\cup \ldots \cup V_q)}$ is $|V_1|+\ldots+|V_q|$. 
Since $V_{q+1}$ is disjoint from the set $V_1\cup \ldots \cup V_q$, using the same argument as above, $Q_{\ell}^{s}$ has the monomial $y^{\chi(V_1\cup \ldots \cup V_{q+1})}$, where $\ell=q-k^\star+1$ and $s=|V_1|+\ldots+|V_{q+1}|$. Since for  $\ell=q-k^\star+1$, $T_{k^\star+\ell} = T_{q+1}$, it follows that $T_{q+1}$ has a monomial $y^{\chi(V_1\cup \ldots \cup V_{q+1})}$.
\item[Induction Step:] Suppose that the claim is true for $t=t'-1$. We next prove it for $t=t'$. Note that  $t'\leq k^\star-1$ because $|\Co{V}_h| \leq k^\star-1$. Also, $t' \leq k-k^\star$ because $|\Co{V}_1|=k^\star$. Therefore, in the algorithm we compute polynomials for $i=h$ and $j=t'$. We also note that $q+t'\leq k$, hence, for values $i=h$ and $j=t'$, we compute the polynomial $Q_\ell^s$, where $\ell=q+t'-k^\star$ and $s=|V_1|+\ldots+|V_{q+t'}|$. For constructing the polynomial $Q_{q+t'-k^\star}^{|V_1|+\ldots+|V_{q+t'}|}$, we consider the multiplication of polynomials $P_h^{|V_{q+t'}|}$ and $\Co{H}_{|V_1|+\ldots+|V_{q+t'-1}|}(T_{q+t'-1})$ in the algorithm. Using inductive hypothesis, the polynomial $T_{q+t'-1}$ has the  monomial $y^{\chi(V_1\cup \ldots \cup V_{q+t'-1})}$ whose Hamming weight is  %and as argued above $\Co{H}(V_1\cup \ldots \cup V_{q+t'-1}))=
$|V_1|+\ldots+|V_{q+t'-1}|$ as argued above. Since the sets $V_{q+t'}$ and $V_1\cup \ldots \cup V_{q+t'-1}$ are disjoint, using the same arguments as above, %Therefore,
 the polynomial $Q_{q+t'-k^\star}^{|V_1|+\ldots+|V_{q+t'}|}$ has the monomial $y^{\chi(V_1\cup \ldots \cup V_{q+t'})}$. Hence, we can conclude that the polynomial $T_{q+t'}$ has the monomial $y^{\chi(V_1\cup \ldots \cup V_{q+t'})}$. \qedhere
\end{description}
\end{description} 
\end{proof}

%\il{}

Hence, we can conclude that the polynomial
%Due to Claim~\ref{clm:algo-correct}, 
$T_k$ contains the monomial $y^{\chi(V(G))}$. Hence, the algorithm returns \yes. 
\end{proof}
  In the next lemma, we prove the soundness of the algorithm.
%  \todo[inline]{change the statement for Target-GM}
  \begin{lemma}
  If the algorithm in \Cref{sec:Exact Algorithm} returns ``\yes'' for an instance $I=(G,\Co{C},\{w_v\colon \Co{C}\rightarrow \mathbb{Z}^+\}_{v \in V(G)},c_1,k, k^\star)$ for the tie-breaking rule $\eta$, 
  then $I$ is a \yes-instance of \targm\ under the tie-breaking rule $\eta$. 
  \end{lemma}

  \begin{proof}
We first prove the following claims. 
\begin{claim}
\label{clm:for Tk polynomial}
If $T_{k^\star}$ has a monomial $y^{S}$, then there are $k^\star$  pairwise disjoint districts $Y_1,\ldots,Y_{k^\star}$   such that $\chi(Y_1\cup \ldots \cup Y_{k^\star})=S$ and $c_1$ wins in all the districts. % characteristic vector of their union is $S$ and $p$ wins in
\end{claim}

\begin{proof}
To prove our claim, we prove that for every monomial $y^S$ in $Q_{1,j}^s$, where $j\in \{1,\ldots,k^\star-1\}$ and $s \in \{j+1,\ldots, n\}$, there exists $j+1$ pairwise disjoint districts won by $c_1$ such that the characteristic vector of their union is $S$. We prove it by induction on $j$. 
\begin{description}[wide=0pt]
\item[Base Case:] $j=1$. Note that to construct polynomial $Q_{1,1}^s$, where $s \in \{2,\ldots, n\}$, we consider the multiplication of polynomials $P_1^{s'}$ and $P_1^{s''}$ such that $s=s'+s''$. So, we have a monomial $y^{\chi(S')+\chi(S'')}$ in $Q_{1,1}^s$, where $S'$ and $S''$ are sets in the family $\Co{F}_1$ of size $s'$ and $s''$, respectively. Since $y^{\chi(S')+\chi(S'')}$ is a monomial in $Q_{1,1}^s$, % We finally  take Hamming projection of all the monomials in $Q_{1,1}^s$ to $s$. So, if 
the monomial $y^{\chi(S')+\chi(S'')}$ has Hamming weight $s=s'+s''$. Therefore, due to Observation~\ref{obs:disjoint-binary-vectors}, $\chi(S')$ and $\chi(S'')$ are disjoint binary vectors. This implies that $S'$ and $S''$ are two disjoint sets in the family $\Co{F}_1$ (thus, won by $c_1$ uniquely) and due to Proposition~\ref{prop:disjoint-set}, $\chi(S'\cup S'')=\chi(S')+\chi(S'')$.
\item[Induction Step:] Suppose that the claim is true for $j=j'-1$. We next prove it for $j=j'$.  Let $y^S$ be a monomial of Hamming weight $s$ in $Q_{1,j}^s$. To construct polynomial $Q_{1,j}^s$, we consider the multiplication of polynomials $P_1^{s'}$ and $Q_{1,j-1}^{s''}$ such that $s=s'+s''$. Therefore, $y^S=y^{S'}\times y^{S''}$, where $y^{S'}$ is a monomial of Hamming weight $s'$ in $P_1^{s'}$ and $y^{S''}$ is a monomial of Hamming weight $s''$ in $Q_{1,j-1}^{s''}$. % such that $s=s'+s''$. 
Since $y^S$ is a monomial in $Q_{1,j}^s$, by the construction of polynomials we have that $\Co{H}(S'+S'') = s$. Therefore, due to Observation~\ref{obs:disjoint-binary-vectors}, $S'$ and $S''$ are disjoint vectors. 
%
%for the monomial $y^S$ in the polynomial $Q_{1,j}^s$, there exist monomial $y^{S'}$  in $P_1^{s'}$ with Hamming weight $s'$ and $y^{S''}$ in $Q_{1,j-1}^{s''}$ with Hamming weight $s''$ such that $y^S=y^{S'+S''}$. %   we have a monomial $y^{S'+S''}$, where $y^{S'}$ is a monomial in $P_1^{s'}$ and $y^{S''}$ is a monomial in $Q_{1,j-1}^{s''}$. 
%Now, after taking Hamming projection to $s$, if  $y^{S'+S''}$ is a monomial in $Q_{1,j}^s$, then this implies that $S'$ and $S''$ are disjoint characteristic vectors. %set in the family $\Co{F}_1$ of size $s'$ and $y^{\chi(S'')}$ is a monomial in $Q_{1,j-1}^{s''}$. 
Using inductive hypothesis, there exists $j-1$ pairwise disjoint districts, say $Y_1,\ldots,Y_{j-1}$, won by $c_1$ uniquely such that $\chi(Y_1\cup \ldots \cup Y_{j-1})=S''$. 
%By the construction of polynomials $P_1^{s'}$, there exist a set, say $Y$, in $\Co{F}_1$ of size $s'$. 
Let $Y$ be the set such that $\chi(Y)=S'$.  
Since $S'$ and $S''$ are disjoint characteristic vectors, $Y$ is disjoint from $Y_1\cup \ldots \cup Y_{j-1}$. Thus, $Y_1,\ldots,Y_{j-1}, Y$ are $j$ pairwise disjoint districts won by $c_1$ uniquely. Due to Proposition~\ref{prop:disjoint-set}, $\chi(Y_1\cup \ldots \cup Y_{j-1}\cup Y)=\chi(Y_1\cup \ldots \cup Y_{j-1})+\chi(Y)=S'+S''=S$. %the characteristic vector of their union is $S'+S''$. 
%\item 
\end{description}
Recall that $T_{k^\star}=\sum_{s=k^\star}^nQ_{1,(k^\star-1)}^s$. Therefore, if $T_{k^\star}$ has a monomial $y^{S}$, then there are $k^\star$ pairwise disjoint districts won by $c_1$ and the characteristic vector of the union of the districts is $S$.
\end{proof}

\begin{claim}\label{clm:allother-win-in-k-1}
For a pair of integer $i,j$, where $ i \in \{2,\ldots,m\}$ and $ j \in \{1,\ldots,\min\{k-k^\star,k^\star-1\}\}$, let $T_{k^\star+1},\ldots,T_{k}$ be the family of polynomials constructed in the algorithm at the end of for loops for  $i$ and $j$,  in \Cref{sec:Exact Algorithm}. Let $y^S$ be a monomial in $T_t$, where $t \in \{k^\star+1,\ldots,k\}$. Then, the following hold: 
\begin{itemize}
\item  there are $t$ pairwise disjoint districts $Y_1,\ldots,Y_t$ such that $\chi(Y_1\cup \ldots \cup Y_t)= S$
\item $c_1$ wins in $k^\star$ districts in $\{Y_1,\ldots,Y_t\}$
\item $c_i$ wins in $j$ districts in $\{Y_1,\ldots,Y_t\}$
\item for $2\leq q<i$, $c_q$ wins in at most $k^\star-1$ districts in $\{Y_1,\ldots,Y_t\}$
\item for $q>i$, $c_q$ does not win in any district in $\{Y_1,\ldots,Y_t\}$
\end{itemize} 
\end{claim}
%The proof of this claim follows by using nested induction on $i$ and $j$.

\begin{proof} %\ma{to Appendix?}
We prove it by induction on $i$.
\begin{description}[wide=0pt]
\item[Base Case:] $i=2$. In this case we prove that the claim is true for $i=2$ and every $ j \in \{1,\ldots,\min\{k-k^\star,k^\star-1\}\}$. We again use induction on $j$. 
\begin{description}[labelindent=0.75cm]
\item[Base Case:] $j=1$. We first observe that for $i=2$ and $j=1$, $T_{k^\star+2},\ldots,T_k$ are zero polynomials. Let $y^S$ be a monomial in $T_{k^\star +1}$ with Hamming weight $s$. By the construction of monomials, there exists monomial $y^{S'}$ in $T_{k^\star}$ whose Hamming weight is $s'$ and $y^{S''}$ in $P_2^{\Co{H}(S'')}$ whose hamming weight is $s''$ such that $s=s'+s''$ and $y^S=y^{S'}\times y^{S''}$. %, and $\Co{H}(\chi(S')+\chi(S''))=s$. 
Since $\Co{H}(S'+S'')=s=s'+s''$ and $\Co{H}(S')+\Co{H}(S'')=s'+s''$, due to Observation~\ref{obs:disjoint-binary-vectors}, %
%
%Since $\Co{H}(\chi(S')+\chi(S''))=s=s'+s''$, we have that the characteristic 
vectors $S'$ and $S''$ are disjoint. Due to Claim~\ref{clm:for Tk polynomial}, there exists $k^\star$ pairwise disjoint districts won by $c_1$ uniquely, say $Y_1,\ldots,Y_{k^\star}$, such that $\chi(Y_1\cup \ldots \cup Y_{k^\star})=S'$. Since $y^{S''}$ is a monomial in $P_2^{s''}$, there exists a district in $\Co{F}_2$ of size $s''$ won by $c_2$ uniquely, say $Y$, such that $\chi(Y)=S''$. Since $S'$ and $S''$ are disjoint characteristic vectors, sets $Y$ and $Y_1\cup \ldots \cup Y_{k^\star}$ are disjoint.  Therefore, due to Proposition~\ref{prop:disjoint-set}, $\chi(Y_1\cup \ldots \cup Y_{k^\star}\cup Y)=\chi(Y_1\cup \ldots \cup Y_{k^\star})+\chi(Y)=S'+S''=S$. Therefore, $Y_1,\ldots,Y_{k^\star},Y$ are $k^\star +1$ districts satisfying all the properties mentioned in the claim. 
\item[Induction Step:] Suppose that the claim is true for $j=j'-1$. Next, we prove the claim for $j=j'$. Note that in the algorithm, we add monomials in $T_t$, where $t \in \{k^\star+1, \ldots, k\}$, constructed in the previous iteration. Let $\tilde{T}_{k^\star+1},\ldots,\tilde{T}_k$ be the polynomials constructed for $i=2$, $j=j'-1$. Let $T_t$, where $t \in \{k^\star+1, \ldots, k\}$, has a monomial $y^S$ with Hamming weight $s$ which is not in $\tilde{T}_t$. By the construction of monomials $y^S=y^{S'}\times y^{S''}$, where $y^{S'}$ is a monomial in $\tilde{T}_{t-1}$ with Hamming weight $s'$ and $y^{S''}$ is a monomial in $P_2^{\Co{H}(S'')}$ with Hamming weight $s''$ such that $s=s'+s''$. % and  $\Co{H}(S'+S'')=s$. 
As argued above  the characteristic vectors $S'$ and $S''$ are disjoint. Using inductive hypothesis, there exists $t-1$ pairwise disjoint districts, say $Y_1,\ldots,Y_{t-1}$, such that $\chi(Y_1\cup \ldots \cup Y_{t-1})=S'$, $c_1$ wins in $k^\star$ districts uniquely, $c_2$ wins in $j'-1$ districts uniquely, and all the other candidates do not win any district in $Y_1,\ldots,Y_{t-1}$.  Since $y^{S''}$ is a monomial in $P_2^{s''}$, there exists a district in $\Co{F}_2$ of $s''$ won by $c_2$ uniquely, say $Y$, such that $\chi(Y)=S''$. Since $S'$ and $S''$ are disjoint characteristic vectors, sets $Y$ and $Y_1\cup \ldots \cup Y_{t-1}$ are disjoint. Therefore, due to Proposition~\ref{prop:disjoint-set}, $\chi(Y_1\cup \ldots \cup Y_{t-1}\cup Y)=\chi(Y_1\cup \ldots \cup Y_{t-1})+\chi(Y)=S'+S''=S$. Thus, $Y_1,\ldots,Y_{t-1},Y$ are $t$ districts satisfying all the properties mentioned in the claim. 
\end{description}
\item[Induction Step:] Suppose that the claim is true for $i=i'-1$. We next prove it for $i=i'$. Again we need to prove the claim for $i=i'$ and every $j \in \{1,\ldots,\min\{k-k^\star,k^\star-1\}\}$. The proof is same as above using induction on $j$. We give here the proof for completeness. %We again use induction on $j$ and the proof uses the same arguments as above.
\begin{description}[labelindent=0.75cm]
\item[Base Case:] $j=1$. Let $\tilde{T}_{k^\star+1},\ldots,\tilde{T}_k$ be the polynomials constructed for $i=i'-1$, $j = \min\{k-k^\star,k^\star-1\}$.  Let $T_t$, where $t \in \{k^\star+1, \ldots, k\}$, has a monomial $y^S$ with Hamming weight $s$ which is not in $\tilde{T}_t$. As argued above $y^S=y^{S'}\times y^{S''}$, where $y^{S'}$ is a monomial in $\tilde{T}_{t-1}$ with Hamming weight $s'$ and $y^{S''}$ is a monomial in $P_{i'}^{\Co{H}(S'')}$ with Hamming weight $s''$ such that $s=s'+s''$. % and  $\Co{H}(S'+S'')=s$. 
As argued above vectors $S'$ and $S''$ are disjoint. Using inductive hypothesis, there exists $t-1$ pairwise disjoint districts, say $Y_1,\ldots,Y_{t-1}$, such that $\chi(Y_1\cup \ldots \cup Y_{t-1})=S'$, $c_1$ wins in $k^\star$ districts uniquely, for $q \in \{2,\ldots, i'-1\}$, $c_q$ wins in at most $k^\star-1$ districts uniquely, and all the other candidates do not win in any district in $Y_1,\ldots,Y_{t-1}$.  Since $y^{S''}$ is a monomial in $P_{i'}^{s''}$, there exists a district in $\Co{F}_{i'}$ of size $s''$ won by $c_{i'}$ uniquely, say $Y$, such that $\chi(Y)=S''$. Since $S'$ and $S''$ are disjoint characteristic vectors, sets $Y$ and $Y_1\cup \ldots \cup Y_{t-1}$ are disjoint. Therefore, as argued above, due to Proposition~\ref{prop:disjoint-set}, $\chi(Y_1\cup \ldots \cup Y_{t-1}\cup Y)=S$. Thus, $Y_1,\ldots,Y_{t-1},Y$ are $t$ districts satisfying all the properties mentioned in the claim. 
\item[Induction Step:] Suppose that the claim is true for $j=j'-1$. Next, we prove it for $j=j'$. Let $\tilde{T}_{k^\star+1},\ldots,\tilde{T}_k$ be the polynomials constructed for $i=i'$, $j=j'-1$. Let $T_t$, where $ t \in \{k^\star+1,\ldots, k\}$, has a monomial $y^S$ with Hamming weight $s$ which is not in $\tilde{T}_t$. As argued above $y^S=y^{S'}\times y^{S''}$, where $y^{S'}$ is a monomial in $\tilde{T}_{t-1}$ with Hamming weight $s'$ and $y^{S''}$ is a monomial in $P_{i'}^{\Co{H}(S'')}$ with Hamming weight $s''$ such that $s=s'+s''$. As argued above vectors $S'$ and $S''$ are disjoint. Using inductive hypothesis, there exists $t-1$ pairwise disjoint districts, say $Y_1,\ldots,Y_{t-1}$, such that $\chi(Y_1\cup \ldots \cup Y_{t-1})=S'$, $c_1$ wins in $k^\star$ districts uniquely, for $q \in \{2,\ldots, i'-1\}$, $c_q$ wins in at most $k^\star-1$ districts, $c_{i'}$ wins in $j'-1$ districts uniquely, and all the other candidates do not win any district in $Y_1,\ldots,Y_{t-1}$.  Since $y^{S''}$ is a monomial in $P_{i'}^{s''}$, there exists a district in $\Co{F}_{i'}$ of size $s''$ won by $c_{i'}$ uniquely, say $Y$, such that $\chi(Y)=S''$. Since $S'$ and $S''$ are disjoint characteristic vectors, sets $Y$ and $Y_1\cup \ldots \cup Y_{t-1}$ are disjoint. Therefore, as argued above, $\chi(Y_1\cup \ldots \cup Y_{t-1}\cup Y)=S$. Thus, $Y_1,\ldots,Y_{t-1},Y$ are $t$ districts satisfying all the properties mentioned in the claim. \qedhere
\end{description}
\end{description}
\end{proof}

If the algorithm returns \yes, then we know that there is a monomial $y^{\chi(V(G))}$ in $T_k$. Therefore, due to Claim~\ref{clm:allother-win-in-k-1}, there are $k$ districts such that $c_1$ wins in $k^\star$ districts and all the candidates win in at most $k^\star-1$ districts. %\todo{need to prove for other candidates.}
  \end{proof}
    %\subsection{Running Time Analysis}\label{subsec:time}
 Next, we will analyze the running time of the algorithm. %in \Cref{sec:exact-algo}.
    \begin{lemma}
    The algorithm in \Cref{sec:Exact Algorithm} runs in $2^n (n+m)^{\OO(1)}$ time. 
    \end{lemma}
    \begin{proof}
    In the algorithm, we first construct a family of districts, $\Co{F}_i$ for each candidate $c_i \in \Co{C}$. Since we check all possible sets of $V(G)$ for the potential candidate of the family $\Co{F}_i$, this step of the algorithm takes $\OO(2^n(n+m))$ time.  Then, for each family, we construct at most $n$ polynomials. Since $|\Co{F}_i|\leq 2^n$, the number of terms in every polynomial is $\OO(2^n)$. Thus, these polynomials can be constructed in $\OO(2^n (n+m))$ time. Then, for every pair of integers, $i,j$, where $i \in \{2,\ldots, m\}$, $j\in \{1,\ldots,\min\{k^\star-1,k-k^\star\}\}$, we multiply polynomials at most $nk$ times. Since every polynomial has degree at most $2^n$, using Proposition~\ref{prop:poly-multiplication}, every polynomial multiplication takes $\OO(2^n n)$ time. Hence, the algorithm runs in  $2^n (n+m)^{\OO(1)}$ time. 
    \end{proof}
    
Hence, \Cref{lem-exact-algo} is proved.     
 % !TEX root = main.tex
 
 \section{In Conclusion}\label{sec:conclusion}
 We have shown that \gmpvr on paths is \npc, thereby resolving an open question in \cite{ItoKKO19}. This gives parameterized intractability for parameters such as maximum degree of a vertex in the graph. Furthermore, we have presented \FPT\ algorithms for paths when parameterized by the number of districts. We also give an \FPT\ algorithm running in time $2^n (n+m)^{\OO(1)}$ on general graphs. 
  
%\noindent{\bf Future Outlook.} 
We conclude with a few directions for further research: \begin{enumerate*}[label=(\roman*)] \item Does there exist a $\OO(c^n)$ algorithm for \gm when there are possibly multiple winners in a district?; \item Is \gm on paths \FPT\ parameterized by the number of candidates?; \item Is \gm on trees \FPT\ parameterized by the number of districts?
\end{enumerate*}

\newpage

%% Bibliography
%%

%% Please use bibtex, 
\bibliographystyle{plain}
\bibliography{references}

\end{document}